%% file: apssamp.tex
\documentclass[reprint,onecolumn,12pt,aps,prx,notitlepage,longbibliography,amsmath,amssymb,nofootinbib]{revtex4-2}

\usepackage{graphicx}
\usepackage{bm}
\usepackage{booktabs}
\usepackage{xcolor}
\usepackage{hyperref}
\usepackage{mathtools}
\hypersetup{colorlinks=true,linkcolor=blue,citecolor=blue,urlcolor=blue}

\newcommand{\Tr}{\operatorname{Tr}}
\newcommand{\E}{\mathbb E}
\newcommand{\id}{\mathbb I}
\newcommand{\cPT}{\Upsilon}

\newtheorem{theorem}{Theorem}
\newtheorem{proposition}{Proposition}
\newtheorem{corollary}{Corollary}

\makeatletter
\newenvironment{fixedtable}
  {\par\medskip\noindent\begin{minipage}{\linewidth}\centering}
  {\end{minipage}\par\medskip}
\newcommand{\fixedtablecaption}[2]{%
  \refstepcounter{table}%
  \@makecaption{\tablename~\thetable}{#1}%
  \label{#2}%
}
\makeatother

\makeatletter
\AtBeginDocument{\immediate\write\@auxout{\string\citation{REVTEX42Control}}}
\makeatother

\begin{document}

\title{Continuous-Process Randomized Compilation for Quantum Process Tensors}
\author{He Wang}
\email{dr\_hewang19@163.com}

\begin{abstract}
Traditional discrete randomized compiling (RC) shapes errors by selecting random control frames at logical-gate or cycle boundaries. However, a boundary-only description does not resolve the control evolution during a finite-duration instrument or the system--environment coupling that persists throughout the operation. When the environment retains memory, error control concerns multi-time deviations of the complete experimental record rather than an isolated single-gate channel. Here we establish continuous-process randomized compilation (CPRC) within the continuous process tensor (cPT) framework. We place random unitary control-frame trajectories, pointwise-covariant finite-duration instruments, and persistent environmental dynamics on a common physical time axis. We study four control protocols: fixed Pauli refresh, bounded Cayley paths, unitary Brownian motion, and Ornstein--Uhlenbeck (OU) driving, with Poisson refresh treated as an event-time extension.
We obtain three principal results. First, pointwise covariance ensures that every control trajectory in the absence of system--environment coupling exactly reproduces the target logical experiment. A closed control-frame path shows that endpoint compensation alone can introduce first-order errors in the instrument action. Second, independent Pauli conjugations applied over each complete interval diagonalize error histories at designated boundaries, whereas an exactly solvable static-bath model shows that the resulting classical labels can retain environment-induced correlations. Third, for bounded centered coupling and controls satisfying the required mixing conditions, we derive a total-variation error bound for complete adaptive output records. We further establish Fock-norm convergence of the entire trajectory-averaged, trace-closed cPT on a fixed finite time interval for finite-dimensional non-Gaussian environments and Gaussian baths satisfying the stated physical-cutoff realization, time-integrable covariance, and cutoff-stability conditions, including the uncut Drude spectrum with divergent equal-time correlations. We evaluate finite-duration output probabilities and operator conditional mutual information (CMI) in a qubit-environment model, and present numerical verification of all operation-label coefficients at a prescribed pair of insertion times in the two-particle sector, for both a qubit environment and a thermal Drude bath.
This work extends RC from discrete gate-level settings to continuous-time quantum processes with memory. It distinguishes logical transparency, interval classicalization, and operational decoupling in memory-bearing quantum processes, providing a unified theoretical framework and quantitative tools for error suppression in continuous-time quantum control, as well as control protocols and verification benchmarks for quantum error tailoring in non-Markovian environments.
\end{abstract}

\maketitle

\section{Introduction}
Error suppression in quantum control is a central challenge for scalable quantum computing. Randomized compiling (RC) is a practical error-tailoring technique that inserts random unitary operations at logical-gate or cycle boundaries and absorbs their inverses into neighboring operations. Under suitable conditions, this procedure tailors gate-dependent coherent noise into a stochastic noise model while preserving the logical function of the ideal circuit \cite{wallman2016,hashim2021}. Studies of coherent-noise characterization and experimental noise tailoring support this approach \cite{wallman2015,hashim2021}. Related mitigation methods and reviews place RC within a broader error-reduction framework \cite{temme2017,endo2021,cai2023,perrin2024,kurita2023}.

In conventional discrete-circuit settings, Pauli twirling with the appropriate frame updates yields a Pauli-diagonal error channel for memoryless gates \cite{wallman2016}. When the system interacts with an environment that retains memory, however, the relevant errors can be correlated across multiple times throughout the control sequence \cite{combref004,combref005,combref078,prakash2024}. A description in terms of independent single-gate channels is then insufficient. A further limitation arises when operations are represented only at their endpoints. Such a description does not specify how a random control frame acts during a finite-duration instrument. Within its active interval, the environment evolves and retains historical information, while the changing control frame modulates the system--environment interaction. Endpoint compensation alone does not, in general, determine the resulting logical operation. Addressing this problem requires a framework that resolves both the continuous implementation of instruments and the relevant multi-time correlations.

Process tensors, also known as quantum combs, provide a standard framework for describing multi-time quantum processes. They map sequences of interventions to the corresponding output statistics and encode temporal correlations arising from environmental memory \cite{combref009,combref017,combref019,combref029}. In the discrete process-tensor setting, Choi-channel representations show that Pauli twirling under the corresponding protocol conditions can convert quantum temporal noise correlations into classical correlations. This construction supports probabilistic error cancellation and virtual channel purification \cite{liu2024choichannel}. Related calibration-free suppression has been demonstrated experimentally in nuclear-spin systems \cite{liu2026calibrationfree}. Tensor-network process-tensor decoders provide a complementary approach to error correction that retains spatiotemporal correlations \cite{kobayashi2024decoders}. These developments establish the relevance of multi-time error descriptions, but a discrete process tensor with intervention slots only at the boundaries of an instrument window does not resolve the dynamics within that window. The continuous process tensor (cPT) framework places interventions on a continuous physical time axis \cite{wassnerCPT}. Its field representation uses continuous matrix product states (cMPS), providing a compact tensor-network description of continuous quantum processes \cite{wassnerCPT,combref024,combref039}. A recent dilation theorem further shows that process-canonical representations are dense in the space of physical cPTs \cite{fuksa2026dilation}.

Research on quantum processes has developed along two complementary directions: process characterization and numerical simulation. Operational process tomography addresses non-Markovian dynamics under restricted experimental access \cite{combref003,combref004,combref051} and has been demonstrated on superconducting quantum processors \cite{combref005,combref076}. Experimental characterization of non-Markovian processes has also been reported in other settings \cite{combref079}. Subsequent work has improved sampling efficiency and self-consistent reconstruction \cite{combref006,combref007}; Hamiltonian learning provides another route to characterizing superconducting-processor dynamics \cite{combref091}. Isometric, compressed, and gate-set approaches offer complementary tools for learning finite combs and gate models \cite{combref081,combref082,combref087,combref088}. These methods make multi-time process descriptions accessible to experimental characterization and reconstruction.

Numerical methods address the cost of simulating a process while retaining its environmental history. Time-evolving influence functionals and matrix-product compression reduce the computational cost of long histories while preserving bath memory \cite{combref001,combref002,combref008,combref026,combref027,combref028,combref065}. The geometry and optimization of matrix product states (MPS), together with dynamical low-rank integration, provide further tools for these calculations \cite{combref068,combref070,combref071,combref072,combref074,combref075}. Continuous tensor-network and cMPS methods supply a numerical language for continuous fields \cite{combref025,combref048,combref069,combref092,combref093,combref094,combref095}. These methods provide the computational tools needed to retain environmental memory and evaluate the effect of randomized control on multi-time processes.

For active error suppression, Hamiltonian-level random control and dynamical decoupling provide a complementary route. Viola and Knill studied random group-valued paths in a logical frame and derived an $O(g^2T\Delta t)$ error-scaling law for bounded interactions with independently refreshed control frames \cite{viola2005}. Finite-strength continuous control and random bang-bang decoupling have also been investigated within the dynamical-decoupling framework \cite{violaKnill2003,chen2006,hahn2022}. Deterministic sequences, protected gates, colored noise, and pulse imperfections further constrain performance \cite{viola1998,viola1999,khodjasteh2005,uhrig2007,alvarez2011,souza2011,cywinski2008,west2010,combref064}. These studies establish methods for suppressing system--environment coupling and preserving intended operations. They do not by themselves provide a cPT formulation that combines finite-duration logical instruments with a quantitative bound on complete adaptive output records.

The central question is how to preserve the intended logical experiment along a random control path while bounding its complete-record error in the presence of environmental memory. To address this question, we establish continuous-process randomized compilation (CPRC) within the cPT framework. We place random control-frame trajectories, finite-duration covariant instruments, and memory-bearing environmental dynamics on a common physical time axis and study four control protocols. Each protocol has a distinct temporal structure and implementation constraint. Fixed Pauli refresh independently selects control frames at prescribed nodes and serves as the benchmark for ideal digital updates. Bounded Cayley paths interpolate between independently sampled target frames over finite intervals, making the control evolution between nodes explicit. Refining these nodes increases the required control speed. Unitary Brownian motion provides an analytically tractable idealized diffusion model \cite{nechita2013,onorati2017}, rather than a finite-bandwidth drive. Ornstein--Uhlenbeck (OU) driving produces smoother control paths with temporal correlations, although a finite correlation time does not imply a strict frequency cutoff. OU-driven qubit dynamics have also been studied in the context of classical noise \cite{dekeijzer2025}. These four protocols are studied in the main text; Appendix~\ref{app:poisson} treats Poisson refresh with random event times as an event-time extension.

\paragraph{Contributions and scope.}
The results are organized into three levels.
\begin{enumerate}
\item \textbf{Logical implementation and physicality.} We formulate the pointwise covariance rule, prove that every control trajectory reproduces the target logical experiment in the absence of coupling, and use a closed control-frame path to show that endpoint compensation alone is generally insufficient. The averaged process continues to satisfy the corresponding positivity and causality requirements.
\item \textbf{Error structure and complete records.} We identify the designated boundaries at which independent whole-interval Pauli conjugations diagonalize error histories, and use an exactly solvable two-cell model to show that the resulting classical labels can still carry environmental memory. For bounded centered coupling and controls satisfying the required mixing conditions, we derive a total-variation (TV) bound for complete adaptive records. The OU specialization depends on a mixing condition for the joint process of the control frame and OU driving variables.
\item \textbf{Process convergence and numerical checks.} We prove convergence of trajectory-averaged cPT coefficients in each fixed particle sector for finite-dimensional non-Gaussian environments and Gaussian baths with the stated physical-cutoff realization and time-integrable covariance, including the uncut Drude spectrum. A uniform particle-sector tail estimate then establishes Fock-norm convergence of the entire trajectory-averaged, trace-closed cPT on a fixed finite time interval; in the Gaussian case, cutoff removal is justified in the same norm. A qubit environment is used to study finite-duration output probabilities and operator conditional mutual information (CMI). We present numerical verification of all operation-label coefficients at a prescribed pair of insertion times in the two-particle sector, for both a qubit environment and a thermal Drude bath. The process-convergence results and the complete-record bound under bounded coupling concern different convergence objects and rely on different premises.
\end{enumerate}
For clarity, Table~\ref{tab:theorems} summarizes the distinct convergence objects, assumptions, and conclusions established in this work.

\begin{fixedtable}
\fixedtablecaption{Distinct statements proved in this work. ``Fixed sector'' means that the number and times of process-field insertions are held fixed while the refresh rate changes. A statement about all allowed experimental outcomes does not, by itself, give an entropy-continuity theorem for the infinite Fock representation.}{tab:theorems}
\begin{ruledtabular}
\begin{tabular}{p{0.22\linewidth}p{0.38\linewidth}p{0.30\linewidth}}
Object & Main premise & Conclusion\\ \hline
Interval history & factorized initial state; independent whole-interval Pauli frames & diagonal histories; memory may remain\\
Finite environment sector & bounded interaction; all finite control moments; Proposition~\ref{prop:finite-sector} & fixed coefficient convergence\\
Gaussian sector & factorized Gaussian input; Hermitian linear fields; physical cutoff realization; $L^1$ covariance and cutoff stability; all finite control moments; Proposition~\ref{prop:fixed-sector} & fixed coefficient convergence\\
Full trace-closed cPT & normalized input; pointwise-covariant logical operations; either sector premise above; fixed finite system, interval, and operation basis; Theorem~\ref{thm:fock-convergence} & Fock-norm convergence of the trajectory average\\
Uncut Drude sector & $C\in L^2$, $C(0)=\infty$; Eq.~\eqref{eq:app-Drude-refresh-rate} & $O[h\log(\tau_0/h)]$ refresh rate\\
Complete experiment & integrable norm envelope; centered coupling; initial control--state independence; autonomous no-feedback Markov control; pointwise-covariant logical instruments; Theorem~\ref{thm:complete-experiment} & Eq.~\eqref{eq:general-bound}\\
Two-cell benchmark & static finite bath & exact classical-memory/TV formulas
\end{tabular}
\end{ruledtabular}
\end{fixedtable}
The first row applies the standard Pauli character average to designated cPT intervals. The fixed-refresh scaling in the complete-experiment row has a random-decoupling precedent \cite{viola2005}; its present operational output metric and the correlated-control versions of Eq.~\eqref{eq:general-bound} are the extensions assessed here. The finite-environment and Gaussian rows concern cPT amplitudes; the full-process row combines their convergence with a uniform summable tail as proved in Theorem~\ref{thm:fock-convergence}. These results use hypotheses different from the bounded-interaction premise of the complete-experiment theorem. The Drude and two-cell rows make those distinctions explicit rather than identifying diagonalization with loss of memory.

This work extends RC from discrete circuit settings to continuous-time quantum processes. It distinguishes logical transparency, interval classicalization, and operational decoupling in memory-bearing quantum processes, providing a systematic theoretical framework and quantitative methods for error suppression in quantum control in non-Markovian environments.

\input{background}

\section{Compiled dynamics and finite-duration instruments}
We now specify the laboratory dynamics and the physical implementation of
instruments contracted with the cPT of Eq.~\eqref{eq:cpt-expansion}.
The operational sequence underlying this construction is summarized in
Fig.~\ref{fig:protocol}.

\begin{figure}[t]
\includegraphics[width=\linewidth]{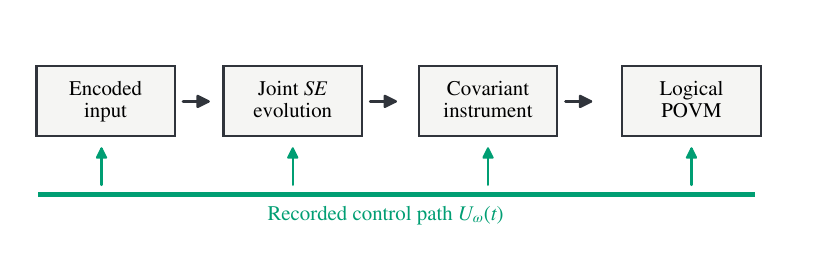}
\caption{Operational structure. The encoded input, joint system--environment evolution, finite-duration covariant instrument, and logical POVM are evaluated along one complete random control path. The same path is used on both sides of every multi-time contraction before trajectory averaging.}
\label{fig:protocol}
\end{figure}

\subsection{Laboratory and random-unitary frames}
We use the laboratory Hamiltonian
\begin{equation}
 H_{\rm lab}^{\omega}(t)=H_E(t)+H_{SE}(t)+H_{\rm ctrl}^{\omega}(t)
 +H_{\rm inst}^{\omega}(t),
\end{equation}
where the autonomous random control generates a unitary frame $U_\omega(t)$ on $S$. Transforming with $U_\omega(t)$ gives
\begin{equation}
 H_{\rm I}^{\omega}(t)=H_E(t)
 +\bigl(U_\omega^\dagger(t)\otimes\id_E\bigr)H_{SE}(t)
 \bigl(U_\omega(t)\otimes\id_E\bigr).
 \label{eq:random-frame}
\end{equation}
Here $H_E(t)$ and $H_{SE}(t)$ may have prescribed, deterministic time dependence independent of the control path and measurement record. We assume that $H_E(t)$ generates a well-defined propagator. The system-identity component of $H_{SE}(t)$ is absorbed into $H_E(t)$, leaving a system-traceless coupling; the static Hamiltonians used in the numerical examples are special cases. The ideal comparison retains the same environment propagation and removes only this coupling.
This baseline model has no separate system Hamiltonian $H_S$: the instrument Hamiltonian generates the intended system operations on their finite supports. Equation~\eqref{eq:random-frame} isolates the uncontrolled system--environment process; the logical instrument Hamiltonian is added during each instrument window. Changing frame leaves the environmental state intact throughout the trajectory. For the Drude calculation we use $H_{SE}=g\sigma_z\otimes B$ with a finite-temperature Gaussian bath. For the qubit checks we use $H_E=(1.7/2)\sigma_z^{(E)}$ and $H_{SE}=0.8\sigma_z^{(S)}\otimes\sigma_x^{(E)}$.

\subsection{Covariant finite-duration instruments}
Let $\mathcal A_t$ be a logical instrument in the logical frame. A physical implementation at a trajectory $\omega$ must satisfy, at every time during its support,
\begin{equation}
 \mathcal A_t^{\omega}[X]=(U_\omega(t)\otimes\id_A)\,
 \mathcal A_t\!\left[(U_\omega^\dagger(t)\otimes\id_A)X
 (U_\omega(t)\otimes\id_A)\right]
 (U_\omega^\dagger(t)\otimes\id_A),
\label{eq:covariant-instrument}
\end{equation}
with the evident extension to an adaptive classical record and an auxiliary system $A$. For a Hamiltonian realization with instrument memory $M$, this requires the laboratory interaction at every time in the instrument support to be
\begin{equation}
 H_{SM,\omega}^{\rm inst,lab}(t)
 =(U_\omega(t)\otimes\id_M)H_{SM}^{\rm inst,log}(t)
 (U_\omega^\dagger(t)\otimes\id_M).
 \label{eq:instrument-hamiltonian-covariance}
\end{equation}
The frame-generating control Hamiltonian is a separate term in $H_{\rm lab}^{\omega}$. Pointwise covariance governs the physical implementation throughout $[a,b]$. Once an isolated instrument obeys this rule, its complete outcome map has only endpoint coordinate factors: $\mathcal J_{x,\omega}^{\rm lab}=\operatorname{Ad}_{U_\omega(b)}\circ\mathcal J_x^{\rm log}\circ\operatorname{Ad}_{U_\omega(a)^\dagger}$. This endpoint identity follows from continuous covariance; endpoint pulses alone do not implement it in general. Appendix~\ref{app:instrument-covariance} derives the identity with an instrument ancilla and record. A target logical POVM $\{M_x\}$ is implemented as $U_\omega(T)M_xU_\omega^\dagger(T)$, or by an inverse-frame operation followed by a fixed detector. At a closed trace boundary the correction is invisible because $\Tr(U X U^\dagger)=\Tr X$.

\begin{proposition}[Trajectory-wise logical transparency]\label{prop:transparency}
If $H_{SE}(t)=0$, the input is encoded in the initial frame, the finite-duration instrument Hamiltonian obeys Eq.~\eqref{eq:instrument-hamiltonian-covariance} at every time, and the final POVM is frame covariant, then every control trajectory produces the target distribution of the complete adaptive logical record.
\end{proposition}
To prove the proposition before any trajectory average, transform the laboratory state to logical coordinates:
\begin{equation}
 \widetilde\rho_\omega(t)=(U_\omega^\dagger(t)\otimes\id_A)
 \rho_\omega^{\rm lab}(t)(U_\omega(t)\otimes\id_A).
\end{equation}
The environment and any instrument memory are carried along unchanged. The frame-generating term cancels the derivative of $U_\omega$ in this representation, while Eq.~\eqref{eq:instrument-hamiltonian-covariance} becomes the same logical instrument Hamiltonian for every $\omega$. Preparing the laboratory input in the encoded frame makes $\widetilde\rho_\omega(0)$ equal to the target input. The common time-ordered propagation generated by $H_E(t)$ commutes with the system frame, so every finite-duration window and every adaptive outcome map propagates the same logical conditional state on every trajectory. Finally, $M_{x,\omega}^{\rm lab}=U_\omega(T)M_xU_\omega^\dagger(T)$ gives $\Tr[M_{x,\omega}^{\rm lab}\rho_\omega^{\rm lab}(T)]=\Tr[M_x\widetilde\rho_\omega(T)]$. Appendix~\ref{app:physicality} gives the full record-preserving cancellation and the corresponding endpoint formulas.
\paragraph{A closed-frame counterexample.}
Endpoint frame checks do not enforce the pointwise rule. Set $H_{SE}=0$ and take a target instrument Hamiltonian $h_{\rm inst}^{\rm log}=\eta\sigma_x$ throughout $[0,T]$. For $\Omega>0$, consider the bounded, closed control path
\begin{equation}
 \theta(t)=\Omega\min\{t,T-t\},\qquad
 U(t)=e^{-i\theta(t)\sigma_z},\qquad U(0)=U(T)=\id.
 \label{eq:endpoint-loop}
\end{equation}
Its generator $H_{\rm ctrl}(t)=\dot\theta(t)\sigma_z$ has norm $\Omega$ almost everywhere.
An endpoint-only prescription sees the identity frame at both ends and leaves $h_{\rm inst}^{\rm lab}=\eta\sigma_x$ unchanged. In logical coordinates, however, the instrument acts as
\begin{equation}
 \widetilde h_{\rm inst}(t)=U^\dagger(t)h_{\rm inst}^{\rm lab}U(t)
 =\eta[\cos(2\theta(t))\sigma_x-\sin(2\theta(t))\sigma_y].
 \label{eq:endpoint-logical-hamiltonian}
\end{equation}
The two halves of the loop give
\begin{equation}
 \int_0^T\!dt\,\frac{\widetilde h_{\rm inst}(t)}{\eta}
 =\frac{\sin(\Omega T)}{\Omega}\sigma_x
 -\frac{1-\cos(\Omega T)}{\Omega}\sigma_y.
 \label{eq:endpoint-loop-integral}
\end{equation}
At $\Omega T=\pi$ and to first order in $|\eta|T$, the endpoint-only logical propagator is
$\widetilde V(T)=\id+(2i\eta/\Omega)\sigma_y+O((\eta T)^2)$,
whereas the target is $V_{\rm id}(T)=\id-i\eta T\sigma_x+O((\eta T)^2)$.
For input $|0\rangle$ and readout $\Pi_{y,+}=(\id+\sigma_y)/2$,
\begin{equation}
 p_{\rm end}(+)=\tfrac12+O((\eta T)^2),\qquad
 p_{\rm id}(+)=\tfrac12-\eta T+O((\eta T)^2),
 \label{eq:endpoint-loop-probability}
\end{equation}
so the output-probability difference is $\eta T+O((\eta T)^2)$ despite identical endpoint frames. Implementing $h_{\rm inst}^{\rm lab}(t)=U(t)(\eta\sigma_x)U^\dagger(t)$ instead restores the target generator at every time. The counterexample exposes the intrinsic limitation of endpoint compensation; Proposition~\ref{prop:transparency} gives a sufficient pathwise rule for preserving logical transparency and provides a direct design criterion for finite-duration instruments.

The same obstruction can be stated directly as a failure of endpoint-comb coarse graining. Let $E$ be a qubit initialized in $|0\rangle_E$, set $H_E=0$, and compare process A, with $H_{SE}^{A}(t)=0$, to process B, with
\begin{equation}
 H_{SE}^{B}(t)=\dot\theta(t)\sigma_z^S\otimes\sigma_z^E,
\end{equation}
where $\theta(t)$ is the closed loop in Eq.~\eqref{eq:endpoint-loop}. Both joint propagators from $0$ to $T$ are $\id_{SE}$, so every experiment restricted to the two endpoints sees the same channel Choi operator. Apply nevertheless the same finite-duration instrument $h_{\rm inst}(t)=\eta\sigma_x^S$ throughout the interval. In process B, the first-order interaction-picture generator on the bath branch $|0\rangle_E$ is
\begin{equation}
 \langle0|_E\!\int_0^T\!dt\,
 V_B^\dagger(t,0)(\sigma_x^S\otimes\id_E)V_B(t,0)|0\rangle_E
 =\frac{\sin(\Omega T)}{\Omega}\sigma_x
 -\frac{1-\cos(\Omega T)}{\Omega}\sigma_y.
 \label{eq:endpoint-comb-loop-instrument}
\end{equation}
At $\Omega T=\pi$ this is $-2\sigma_y/\Omega$, whereas process A gives $T\sigma_x$. For input $|0\rangle$ and terminal $\sigma_y$ readout, the two $+1$ probabilities therefore differ by $\eta T+O((\eta T)^2)$. Thus the endpoint map is not injective for finite-duration experiments: a cPT, or a comb with an interior slot, is required to resolve the distinction.

\section{Control protocols: fixed refresh, bounded paths, and continuous drives}
\label{sec:protocols}
The main text follows the physical progression from scheduled node updates to continuous stochastic control. All protocols are autonomous and initialized independently of the logical input. We focus on zero deterministic drift; a drift can be added as a separate Markov generator. Table~\ref{tab:protocols} summarizes the state variables, control parameters, and averaging objects for the four protocols.

\begin{fixedtable}
\fixedtablecaption{Control models used in the main text. Fixed refresh samples independent Pauli target frames at prescribed nodes. A bounded Cayley path reaches the same targets through a finite-time interpolation. Brownian and OU controls start from a Haar frame; OU driving starts from its stationary Gaussian law. The Poisson event-time extension and its numerical results are given in Appendix~\ref{app:poisson}.}{tab:protocols}
\begin{ruledtabular}
\begin{tabular}{llll}
Protocol & Markov state & Parameters & Averaging variable\\ \hline
Fixed refresh & current Pauli frame & holding lengths $h_j$ & interval frames\\
Cayley path & frame and edge coordinate & path duration & node frames and paths\\
Brownian & $U_t$ & $\kappa$ & frame path\\
OU & $(U_t,\boldsymbol\xi_t)$ & $\tau,\sigma^2$ & joint frame--velocity path
\end{tabular}
\end{ruledtabular}
\end{fixedtable}

\subsection{Fixed Pauli refresh}
At prescribed nodes $0=t_0<t_1<\cdots<t_K=T$, draw $G_j$ independently and uniformly from the single-qubit Pauli frame set. The frame is held at $G_j$ on $[t_j,t_{j+1})$, while the same system--environment state continues to evolve. If a complete operation slot coincides with a holding interval, its propagator is conjugated by the constant $G_j$ and the independent-interval identity in Eq.~\eqref{eq:pauli-diagonal} applies. Internal cuts of one holding interval share the same frame and are not independent compiled slots.

\paragraph{Relation to gate-level randomized compiling.}
Circuit RC samples frames around prescribed gates or cycles and compiles compensating operations so that the ideal circuit is unchanged \cite{wallman2016,hashim2021}. Fixed refresh retains these node-level draws but specifies the continuous propagation between nodes. CPRC adds the pointwise covariance rule in Eq.~\eqref{eq:instrument-hamiltonian-covariance}, so a finite-duration instrument is implemented in the current frame throughout its support rather than corrected only at its endpoints.

\subsection{Bounded Cayley realization of the fixed nodes}
The independently sampled node frames can be reached by a bounded path on the Pauli Cayley graph. For a relative target $Q_j=G_{j+1}G_j^\dagger$, choose a shortest word in the available generators $X$ and $Z$, and let $V_j(s)$ be the corresponding piecewise exponential path over the transition window $[t_{j+1}-\tau_{\rm tr},t_{j+1}]$:
\begin{equation}
 U_\omega(t)=V_j(t)G_j,\qquad V_j(0)=\id,\quad V_j(\tau_{\rm tr})=Q_j.
 \label{eq:cayley-main-path}
\end{equation}
The target frame at every node has the same independent distribution as in fixed refresh. During interpolation $V_j(s)$ is generally outside the Pauli set, so exact character cancellation applies only to whole constant-frame intervals. A Duhamel comparison with the ideal jump adds at most $2g_{\rm p}$ times the total transition duration. For the full-interval numerical paths this duration is of order $T$, and that comparison does not yield a decreasing bound. The direct estimate instead uses the two adjacent independent target frames on which each path segment depends. Their interaction correlations vanish after two node crossings, giving Eq.~\eqref{eq:cayley-direct-bound}. Appendix~\ref{app:cayley} proves the complete-experiment estimate. The required edge speed rises as $h=T/K$ shrinks, so the sequence represents a control-speed refinement. Related bounded-control and random-decoupling precedents are in Refs.~\cite{violaKnill2003,chen2006,viola2005,hahn2022}.

\subsection{Brownian and OU frame processes}
For Brownian control, choose a grid $u_0<\cdots<u_N$ and define $\Gamma_k=U(u_{k+1})U^\dagger(u_k)$. Its Markov chain rule is
\begin{equation}
 d\Pr(U_0,\ldots,U_N)=d\mu(U_0)\prod_{k=0}^{N-1}
 q_{\Delta_k}(U_{k+1}U_k^\dagger)d\mu(U_{k+1}),
 \label{eq:path-density}
\end{equation}
where $q_\Delta$ is the increment density relative to Haar measure. Brownian increments are independent, but successive endpoint frames remain correlated. For a Brownian heat kernel,
\begin{equation}
 q_\Delta(G)=\sum_{j=0,1/2,1,\ldots}(2j+1)e^{-\kappa j(j+1)\Delta/2}\chi_j(G).
\end{equation}
For OU driving, $\tau$ and $\sigma^2$ set the velocity correlation time and stationary variance. The joint frame--velocity state must be propagated without redrawing the velocity at each grid point. Its high-frequency tail means that finite $\tau$ is not a hard bandwidth cutoff. Brownian motion is obtained from OU only in a separate short-correlation limit with the integrated variance held fixed.

\paragraph{Event-time extension.}
Poisson refresh draws the same Pauli frames as fixed refresh but places the nodes on a rate-$\lambda_r$ clock. It is useful when an unsynchronized, time-homogeneous event model is required, but a finite-rate trajectory can contain a long interval without an update. The complete event-schedule definition, conditional classicalization proof, exponential two-time correlation, finite-pulse bound and numerical scans are collected in Appendix~\ref{app:poisson}.

\begin{proposition}[Physicality of the averaged process]\label{prop:averaged-cpt}
For each main-text control law, whenever the pathwise joint propagator and its cPT are well defined and integrable, the trajectory-averaged logical-frame cPT is positive and causal. Averaging complete trajectories commutes with recovery of any fixed finite comb and with contraction against a common logical finite-duration instrument implemented covariantly on each path. The same statement holds for the event-time extension of Appendix~\ref{app:poisson}.
\end{proposition}
For each trajectory, unitary system--environment dilation and completely
positive instrument maps give positivity, while forward time ordering gives
the causal recursion for the open-output finite comb. Closing its final system
output with the terminal trace gives the corresponding tester normalization;
neither step requires time-independent $H_E$ or $H_{SE}$. Convex combinations
preserve these properties. Linearity of recovery and contraction maps then
gives the stated commutation; the full construction is in
Appendix~\ref{app:physicality}.

With the control processes and their implementation constraints specified, we next examine which error structures are removed by interval-wise averaging and which environmental correlations remain.

\section{Classicalization and operational decoupling}
\subsection{Independent interval classicalization}
\begin{proposition}[Classicalization at independent interval boundaries]\label{prop:interval-classicalization}
Assume the initial joint state factorizes as
$\rho_{SE}(0)=\rho_S\otimes\rho_E$ and is independent of the sampled
frames. Suppose each designated whole-interval joint propagator is conjugated
by an independently sampled uniform Pauli frame, with operation slots only at
the interval boundaries. After averaging, its finite comb is a mixture of
Pauli error histories with nonnegative normalized joint weights. The weights
may retain correlations across intervals.
\end{proposition}
Let $a_j,b_j$ denote the Pauli error labels on the two sides of an amplitude, and define the real Pauli character $\chi_a(G)\in\{+1,-1\}$ by $G^\dagger P_aG=\chi_a(G)P_a$ (global phases are irrelevant). If independently refreshed Pauli frames $G_j$ are assigned to the prescribed intervals, the character average factorizes:
\begin{equation}
 \E\!\left[\prod_j\chi_{a_j}(G_j)\chi_{b_j}(G_j)\right]
 =\prod_j\delta_{a_jb_j}.
 \label{eq:pauli-diagonal}
\end{equation}
The assumption behind Eq.~\eqref{eq:pauli-diagonal} is a conjugation of each \emph{whole} interval propagator by one constant frame, with operation slots at the interval boundaries. Prescribed $H_E(t)$ and $H_{SE}(t)$ are included because their time ordering is contained in the frame-independent whole-interval propagator $W_j$. To see why this produces a genuine probability distribution rather than merely a diagonal-looking expansion, write $W_j=\sum_aP_a\otimes B_{j,a}$, where $B_{j,a}=d^{-1}\Tr_S[(P_a^\dagger\otimes\id_E)W_j]$ is an operator on the environment and $d$ is the system dimension. Conjugation by $G_j$ multiplies its system term by $\chi_a(G_j)$. After the independent averages remove every unequal forward--backward pair, the surviving history $\boldsymbol a=(a_1,\ldots,a_m)$ has weight
\begin{equation}
 p(\boldsymbol a)=\Tr_E\!\left[B_{m,a_m}\cdots B_{1,a_1}\rho_E
 B_{1,a_1}^\dagger\cdots B_{m,a_m}^\dagger\right]\ge0.
 \label{eq:main-history-prob}
\end{equation}
Unitarity gives $\sum_a B_{j,a}^\dagger B_{j,a}=\id_E$, so summing the labels backwards yields $\sum_{\boldsymbol a}p(\boldsymbol a)=1$. The specified slots therefore have an exact mixture of Pauli error histories. The common environment state in Eq.~\eqref{eq:main-history-prob} can still correlate their classical labels. For a continuously varying controlled Hamiltonian, independent endpoint frames alone do not supply the required whole-interval conjugation. Equation~\eqref{eq:pauli-diagonal} applies to that interval construction, with its process-level proof in Appendix~\ref{app:classicalization}.

Fixed refresh gives Eq.~\eqref{eq:pauli-diagonal} at its prescribed nodes. Appendix~\ref{app:poisson} derives the event-conditioned version and the residual correlation across preassigned cuts for Poisson refresh.

Continuous frames do not obey Eq.~\eqref{eq:pauli-diagonal} at arbitrary time cuts either. For an isotropic Brownian frame, the adjoint rotation obeys
\begin{equation}
 \E[R_{ab}(U_t)R_{cd}(U_s)]
 =\frac13e^{-\kappa(t-s)}\delta_{ac}\delta_{bd}, \qquad t\ge s,
 \label{eq:brown-correlation}
\end{equation}
which is nonzero at finite separation. OU paths are smoother still at short times, with $a_{\rm OU}(\Delta)=1-O(\Delta^2)$. Refining a numerical grid resolves a correlated trajectory; it does not generate independent Haar frames. Thus exact interval classicalization and rapid operational decoupling are distinct statements.

\subsection{An exact two-cell memory benchmark}
This example makes the distinction between an exactly diagonal error history and operational decoupling explicit. Let the system and environment each be a qubit and take the static bath
\begin{equation}
 H_E=0,\qquad H_{SE}=gZ_S\otimes |1\rangle\langle1|_E,
 \qquad \rho_E=(1-\varpi)|0\rangle\langle0|_E
 +\varpi|1\rangle\langle1|_E,\quad 0<\varpi<1.
 \label{eq:two-cell-model}
\end{equation}
The system starts in $|+\rangle$. On two cells of duration $h$, independently sample uniform Pauli frames $G_j$ and write $y_j=\chi_Z(G_j)=\pm1$. In bath branch $b\in\{0,1\}$, the logical-frame propagation of cell $j$ is $U_{j,b,y}=e^{-igbhy_jZ}$. The bath branch is unchanged at the intermediate cut. Put $c=\cos(gh)$ and $s=\sin(gh)$. For $b=1$, the cell error operator is $cI-iy_jsZ$; for $b=0$, it is $I$. The character identity in Eq.~\eqref{eq:pauli-diagonal} removes forward--backward interference between different labels $a_j\in\{0,1\}$, where $0$ denotes $I$ and $1$ denotes $Z$. Direct multiplication gives
\begin{equation}
 p_{00}=1-\varpi+\varpi c^4,\qquad
 p_{01}=p_{10}=\varpi c^2s^2,\qquad p_{11}=\varpi s^4.
 \label{eq:two-cell-weights}
\end{equation}
These four weights are nonnegative and sum to one. Each marginal error probability is $q=\varpi s^2$, yet
\begin{equation}
 p_{11}-p_{1\cdot}p_{\cdot1}=\varpi(1-\varpi)s^4>0\qquad(s\ne0).
 \label{eq:two-cell-covariance}
\end{equation}
Independent Pauli frames therefore remove quantum interference between error labels, while the shared, unreset bath branch correlates their classical weights. If $H_2$ and $H_4$ denote binary and four-outcome Shannon entropies in bits, the classical error-label mutual information is
\begin{equation}
 I(A_1:A_2)=2H_2(\varpi s^2)
 -H_4(p_{00},p_{01},p_{10},p_{11}).
 \label{eq:two-cell-information}
\end{equation}
This is an entropy of the designated interval labels, not the continuous-field operator CMI introduced below.

To read the history operationally, measure $X$ projectively after each cell, retain both results, and implement these measurements covariantly. The first measurement does not alter the classical bath branch. Conditional on $b=1$, an $X$ outcome flips on one cell with probability $s^2$; the Pauli sign changes only its phase. With initial outcome convention $x_0=+$, the flip labels map to records as $(a_1,a_2)=(0,0)\mapsto(+,+)$, $(0,1)\mapsto(+,-)$, $(1,0)\mapsto(-,-)$, and $(1,1)\mapsto(-,+)$. The uncoupled ideal record is certainly $(+,+)$. Hence the complete-record total-variation distance is exactly
\begin{equation}
 \epsilon_{2\,{\rm cells}}=1-p_{00}
 =\varpi[1-\cos^4(gh)]
 =\varpi g^2Th+O(g^4h^2T^2),\qquad T=2h.
 \label{eq:two-cell-record}
\end{equation}
Corollary~\ref{cor:static-record} gives $\epsilon\le(3/2)g^2Th$ for the same interaction norm $g$. The exact example attains the $g^2Th$ order but not the conservative prefactor. It is a memory-bearing classicalization benchmark at specified constant-frame slots; a full-interval Cayley transition does not satisfy the same whole-cell conjugation premise.

The construction extends to $K=T/h$ independently sampled constant frames, with an $X$ measurement and retained result at each node. For an error string $\boldsymbol a\in\{0,1\}^K$, let $|\boldsymbol a|$ be its Hamming weight. Since the bath branch persists,
\begin{equation}
 p(\boldsymbol a)=(1-\varpi)\mathbf 1_{\boldsymbol a=\boldsymbol 0}
 +\varpi c^{2(K-|\boldsymbol a|)}s^{2|\boldsymbol a|},
 \qquad
 \epsilon_K=\varpi(1-c^{2K}).
 \label{eq:K-cell-TV}
\end{equation}
At fixed $T$ and small $h$, $\epsilon_K=\varpi g^2Th+O(g^4T^2h^2)$. These closed expressions provide an analytic benchmark for the rate and the persistence of classical memory.

\subsection{Process convergence and complete-record error}
\begin{proposition}[Bounded-environment fixed-sector limit]\label{prop:finite-sector}
Fix a cPT particle number, insertion times, and bounded insertion maps. Suppose the system--environment interaction is bounded, the control is autonomous and independent of the initial joint state, and the intended instrument is pointwise covariant. If every finite multi-time tensor moment of the rotated, centered system coupling vanishes at distinct fixed times in the fast-control limit, then each trajectory-averaged particle amplitude converges to its uncoupled value. This applies to a finite-dimensional, non-Gaussian environment without resetting its state. The single-qubit environment used below satisfies the boundedness premise.
\end{proposition}
The proof expands the complete joint evolution in a bounded Dyson series. Finite-time control moments eliminate every nonzero order, and a uniform exponential majorant permits summation. Prescribed $H_E(t)$ and $H_{SE}(t)$ are allowed when the environment propagator exists and the coupling has a control-independent integrable norm envelope; Appendix~\ref{app:fixed-sector-general} gives the corresponding majorant. In the numerical qubit model this applies to every operation-label pair at the specified two insertion times. Fixed and Poisson refresh, node-matched full-cell Cayley paths, and Brownian control satisfy the moment condition in their fast limits. OU does so under \emph{overall time acceleration} and mixing of the joint frame--velocity process. Appendix~\ref{app:fixed-sector-general} proves the result and checks each protocol. This coefficient statement uses weaker mixing assumptions than a finite-rate record bound; its conclusion concerns fixed coefficients.

\begin{proposition}[Time-integrable Gaussian fixed-sector limit]\label{prop:fixed-sector}
Assume an initially factorized state $\rho_{SE}(0)=\rho_S\otimes\rho_E$,
where $\rho_E$ is centered Gaussian (or quasifree), and a coupling linear in
finitely many Hermitian bath fields with bounded Hermitian system factors and covariance matrix $C(t,s)$ integrable in the matrix spectral norm on
$[0,T]^2$. Assume the autonomous control is independent of this initial
state. At each cutoff, assume the purified Fock-field realization and integrable form-factor norm of Eq.~\eqref{eq:app-gaussian-cutoff-realization}; Appendix~\ref{app:gaussian-L1-sector} constructs its unitary propagator and identifies its insertion coefficients with the Dyson--Wick expansion. Let the covariance family converge to $C$ in that $L^1$ norm; define uncut fixed
coefficients by this cutoff limit. For fixed particle order, times, and bounded
control-independent logical insertions, assume
the same finite-time control-moment condition for every coupling channel.
Then each trajectory-averaged cPT coefficient tends to its uncoupled value.
The ultraviolet-cutoff and fast-control limits commute for this fixed sector.
In the stationary single-axis model, independent Pauli refresh gives a
conservative $O(\sqrt h)$ bound when $C\in L^2(0,T)$.
\end{proposition}
The proof retains all bath pairings in a cutoff-uniform Dyson--Wick series. A time-dependent environment is included when it preserves the quasifree bath structure and the resulting nonstationary covariance satisfies the stated $L^1$ and cutoff-stability conditions; prescribed time-dependent coupling coefficients enter that covariance. In the independent-sign case, separating within-cell and cross-cell influence terms gives a direct estimate. The uncut Drude spectrum has $C(0)=\infty$ but $C\in L^2(0,T)$, so it satisfies the integrable-covariance premise. Its short-time logarithm sharpens the fixed-refresh rate to $O[h\log(\tau_0/h)]$. For Brownian and time-accelerated OU frames with exponential \emph{joint-state} mixing, the Drude sector has an $O(\log r/r)$ bound, where $r$ is the dimensionless control acceleration. A general $L^2$ kernel with integrable control mixing gives $O(r^{-1/2})$. Appendix~\ref{app:fixed-sector-general} proves these rates under stated cutoff and mixing hypotheses. The numerical Drude OU sweep instead holds the velocity correlation time fixed while changing its variance.
\begin{theorem}[Fock-norm convergence of the trajectory-averaged cPT]
\label{thm:fock-convergence}
Fix a finite-dimensional system, a finite interval $[0,T]$, and the
control-independent operation basis of Eq.~\eqref{eq:cpt-expansion}.
Consider the trace-closed logical-frame cPT before contraction with a
selected instrument, with normalized input and pointwise-covariant logical
operations. The autonomous control is independent of the initial state,
and the rotated, centered system factors satisfy the all-finite-moment
condition of Proposition~\ref{prop:finite-sector}. Assume either
(i) a finite-dimensional environment with the bounded interaction of
Proposition~\ref{prop:finite-sector}, or
(ii) the factorized centered Gaussian input, finite-channel linear coupling,
and $L^1$ covariance and cutoff-stability assumptions of
Proposition~\ref{prop:fixed-sector}.
Define the trajectory average by
$|\overline\Upsilon_{T,r}\rangle=\E_\omega|\Upsilon_{T,r,\omega}\rangle$,
first at each cutoff in case (ii).
Writing $|\Upsilon_{T,{\rm id}}\rangle$ for its uncoupled counterpart,
\begin{equation}
 \lim_{r\to\infty}
 \bigl\||\overline\Upsilon_{T,r}\rangle
       -|\Upsilon_{T,{\rm id}}\rangle\bigr\|_\Gamma=0.
 \label{eq:fock-convergence}
\end{equation}
In case (ii), the uncut averaged process is defined by the Fock-norm limit
of its regularized processes. The ultraviolet-cutoff and fast-control
limits commute in this norm. This includes the uncut Drude model under
the covariance assumptions of Proposition~\ref{prop:fixed-sector}.
\end{theorem}

The ordered-simplex Fock norm and factorial sector estimates are developed
in Ref.~\cite{fuksa2026dilation}. Appendix~\ref{app:fock-convergence}
proves the uniform coefficient majorant and summable tail needed here,
then combines them with Propositions~\ref{prop:finite-sector}
and \ref{prop:fixed-sector}. Thus the conclusion concerns the entire
trajectory-averaged process vector, not an individual control trajectory.
For any fixed continuous instrument $\mathcal I$ with finite induced norm,
the continuous Born rule gives
\begin{equation}
 \|p_r-p_{\rm id}\|_{\rm TV}
 \le\frac12\|\mathcal I\|_{\Gamma\to L^1}
 \bigl\||\overline\Upsilon_{T,r}\rangle
       -|\Upsilon_{T,{\rm id}}\rangle\bigr\|_\Gamma,
 \label{eq:fock-instrument-bound}
\end{equation}
as in Eq.~(8) of Ref.~\cite{fuksa2026dilation}.
This also gives convergence uniformly over instruments with a common
finite norm bound. It supplies no general convergence rate or entropy
continuity statement. The following theorem instead gives an explicit
complete-record bound under bounded coupling without imposing a common
Fock-space instrument norm.

Consider a family $\mathfrak I_{n_{\rm inst}}$ of adaptive experiments with at most $n_{\rm inst}$ finite instruments. Let $x$ denote the complete logical record and define
\begin{equation}
 \epsilon_{n_{\rm inst}}(T)=\sup_{\mathcal I\in\mathfrak I_{n_{\rm inst}}}
 \frac12\sum_x|p_{\rm comp}(x|\mathcal I)-p_{\rm id}(x|\mathcal I)|.
 \label{eq:operational-error}
\end{equation}
\begin{theorem}[Complete-experiment operational error bound]\label{thm:complete-experiment}
Let the full Markov control state $Z_t$ be initially independent of the joint state. Its transition law is autonomous and is not altered by feedback from the system, instruments, or previously generated measurement records; the control may be statistically correlated with the records it generates. Its time-dependent laws $\pi_t$ and transition kernels $P_{s,t}$ satisfy $\pi_sP_{s,t}=\pi_t$. Let prescribed $H_E(t)$ generate a common environment propagator, and let the rotated system--environment interaction be centered at each time, with $\int\pi_t(dz)H_{{\rm int},z}(t)=0$ and $\|H_{{\rm int},z}(t)\|_\infty\le g(t)$ for a nonnegative $g\in L^1[0,T]$. Every finite-duration instrument is pointwise covariant as in Eq.~\eqref{eq:covariant-instrument}; in the logical frame, the complete record-preserving instruments and terminal readout are independent of $z$. With
$b(s,t)=\int\pi_s(dz)\|P_{s,t}(z,\cdot)-\pi_t\|_{\rm TV}$, the complete adaptive-record error obeys, uniformly over this experiment family,
\begin{equation}
 \epsilon_{n_{\rm inst}}(T)\le\min\!\left\{1,4\int_0^Tdt\int_0^t ds\,g(t)g(s)b(s,t)\right\}.
 \label{eq:general-bound}
\end{equation}
The ideal experiment retains the same time-ordered environment propagation and removes only the centered interaction.
\end{theorem}
The bounded-interaction hypothesis is distinct from the time-integrable covariance hypothesis of Proposition~\ref{prop:fixed-sector}. A Gaussian field may be unbounded as an operator even when its finite-time covariance is integrable; conversely, an operator-norm bound does not supply the covariance and cutoff assumptions of the Gaussian proposition.

The proof keeps the system, environment, instrument memory, and complete classical record in one conditional state $\rho_t(z)$. Its control-dependent part is $q_t(z)=\rho_t(z)-\overline\rho_t$, with $\overline\rho_t=\int\pi_t(dz)\rho_t(z)$ and $q_0=0$. The exact Duhamel source $f_s(z)=-i[H_{{\rm int},z}(s),\rho_s(z)]$ obeys $\|f_s\|_{\infty,1}\le2g(s)$ because each conditional joint-record state is normalized. The time-dependent operator-kernel lemma then gives $\|q_t\|_{\pi_t,1}\le4\int_0^t g(s)b(s,t)ds$.

Pointwise centering cancels the interaction acting on $\overline\rho_t$. A second commutator estimate bounds the averaged-state trace distance by $8\int_0^Tdt\int_0^t ds\,g(t)g(s)b(s,t)$. Tracing the bath and reading the complete record divides this bound by two. The number of instruments does not enter because each complete covariant record-preserving map contracts the trace norm. Appendix~\ref{app:bounds} gives the operator domains and full Duhamel derivation.

\begin{corollary}[Uniform coupling envelope]\label{cor:static-record}
Under Theorem~\ref{thm:complete-experiment}, if $g(t)\le g_*$ almost everywhere, then
\begin{equation}
 \epsilon_{n_{\rm inst}}(T)\le\min\!\left\{1,4g_*^2\int_0^Tdt\int_0^t ds\,b(s,t)\right\}.
 \label{eq:static-record-bound}
\end{equation}
If the control is stationary, $\pi_t=\pi$ and $b(s,t)=b(t-s)$; fixed refresh and Cayley instead use their time-dependent laws $\pi_t$. A static coupling of norm at most $g_*$ is one instance of this uniform envelope.
\end{corollary}

For the four main-text protocols, put $a_j=\int_{I_j}g(t)dt$, where $I_j$ is a prescribed fixed-refresh interval or a Cayley cell. If $g\in L^2[0,T]$, Eq.~\eqref{eq:general-bound} gives
\begin{align}
 \epsilon^{\rm fixed}&\le\min\!\left\{1,\frac32\sum_j a_j^2\right\},\\
 \epsilon^{\rm Cayley}&\le\min\!\left\{1,4\left[\frac{15}{32}\sum_{j=0}^{K-2}a_j^2
 +\frac38a_{K-1}^2+\frac34\sum_{j=0}^{K-2}a_ja_{j+1}\right]\right\},\\
 \epsilon^{\rm Br}&\le\min\!\left\{1,\frac{4C_B}{\kappa}\|g\|_2^2\right\},\\
 \epsilon^{\rm OU}&\le\min\!\left\{1,4t_{\rm mix}^{\rm OU}\|g\|_2^2\right\}.
 \label{eq:weighted-protocol-bounds}
\end{align}
Here $C_B$ is the integrated mixing envelope at unit Brownian diffusion strength, and $t_{\rm mix}^{\rm OU}$ belongs to the full frame--velocity process. The OU line uses the corresponding joint-state mixing integral. The fixed and Cayley lines themselves require only $g\in L^1$; $L^2$ is used to express the continuous-drive lines through $\|g\|_2^2$. For Cayley, $I_{K-1}$ is the terminal holding cell and the adjacent-cell term accounts for its shared target frame. Appendix~\ref{app:poisson} gives the event-time weighted counterpart.

For the uniform envelope of Corollary~\ref{cor:static-record}, bounding $a_j\le g_*|I_j|$ and $\|g\|_2^2\le g_*^2T$ gives the static-coupling values used elsewhere in this paper:
\begin{align}
 \epsilon^{\rm fixed}&\le \min\{1,\tfrac32g_*^2T\Delta_r\},\\
 \epsilon^{\rm Cayley}&\le \min\{1,\tfrac{39}{8}g_*^2Th-\tfrac{27}{8}g_*^2h^2\},\\
 \epsilon^{\rm Br}&\le \min\{1,4C_Bg_*^2T/\kappa\},\\
 \epsilon^{\rm OU}&\le \min\{1,4g_*^2Tt_{\rm mix}^{\rm OU}\}.
 \label{eq:protocol-bounds}
\end{align}
For fixed refreshes, $\Delta_r$ is the maximum holding interval. The Cayley expression assumes equal intervals $h=T/K$ and independent uniform Pauli targets. Each shortest-word path depends only on its endpoint frames, and the conjugated interaction is centered at every time. The bound includes the terminal holding interval and follows from a two-frame kernel (Appendix~\ref{app:cayley}). It bounds complete-record total variation, but does not assert exact interval classicalization along a Cayley path. Both scheduled-control rates scale as $O(g_*^2Th)$, consistent with random-decoupling precedents \cite{viola2005}. The present result applies this rate to complete records of finite-duration covariant instruments. The Poisson bound is in Appendix~\ref{app:poisson}; the uncut Drude model is treated separately through Theorem~\ref{thm:fock-convergence} and the fixed-sector rates of Appendix~\ref{app:fixed-sector-general}.

The bound is a sufficient upper bound obtained from the triangle-inequality estimate. Its fixed-refresh scaling is attained in a simple experiment. With $H_{{\rm int},G}=g\chi_z(G)\sigma_z$, initial state $|+\rangle$, and final $\sigma_x$ measurement, $K$ independent equal-length intervals give
\begin{equation}
 \epsilon_{\rm example}(T)=\frac{1-[\cos(2gh)]^K}{2}
 =g^2Th+O(g^4T^2h^2+g^4Th^3),\qquad h=T/K.
 \label{eq:fixed-bound-example}
\end{equation}
The expansion holds when $gh\ll1$ and $g^2Th\ll1$. Equation~\eqref{eq:protocol-bounds} gives $3g^2Th/2$ for the same setting. For illustration, $g=0.8$, $T=1.2$, and $K=24$ give exact error $0.0370$ and bound $0.0576$. The bound has the correct leading rate here, while Appendix~\ref{app:bounds} derives the exact expression.

For Brownian control, the bound uses group-diffusion mixing and scales with $1/\kappa$. For OU control, $t_{\rm mix}^{\rm OU}$ integrates the total-variation mixing of the \emph{joint} frame--velocity process. Under the stated time acceleration $U^{(r)}(t)=U^{(1)}(rt)$ and $\boldsymbol\xi^{(r)}(t)=r\boldsymbol\xi^{(1)}(rt)$, it obeys $t_{{\rm mix},r}^{\rm OU}=t_{{\rm mix},1}^{\rm OU}/r$, provided the baseline integral is finite. This scaling comes with $\tau_r=\tau_1/r$ and $\sigma_r^2=r^2\sigma_1^2$. At fixed $\tau$ and $\sigma$, the velocity correlation time alone gives no bound of this form: slow frame modes can outlive it. The joint generator, rather than $\tau$ alone, must therefore determine a numerical mixing estimate.

The full-interval Cayley bound is independent of the pulse-width comparison: a sequence of paths can decorrelate even when the union of its transition supports has length of order $T$. Conversely, the direct bound assumes the sampled target frames remain independent and the path rule introduces no longer-lived controller memory. It does not cover a fixed hardware speed cap as $K\to\infty$.

\paragraph{Finite-width refreshes and the error budget.}
The fixed-refresh bound in Eq.~\eqref{eq:protocol-bounds} concerns ideal instantaneous frame updates; the Cayley bound separately treats an idealized finite-time path. For a physical update of width at most $\delta$, factor out its intended frame rotation and denote by $H_{\rm err}(t)$ the difference between the remaining joint Hamiltonian and its ideal-jump counterpart. Assume $\|H_{\rm err}(t)\|\le g_{\rm p}$ uniformly over allowed instruments and control records. This $g_{\rm p}$ includes the mismatch in system--environment coupling and any instrument-implementation error. It is distinct from $g_*$: even when the interaction mismatch is the only residual, comparing two conjugates of an interaction of norm at most $g_*$ gives only the general estimate $g_{\rm p}\le 2g_*$. If $N_{\rm p}$ nonoverlapping pulse supports occur in a fixed-refresh experiment, the Duhamel estimate and the triangle inequality give
\begin{equation}
 \epsilon_{n_{\rm inst}}^{\rm fixed,pulse}(T)
 \le\min\!\left\{1,\frac32g_*^2T\Delta_r
 +2g_{\rm p}N_{\rm p}\delta\right\}.
 \label{eq:fixed-finite-pulse-bound}
\end{equation}
For a Poisson schedule the pulse count is random and possible pulse overlap requires a separate term. The resulting ensemble bound and its interpretation are given in Appendix~\ref{app:poisson}.

For an allowed additional pulse error $\eta>0$, a sufficient fixed-refresh condition is $2g_{\rm p}N_{\rm p}\delta\le\eta$ in addition to the ideal-jump bound. The condition $g_{\rm p}\delta\ll1$ for one pulse is not enough when $N_{\rm p}$ is large. With $h=T/K$ and $N_{\rm p}=O(T/h)$, suppressing the accumulated pulse term requires $g_{\rm p}T\delta/h\to0$; making it smaller than the ideal $O(g_*^2Th)$ term requires the stronger sufficient scaling $g_{\rm p}\delta/h=o(g_*^2h)$. These sufficient conditions quantify the pulse-width scaling needed to keep the accumulated correction controlled. The Poisson counterpart is in Appendix~\ref{app:poisson}.

\begin{corollary}[Finite-width optimum of the conservative fixed-refresh bound]\label{cor:finite-width-optimum}
Let $K=T/h$ and let $C_{\rm end}$ count boundary pulses. For at most $K+C_{\rm end}$ pulses,
\begin{equation}
 \epsilon^{\rm fixed,pulse}\le\min\{1,F(h)\},\quad
 F(h)=\frac32g_*^2Th+\frac{2g_{\rm p}T\delta}{h}
       +2g_{\rm p}C_{\rm end}\delta.
 \label{eq:fixed-pulse-envelope}
\end{equation}
It follows from Eqs.~\eqref{eq:fixed-finite-pulse-bound} and \eqref{eq:protocol-bounds}.
If $g_*,g_{\rm p}>0$ are independent of $h$, the continuous envelope has its interior minimum at
\begin{equation}
 h^*=\sqrt{\frac{4g_{\rm p}\delta}{3g_*^2}},\qquad
 F(h^*)=2\sqrt3\,Tg_*\sqrt{g_{\rm p}\delta}
          +2g_{\rm p}C_{\rm end}\delta.
 \label{eq:fixed-pulse-optimum}
\end{equation}
\end{corollary}
This gives the minimizing point of the sufficient envelope in its feasible continuous range. An admissible experiment also requires $h^*\le T$, pulses that fit within the intervals (typically $\delta\ll h^*$), and an integer $K$ chosen near $T/h^*$; if these conditions fail, the envelope is minimized on the feasible discrete set. At fixed $g_*,g_{\rm p},T$, the minimized envelope scales as $\sqrt\delta$ up to the endpoint term. The random-event optimization is discussed with the overlap model in Appendix~\ref{app:poisson}.

The preceding sections establish the structural and convergence results analytically; we now evaluate their operational consequences and coefficient-level behavior in finite-dimensional and Gaussian-bath models.

\section{Computational construction}
The calculations use two equivalent descriptions of the same complete environmental history. The fixed logical basis makes particle insertions simple, while the instantaneous eigenbasis diagonalizes the coupling during each frozen interval. Appendix~\ref{app:amplitudes} derives the fixed-basis coefficient, its moving-basis chain, and the bounded-model grid limit. Appendix~\ref{app:controlalg} develops fixed-refresh averages, Brownian group blocks and OU Hermite blocks; Appendix~\ref{app:poisson} gives the complete Poisson event-time average; Appendix~\ref{app:weak} gives analytic weak-coupling and static-frame checks. The finite-temperature Gaussian and Drude bath kernels are in Appendix~\ref{app:bath}. Appendix~\ref{app:physicality} constructs the continuous Choi image, recovers finite combs, and proves the interval-classicalization statement under its precise protocol assumptions.
\subsection{Finite experiment contraction}
For a trajectory $\omega$, we propagate the finite auxiliary state using the same Hamiltonian evolution that defines the cPT and insert each finite-duration instrument at its physical support. A final fixed POVM is contracted after the inverse-frame correction. The trajectory result is then averaged:
\begin{equation}
 p_{\rm comp}(x)=\E_\omega p_\omega(x).
\end{equation}
This is the numerical realization of the cPT--instrument contraction. Appendix~\ref{app:instrument} gives the operation-basis coefficients, the explicit all-sector expansion, and the measured boundary used in both qubit experiments.

\subsection{Particle amplitudes}
For a fixed pair of times $t_1,t_2$, the two-particle amplitude is obtained from the full time-ordered influence functional. For $H_{SE}=g\sigma_z\otimes B$,
\begin{equation}
 \mathcal C_{\nu_2\nu_1,\omega}^{(2)}(t_2,t_1)=
 \Tr_{SE}\!\left[
 \mathcal U_\omega(T,t_2)\mathcal P_{\nu_2}
 \mathcal U_\omega(t_2,t_1)\mathcal P_{\nu_1}
 \mathcal U_\omega(t_1,0)(\rho_S\otimes\rho_E)
 \right],
 \label{eq:amplitude}
\end{equation}
where $\mathcal U_\omega(t,s)[X]=V_\omega(t,s)X V_\omega^\dagger(t,s)$ is the joint $SE$ propagation for one complete control path and $\mathcal P_{(\mu,\mu')}[X]=(P_\mu\otimes\id_E)X(P_{\mu'}^\dagger\otimes\id_E)$ is the unscaled left/right Pauli insertion. The random-frame rotations and all cross-time environment propagators are contained in the same ordered chain. The environment is traced only after the complete history has been assembled. We compare
\begin{equation}
 \Delta\mathcal C_{\boldsymbol\nu}^{(n)}(\boldsymbol t)=
 \E_\omega[\mathcal C_{\boldsymbol\nu,\omega}^{(n)}(\boldsymbol t)]
 -\mathcal C_{\boldsymbol\nu,{\rm ideal}}^{(n)}(\boldsymbol t).
 \label{eq:amplitude-deviation}
\end{equation}
Fixed and Poisson refresh calculations use the exact holding boundaries and event times. For a finite environment qubit we propagate the joint system--environment unitary directly; for the Drude bath, Brownian and OU calculations use full-memory matrix-product compression on sampled continuous paths. Appendix~\ref{alg:particle-history} gives the step-by-step complex-amplitude contraction, Appendix~\ref{alg:control-paths} specifies each of the four control averages, and Appendix~\ref{alg:representation-blocks} gives the alternative SU(2) block procedure. The reported Drude calculations use full-memory contraction, not representation blocks. Errors are monitored through time-grid refinement, control substeps, bath truncation, MPS discarded weight, and finite trajectory sampling.

\section{Numerical results}
\label{sec:numerical}
\label{sec:numerics}
\subsection{Finite-duration logical experiments}
We test two logical tasks against a common environment-qubit construction, using $T=1.2$, $H_E=(1.7/2)\sigma_z^{(E)}$, and $g=0.8$ per coupled system qubit. The first task starts in $|0\rangle_S|+\rangle_E$ with $H_{SE}=g\sigma_z^{(S)}\otimes\sigma_x^{(E)}$. Its target instrument applies $H$, idles, and then applies $R_y(\pi/4)$ in three intervals of length $0.4$. A fixed logical $Z$ measurement has ideal probability $p_{0,\mathrm{id}}=(2-\sqrt2)/4$.

The second task prepares an entangled state of two system qubits $S_1S_2$ sharing one environment qubit. It starts in $|00\rangle_{S_1S_2}|+\rangle_E$ with $H_{SE}=g(\sigma_z^{(S_1)}+\sigma_z^{(S_2)})\otimes\sigma_x^{(E)}$. A duration-$0.6$ Hadamard on $S_1$ is followed by a duration-$0.6$ controlled-NOT (CNOT) gate from $S_1$ to $S_2$. The final logical two-outcome POVM has $E_0=|\Phi^+\rangle\langle\Phi^+|$, where $|\Phi^+\rangle=(|00\rangle+|11\rangle)/\sqrt2$, and $E_1=\id_{S_1S_2}-E_0$. Thus $p_{0,\mathrm{id}}=1$, while the uncompiled value is $0.52844$. The physical gate Hamiltonians are covariant throughout each window; the local compiling frames are independent on $S_1$ and $S_2$, although both qubits retain the same environment. Appendix~\ref{app:instrument} gives the generators and the $d_S=4$ version of the instrument contraction.

Figure~\ref{fig:operational} compares fixed refresh, bounded Cayley paths, Brownian driving, and OU driving. The upper row shows the single-qubit $p_0$ deviation; the lower row shows the Bell-projection deviation $|\hat p_0-1|$. Fixed and Cayley values for the upper task use exact frame-ensemble contractions. The Bell and continuous-drive curves average 500 complete trajectories per point. At 768 fixed intervals, the separate Bell scan gives $0.001281\pm0.000059$ (one sampling standard error); the node-matched fixed and Cayley values are $0.001299$ and $0.001532$. Doubling the Cayley edge substeps from 8 to 16 changes the Bell infidelity by $2.95\times10^{-5}$ at $K=24$ and $9.95\times10^{-7}$ at $K=768$. The strongest Brownian and OU Bell points give $0.01822\pm0.00043$ and $0.01285\pm0.00029$, respectively, after the grid refinement in Appendix~\ref{app:numerics}. The 768-interval fixed-refresh point in the upper row has exact-ensemble deviation $6.36\times10^{-4}$. The Poisson scan is shown separately in Fig.~\ref{fig:poisson-operational}. These experiments evaluate selected logical probabilities on the respective control-parameter axes.

\begin{figure}[t]
\includegraphics[width=\linewidth]{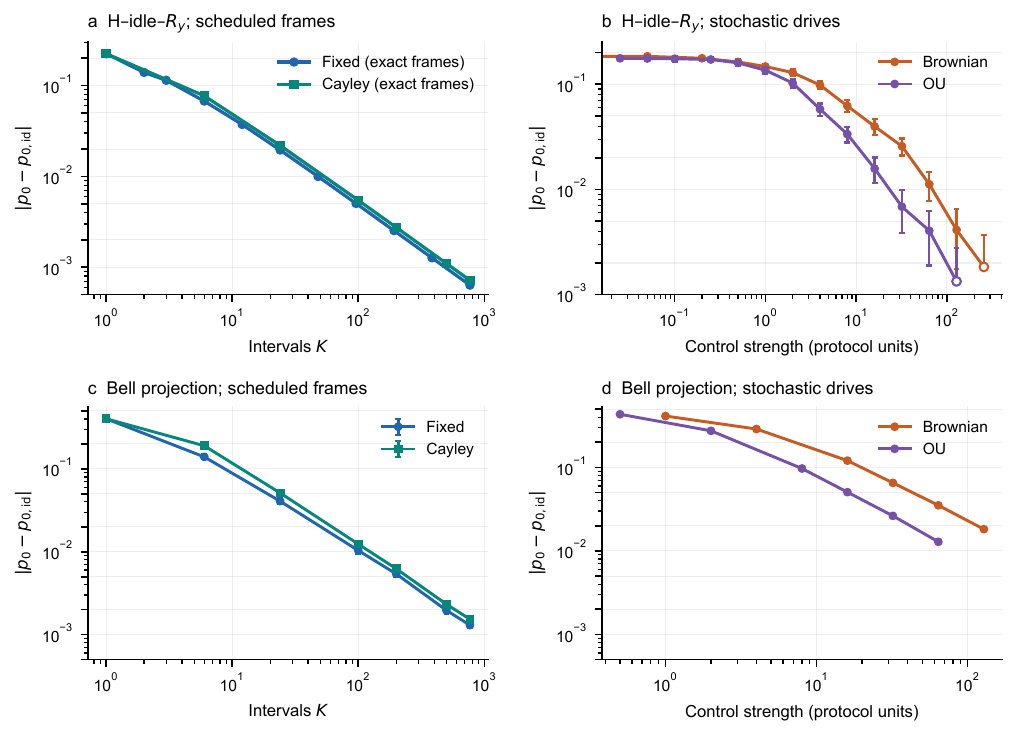}
\caption{Finite-duration logical recovery. Top: $|p_0-p_{0,\mathrm{id}}|$ for the $H$--idle--$R_y$ task. Bottom: Bell-projection infidelity with a shared environment qubit. Left: fixed refresh and its bounded Cayley-path realization at scheduled nodes. Right: Brownian and OU scans on their respective control-strength axes. The upper-left curves use exact frame-ensemble contractions; all other points are means of 500 trajectories with pointwise 95\% sampling intervals at the reported grids. Open markers in the upper-right panel indicate that the interval for the signed probability difference includes zero; only its upper side can be shown on the logarithmic axis. Grid and control-substep comparisons are summarized in Appendix~\ref{app:numerics}. The separate Poisson scan is in Fig.~\ref{fig:poisson-operational}.}
\label{fig:operational}
\end{figure}

\subsection{CMI and multi-time structure}
\label{sec:cmi}
Following the continuous-process construction of Ref.~\cite{wassnerCPT}, we cut the cMPS at time $t$ into past ($P$), present system interfaces ($S$), and future ($F$), and evaluate the operator CMI
\begin{equation}
 \mathcal N_{\rm op}(t)=S(P S)+S(S F)-S(S)-S(P S F),
 \label{eq:operator-cmi}
\end{equation}
with logarithms to base 2. This quantity measures conditional past--future structure carried by the continuous process representation and can contain both quantum and classically mediated correlations.

We report two averaging orders,
\begin{equation}
 \mathcal N_{\rm traj}(t)=\E_\omega[\mathcal N_{\rm op}(\cPT_\omega;t)],\qquad
 \mathcal N_{\rm avg}(t)=\mathcal N_{\rm op}(\E_\omega[\cPT_\omega];t).
 \label{eq:cmi-orders}
\end{equation}
The first evaluates the diagnostic for each logged control path before averaging; the second averages the cPT amplitudes in a common field basis before applying the nonlinear entropy functional. The two quantities therefore describe path-resolved and path-unresolved process ensembles, respectively. The physical-interface contraction, the two averaging procedures, and the entropy reductions are given in Appendix~\ref{app:cmi}. Repeating the contraction after a central Schmidt coordinate transformation gives the same CMI, providing a coordinate-invariance check.

For each protocol and control strength we generated 500 complete trajectories and nine cut positions. The strongest fixed-refresh point gives $\mathcal N_{\rm traj}=5.46\times10^{-6}$ bits and $\mathcal N_{\rm avg}=1.75\times10^{-8}$ bits at $t/T=1/2$ on the 80-block, rank-128 computation. Doubling the blocks changes these values by $2.62\times10^{-7}$ and $8.22\times10^{-10}$ bit, respectively; increasing the common field rank to 192 changes the latter by $1.46\times10^{-10}$ bit. Table~\ref{tab:cmi-resolution} reports grid, rank, and sampling checks for each strong endpoint. Both diagnostics decrease across the sampled Brownian and OU settings in Fig.~\ref{fig:cmi-main}. The grid, rank, and sampling checks summarized in Appendix~\ref{app:numerics} characterize the numerical resolution of the reported values.

The dependence on the cut position is shown separately in Fig.~\ref{fig:cmi-cuts-main}.

\begin{figure}[t]
\includegraphics[width=\linewidth]{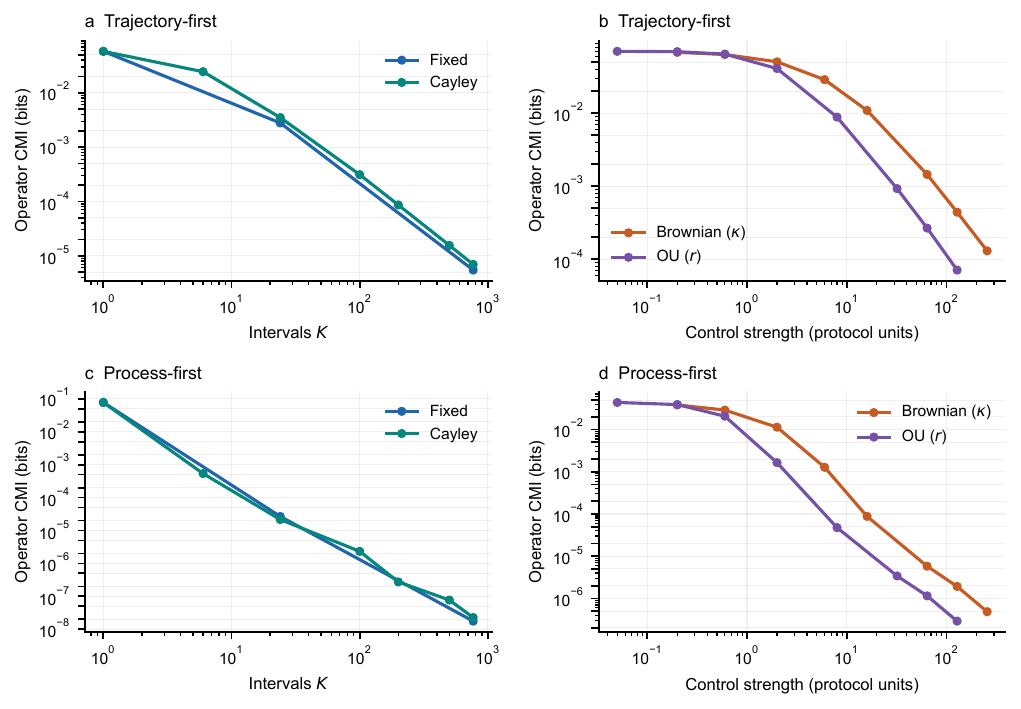}
\caption{Midpoint operator CMI under trajectory-first (top) and process-first (bottom) averaging. The left column compares fixed refresh and full-interval Cayley paths; the right column compares Brownian diffusion strength $\kappa$ and OU time-acceleration factor $r$. The Brownian and OU curves each contain eight sampled settings; the fixed-refresh and Cayley curves use their respective prescribed-node grids. Each point uses 500 complete paths, 80 time blocks, and field rank 128. Lines connect sampled settings; resolution checks for the smallest values are given in Appendix~\ref{app:numerics}. The operator CMI uses the physical interface defined in the text.}
\label{fig:cmi-main}
\end{figure}

\begin{figure}[t]
\includegraphics[width=\linewidth]{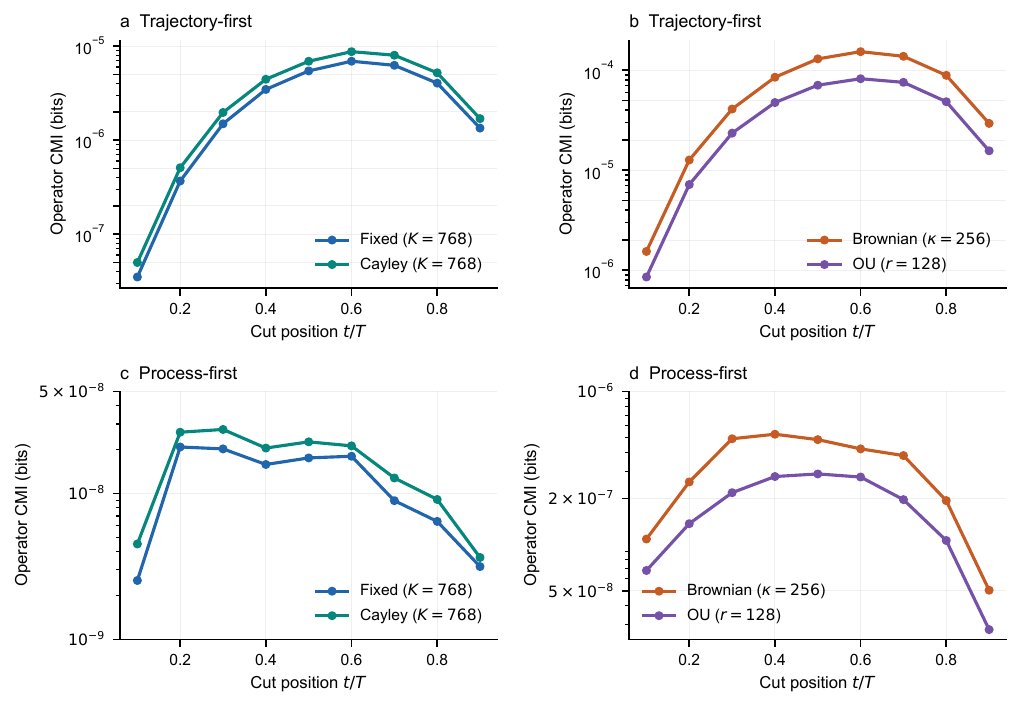}
\caption{Operator CMI across nine cut positions at fixed and Cayley $K=768$, Brownian $\kappa=256$, and OU $r=128$. The upper row averages the diagnostic over logged trajectories; the lower row evaluates it after averaging the cPT over unrecorded paths. These settings do not represent matched control costs. The Poisson cut scan is in Fig.~\ref{fig:poisson-cmi}. The same physical interface is used throughout.}
\label{fig:cmi-cuts-main}
\end{figure}

\subsection{Two-particle amplitudes with a qubit environment}
The finite environment model allows a direct test of all operation-label coefficients at a prescribed pair of insertion times in the two-particle sector, without a bath-kernel expansion or history compression. We use one system qubit and the same environment qubit as in the first probability experiment: $T=1.2$, $H_E=(1.7/2)\sigma_z^{(E)}$, $H_{SE}=0.8\sigma_z^{(S)}\otimes\sigma_x^{(E)}$, and $\rho_{SE}(0)=|0\rangle\langle0|_S\otimes|+\rangle\langle+|_E$. No instrument Hamiltonian acts in this coefficient calculation. At $t_1=T/3$ and $t_2=2T/3$, we evaluate all $15^2=225$ ordered pairs of nonvacuum Pauli left/right labels from Eq.~\eqref{eq:amplitude}. For each label pair, the complex amplitudes of 500 complete control paths are averaged before subtraction of the $g=0$ coefficient. The mean absolute and root-mean-square (RMS) deviations are defined in Eq.~\eqref{eq:app-numeric-metrics}.

The four main-text scans and strong-point heat maps are in Figs.~\ref{fig:qubit-particles-main} and \ref{fig:qubit-heatmaps-main}. At the strongest displayed points, the mean absolute deviations are $2.75\times10^{-4}$ for 768 fixed intervals, $4.82\times10^{-4}$ for 768 full-interval Cayley paths, $5.73\times10^{-3}$ for Brownian strength $\kappa=128$, and $3.65\times10^{-3}$ for OU speed $r=64$. Each protocol is shown against its own control parameter. Heat-map color denotes the modulus of a complex coefficient deviation, not a probability. The environment retains its state through both insertion intervals. Fixed events use their actual times, Cayley paths use finite edge substeps, and Brownian and OU paths use a control mesh. Appendix~\ref{app:numerics} reports the 225-label contraction check and paired strong-end grid comparisons. The Poisson scan is in Fig.~\ref{fig:poisson-qubit-particles}.

\begin{figure}[t]
\includegraphics[width=\linewidth]{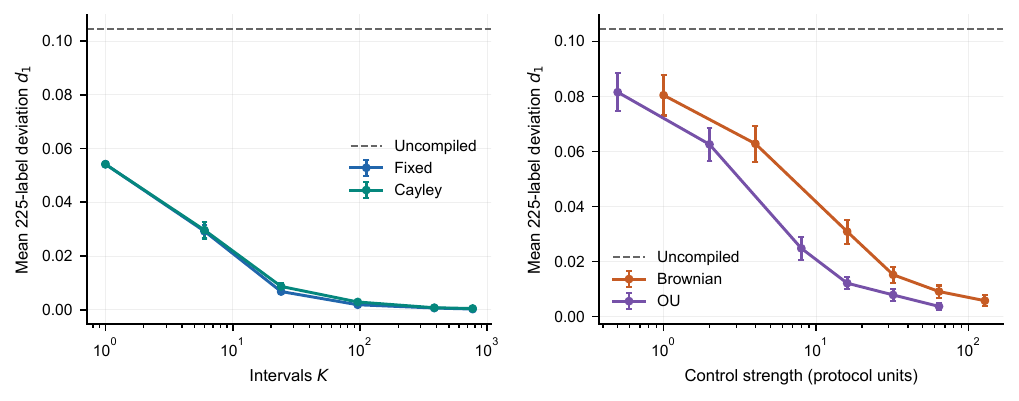}
\caption{Qubit-environment two-particle coefficients at the specified insertion-time pair $(t_1,t_2)=(T/3,2T/3)$. The ordinate is the mean of $|\mathbb E_\omega C^{(2)}_\omega-C^{(2)}_{g=0}|$ over all 225 operation-label pairs at those times; this is not a numerical integration over the continuous two-particle sector. Fixed and bounded Cayley paths share the interval-count axis; Brownian and OU use their respective strength parameters. The dashed line is the deterministic uncompiled $U(t)=\id$, $g=0.8$ baseline. Error bars are pointwise 95\% bootstrap sampling intervals from 500 paths at the reported control grids; paired-grid comparisons are given in Appendix~\ref{app:numerics}.}
\label{fig:qubit-particles-main}
\label{fig:qubit-particles}
\end{figure}

\begin{figure}[t]
\includegraphics[width=\linewidth]{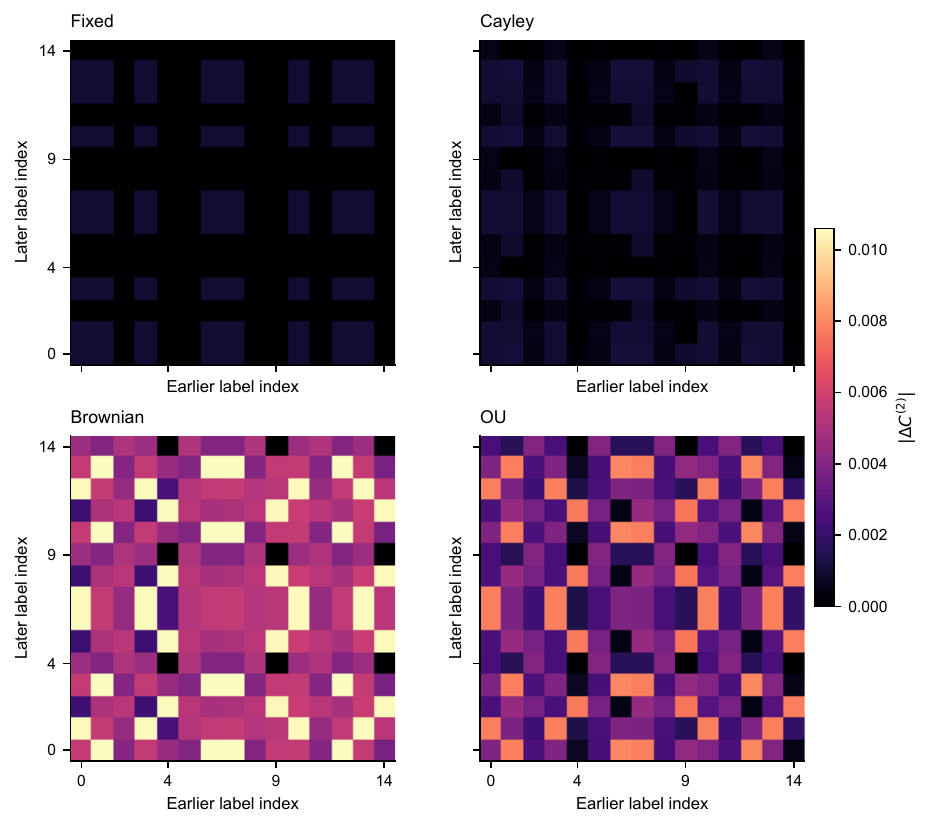}
\caption{Absolute complex-amplitude deviations over all 225 operation-label pairs at $(t_1,t_2)=(T/3,2T/3)$ in the qubit-environment two-particle sector. The displayed table is complete in its operation labels at this time pair, not over the continuous insertion-time domain. Panels show fixed refresh ($K=768$), full-interval Cayley paths ($K=768$), Brownian ($\kappa=128$), and OU ($r=64$), in reading order, on a common color scale. Earlier labels run horizontally and later labels vertically.}
\label{fig:qubit-heatmaps-main}
\end{figure}
For a genuine uncompiled comparator we also set $U(t)=\id$ at every time while retaining $g=0.8$. Its 225-label mean absolute and RMS deviations from the zero-coupling table are $0.10449$ and $0.19591$, respectively. This deterministic identity-frame result is distinct from the one-interval random-frame point, whose mean absolute deviation is $0.05418$. The dashed baseline in Fig.~\ref{fig:qubit-particles-main} gives this uncompiled comparator; the control points below it therefore show improvement relative to an actual uncompiled process, not only proximity to the ideal.

\subsection{Two-particle amplitudes with a Drude bath}
The complementary calculation uses the uncut Drude Gaussian bath of Eq.~\eqref{eq:app-drude} with reorganization energy $\Lambda=1$ and bath parameter $\gamma_B=1$, $g=1.5$, inverse temperature $\beta=5$, $\rho_S=(\id+0.6\sigma_x)/2$, $T=1$, and $(t_1,t_2)=(T/3,2T/3)$; no additional system Hamiltonian is included. We enumerate all 225 two-particle label pairs and average each complex coefficient over 500 trajectories before comparison with its uncoupled ideal. Figure~\ref{fig:drude-particles-main} shows the scans, and Fig.~\ref{fig:drude-heatmaps-main} resolves representative strong-point labels. The strongest mean absolute deviations are $0.003477$ for fixed refresh (768 intervals), $0.003997$ for full-interval Cayley paths ($K=768$, with the displayed system-leg resolution specified for reproducibility), $0.0234$ for Brownian driving ($\kappa=128$), and $0.0725$ for OU driving ($\kappa_{\rm eff}=128$). The Poisson scan is in Fig.~\ref{fig:poisson-drude-particles} of Appendix~\ref{app:poisson}. This continuum bath has finite-temperature correlations and is not replaced by the finite-dimensional qubit calculation. The Drude OU scan fixes $\tau=0.1$ and varies the integrated variance, rather than using the time-accelerated OU family of the qubit calculations.

\begin{figure}[t]
\includegraphics[width=\linewidth]{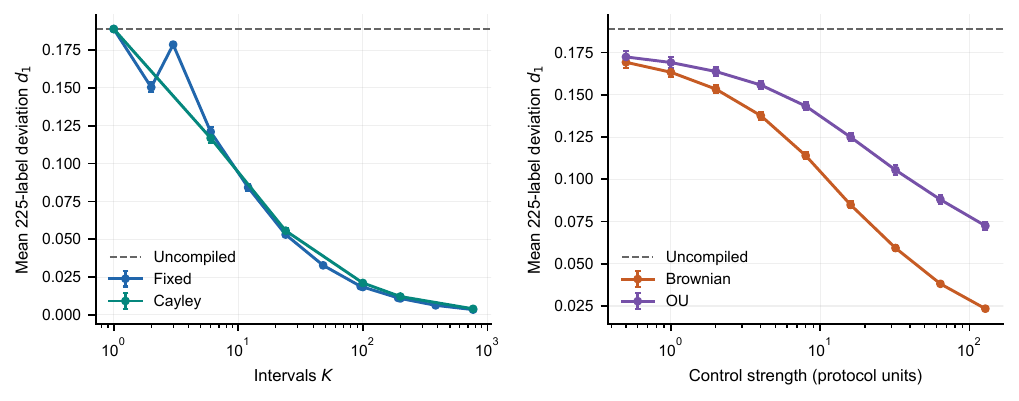}
\caption{Drude-bath two-particle mean absolute deviations over all 225 complex operation-label coefficients at $(t_1,t_2)=(T/3,2T/3)$. This is a complete label table at the specified times, not a numerical integration over the continuous two-particle sector. Fixed/Cayley use prescribed-node counts; Brownian and OU use different strength parameters. The dashed line is the deterministic uncompiled $U(t)=\id$, $g=1.5$ baseline. Error bars are pointwise 95\% bootstrap sampling intervals for 500 paths. The Cayley $K=200$ and $768$ points use two system-leg quadrature levels; the displayed resolution is specified for reproducibility, while bootstrap bars quantify trajectory sampling.}
\label{fig:drude-particles-main}
\end{figure}

\begin{figure}[t]
\includegraphics[width=\linewidth]{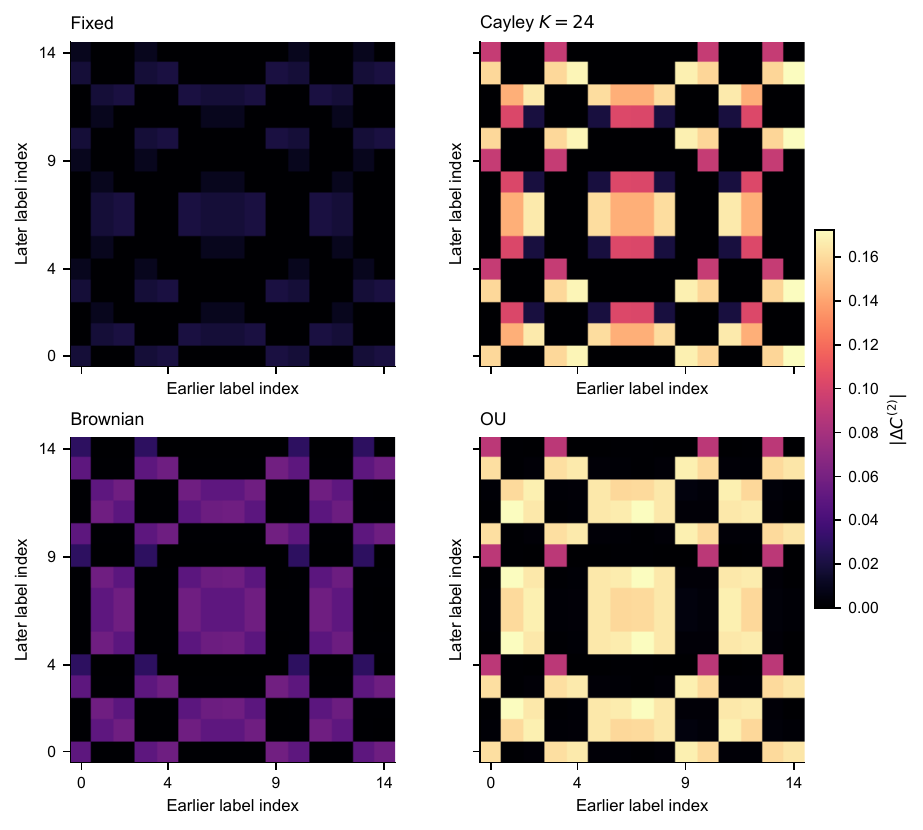}
\caption{Label-resolved Drude-bath two-particle deviations at representative controls: fixed refresh ($K=384$), full-interval Cayley ($K=24$), Brownian ($\kappa=128$), and OU ($\kappa_{\rm eff}=128$). The Cayley panel uses system-leg resolution $P=4$. Color denotes $|\Delta C^{(2)}|$ on one common scale.}
\label{fig:drude-heatmaps-main}
\end{figure}
With $U(t)=\id$ and the same nonzero coupling, the complete 225-label Drude table has mean absolute deviation $0.18893$ and RMS deviation $0.32277$ from the zero-coupling table. In this centered single-axis Gaussian model, changing the sign of the \emph{entire} static coupling history leaves the influence functional invariant, so the one-interval random-Pauli result happens to coincide with this deterministic baseline coefficient by coefficient. We verified that equality directly rather than treating zero refresh as an uncompiled protocol in general. The dashed line in Fig.~\ref{fig:drude-particles-main} gives the deterministic uncompiled value.

\section{Summary and outlook}
\subsection{Methodological position and main conclusions}
Process tensors, or quantum combs, provide the standard operational framework for multi-time quantum processes and for temporal correlations generated by environmental memory \cite{combref009,combref017,combref019,combref029}. Existing work has developed process tomography, dynamical characterization, and efficient numerical simulation, while randomized compiling has mainly been formulated for discrete gate or cycle boundaries. This paper establishes continuous-process randomized compiling (CPRC) within the continuous process tensor (cPT) framework, extending randomized compilation to finite-duration instruments on a continuous physical time axis.

Discrete gate-level randomized compiling uses random unitary frame changes at logical-gate boundaries to tailor coherent gate errors into stochastic Pauli errors \cite{wallman2016,hashim2021}. Random dynamical decoupling addresses the same control problem at the Hamiltonian level and provides the fixed-refresh scaling for bounded interactions \cite{viola2005}; bounded continuous paths and protected-gate constructions show how intended operations can be preserved under finite-strength control \cite{violaKnill2003,chen2006,khodjasteh2005}. CPRC combines these ideas with a pointwise-covariant implementation rule for finite-duration instruments. Consequently, in the absence of system--environment coupling, every random control trajectory reproduces the target logical experiment, while a closed control-frame path shows why endpoint compensation alone does not specify a correct continuous implementation.

Within the cPT description, the Pauli character average diagonalizes error histories at the designated boundaries of independent complete intervals. The exactly solvable two-cell static-bath benchmark shows that the resulting classical labels can still carry environment-induced correlations. For bounded centered interactions and controls satisfying the stated mixing hypotheses, we derive a total-variation bound for complete adaptive output records. The bound applies to fixed Pauli refresh, bounded Cayley paths, unitary Brownian motion, and OU driving through their respective mixing envelopes, with Poisson refresh treated as an event-time extension. These two conclusions concern different observables: the former controls complete adaptive records directly, whereas the latter controls the trajectory-averaged process vector in the cPT representation. We also establish Fock-norm convergence of the entire trajectory-averaged, trace-closed cPT on a fixed finite time interval for finite-dimensional non-Gaussian environments and Gaussian baths satisfying the stated physical-cutoff realization, time-integrable covariance, and cutoff-stability conditions, including the uncut Drude spectrum. The proof combines fixed-sector convergence with a uniform summable particle-sector tail. Qubit-environment and thermal-Drude calculations evaluate finite-duration logical probabilities, operator conditional mutual information, and all two-particle operation-label coefficients at a prescribed pair of insertion times; the reported observables decrease along the sampled control-strength scans in agreement with the theoretical picture.

CPRC therefore retains the multi-time environmental history of a cPT while providing control protocols and verification benchmarks for shaping its error structure. Its main contribution is the joint treatment of pathwise logical implementation, boundary-restricted classicalization, and complete-record error control on one continuous time axis.

\subsection{Scope and outlook}
The results apply under the following conditions, which specify the physical regimes addressed by the theory:
\begin{enumerate}
\item Exact classicalization is formulated for the boundaries of independent, whole-interval Pauli conjugations and assumes the factorized, frame-independent initial-state condition used in the proof. Brownian and OU frames retain temporal correlations at finite separations, so their classicalization is described by the corresponding continuous-path process.
\item The complete-record total-variation bound assumes a bounded, centered system--environment interaction and an autonomous control process with the stated mixing envelope. For OU driving, the relevant condition concerns the integrated mixing of the joint frame--driving state, as specified in Appendix~\ref{app:fixed-sector-general}.
\item Theorem~\ref{thm:fock-convergence} gives Fock-norm convergence of the entire trajectory-averaged, trace-closed cPT for a finite-dimensional system, a fixed finite time interval, and a fixed operation basis. It covers finite-dimensional non-Gaussian environments and Gaussian baths under the stated physical-cutoff realization, $L^1$ covariance, and cutoff-stability conditions, including the uncut Drude model. Its operational consequence applies to continuous instruments with a common finite induced-norm bound; it does not by itself imply infinite-field-space entropy continuity or replace the bounded-interaction complete-record theorem.
\item Finite-width control pulses contribute an additional error budget. Increasing the nominal refresh rate is useful when the accumulated pulse correction remains controlled, as quantified by Eqs.~\eqref{eq:fixed-finite-pulse-bound} and \eqref{eq:poisson-finite-pulse-bound}.
\end{enumerate}

Two directions are especially natural. Experimentally, programmable pulse envelopes and multi-time readout on superconducting, trapped-ion, or solid-state-spin platforms can test the continuous compilation of finite-duration instruments. Theoretically, one can seek energy-constrained complete-record bounds for unbounded bosonic fields and study cascaded error suppression obtained by combining continuous compilation with quantum error-correction codes.

\appendix
\input{appendix_amplitudes}
\input{appendix_instrument}
\input{appendix_bath}
\input{appendix_control}
\input{appendix_weak}
\input{appendix_physicality}
\input{appendix_bounds}
\input{appendix_fixed_sector}
\input{appendix_cayley}
\input{appendix_poisson}
\input{appendix_cmi}
\input{appendix_numerics}

\bibliographystyle{apsrev4-2}
\bibliography{references}
\end{document}

%% file: background.tex
\section{Background: from discrete process tensors to continuous combs}
\label{sec:background}

\subsection{Discrete process tensors and their physical constraints}
Let $S$ be a system of dimension $d_S$, $E$ an environment, and
$0=t_0<t_1<\cdots<t_k=T$ a finite set of laboratory times. An intervention
$\mathcal A_j^{(x_j)}$ is a completely positive (CP), trace-nonincreasing map
associated with outcome $x_j$; summing over the outcomes gives a
trace-preserving map. With $\mathcal U_{j+1:j}$ denoting the joint $SE$
propagation between two interventions, a process tensor maps a chosen
sequence of operations to the resulting system state,
\begin{equation}
 \Upsilon_{k:0}[\mathcal A_{k-1:0}^{(\boldsymbol x)}]
 =\Tr_E\!\left[
 \mathcal U_{k:k-1}\circ\mathcal A_{k-1}^{(x_{k-1})}
 \circ\cdots\circ\mathcal U_{1:0}\circ\mathcal A_0^{(x_0)}
 [\rho_{SE}(0)]\right].
 \label{eq:background-discrete-map}
\end{equation}
Each $\mathcal A_j^{(x_j)}$ acts as the identity on $E$. The same environment
passes through every factor, so Eq.~\eqref{eq:background-discrete-map} does not
assume independent channels or an initially factorized $SE$ state. A final
positive operator-valued measure (POVM) $\{M_y\}$ gives the joint probability
$p(\boldsymbol x,y)=\Tr_S[M_y\Upsilon_{k:0}
[\mathcal A_{k-1:0}^{(\boldsymbol x)}]]$. This is the operational content
that must be preserved when a discrete comb is represented by a Choi operator
\cite{combref009,combref017,combref019}.

If $\widehat{\mathbf A}^{(\boldsymbol x,y)}$ denotes the Choi operator of
the complete tester, including the final effect, the same probability is
the generalized Born-rule contraction
\begin{equation}
 p(\boldsymbol x,y)
 =\Tr\!\left[\widehat\Upsilon_{k:0}
 \bigl(\widehat{\mathbf A}^{(\boldsymbol x,y)}\bigr)^{\mathsf T}\right].
 \label{eq:background-born-rule}
\end{equation}
The transpose is fixed by the Choi convention used here; changing the
vectorization convention changes its placement, not the probability
\cite{combref009,wassnerCPT}.

The process operator $\widehat\Upsilon_{k:0}$ in this contraction is the
open-output comb: its final system output is retained until the tester's
terminal effect is applied. If that output is closed with the deterministic
trace effect, the resulting operator is a trace-closed tester. Its
normalization is inherited from the open comb, but it has no terminal output
leg on which to impose the open-comb recursion.

In a fixed, unnormalized Choi convention, write the operator for the first
$j$ slots as $\widehat\Upsilon_{j:0}$. A physical comb obeys
\begin{equation}
 \widehat\Upsilon_{j:0}\succeq0,\qquad
 \Tr_{o_j}\widehat\Upsilon_{j:0}
 =\id_{i_j}\otimes\widehat\Upsilon_{j-1:0},
 \quad j=1,\ldots,k.
 \label{eq:background-comb-constraints}
\end{equation}
Here $i_j$ is the process input returned by $\mathcal A_{j-1}$ at
$t_{j-1}$, and $o_j$ is the process output at $t_j$. The intervention
$\mathcal A_{j-1}$ thus maps $o_{j-1}$ to $i_j$. The open comb acts on the
terminal output $o_j$ and the operation pairs $(i_r,o_{r-1})$,
$r=1,\ldots,j$. The initial boundary is
$\widehat\Upsilon_{0:0}=\Tr_E\rho_{SE}(0)$ with unit trace.
Tracing $o_j$ leaves the identity on $i_j$ and the shorter open comb
whose terminal output is $o_{j-1}$.
Positivity makes contractions with CP instruments
nonnegative. The trace recursion is the no-future-to-past condition: a
deterministic operation at a later slot cannot change earlier outcome
probabilities. A coarser comb is recovered by linking an intermediate slot
with the identity channel. Filling a slot by a
measurement and forgetting its outcome is generally a different operation.
Thus a process has multiple operational marginals, and state partial traces
alone do not specify which intervention was performed
\cite{combref009,combref019,combref029,combref038}.

These conditions constrain the \emph{whole} Choi operator, not each of its
matrix elements separately. They also distinguish temporal memory from a
failure of complete positivity or causality. A correlated environment can
produce a positive, causal comb whose later output still depends on earlier
interventions. Finite-slot combs describe this dependence exactly at their
chosen slots, but those slots do not specify how an instrument acts during a
nonzero time interval \cite{combref009,combref017,wassnerCPT}. This loss is
concrete: two bounded joint Hamiltonian paths can have the same identity
propagator, and hence the same comb at their chosen endpoints, while the same
Hamiltonian instrument acting between those endpoints gives different output
probabilities. The closed-loop construction following Proposition~\ref{prop:transparency}
gives an explicit qubit example without an independent $H_S$. It establishes
noninjectivity of that endpoint coarse graining, not an inequivalence of complete cPTs.

\subsection{Continuous matrix product states}
A cMPS provides a finite-auxiliary-space
description of a bosonic field on $[0,T]$. For $q$ field species, let
$\psi_\nu^\dagger(t)$ create a species-$\nu$ excitation and let
$|\Omega\rangle$ be the vacuum. With auxiliary matrices $Q(t)$,
$R_\nu(t)$ and a boundary matrix $B$, its defining form is
\begin{equation}
 |\Phi[Q,R,B]\rangle=
 \Tr_{\rm aux}\!\left[B\,\mathcal T
 \exp\!\left(\int_0^T\!dt\,
 \left[Q(t)\otimes\id+
 \sum_{\nu=1}^{q}R_\nu(t)\otimes\psi_\nu^\dagger(t)\right]
 \right)\right]|\Omega\rangle.
 \label{eq:background-cmps}
\end{equation}
The auxiliary propagator
$V(t,s)=\mathcal T\exp[\int_s^t Q(u)\,du]$ makes the ordered
$n$-particle coefficient explicit: for $0<t_1<\cdots<t_n<T$ it is
\begin{equation}
 \phi_{\boldsymbol\nu}^{(n)}(\boldsymbol t)
 =\Tr_{\rm aux}\!\left[
 B V(T,t_n)R_{\nu_n}(t_n)V(t_n,t_{n-1})\cdots
 R_{\nu_1}(t_1)V(t_1,0)\right].
 \label{eq:background-cmps-amplitude}
\end{equation}
Thus $Q$ propagates the auxiliary memory between insertions and $R_\nu$
records an insertion. These are representation matrices; $Q$ is not the
stochastic control generator introduced later. In a lattice approximation of
spacing $\epsilon$, the standard scaling
$Q_j=\epsilon Q(j\epsilon)$ and
$R_{\nu,j}=\sqrt{\epsilon}R_\nu(j\epsilon)$ yields a finite particle
density in the continuum limit \cite{combref024,combref025,combref039}.

A cMPS is a compact representation of a field vector, not automatically a
physical quantum process. The extra process constraints are positivity and
causality. In the construction used below, the auxiliary space contains the
system Liouville space and the environment memory, while the field labels
record possible system interventions \cite{wassnerCPT}.

\subsection{Operational continuum limit and finite-duration interventions}
The continuum construction of Ref.~\cite{wassnerCPT} singles out the identity
operation as the process vacuum. An informationally complete basis of system
superoperators has $d_S^4$ elements; after removing its identity element,
there are $q=d_S^4-1$ nonvacuum field labels $\nu=(\mu,\mu')$. The
corresponding cPT has the expansion
\begin{equation}
 |\cPT_T\rangle=\sum_{n=0}^{\infty}\sum_{\boldsymbol\nu}
 \int_{0<t_1<\cdots<t_n<T}\!
 d\boldsymbol t\; \mathcal C_{\boldsymbol\nu}^{(n)}(\boldsymbol t)
 |\boldsymbol\nu;\boldsymbol t\rangle .
 \label{eq:cpt-expansion}
\end{equation}
Here $n$ counts nontrivial process-field insertions; it is neither the number
of physical particles in $S$ nor the number of instruments or numerical time
steps. A zero insertion at a time means doing nothing there. Accordingly,
identity labels $(0,0)$ recover lower-time marginals. For a trace-closed,
normalized process the zero-particle amplitude is one; after replacing the
trace boundary by a selected POVM outcome, the corresponding vacuum
coefficient need not be one. Individual higher-particle amplitudes are
complex response coefficients, not probabilities
\cite{wassnerCPT,combref038}.

The advantage of Eq.~\eqref{eq:cpt-expansion} is operational rather than a
claim that every mathematical limit of discrete Choi states exists. For a
sufficiently regular finite-duration instrument, a short slice has a local
generator, schematically
$\mathcal A_{[t,t+\epsilon]}=\mathrm{id}+
\epsilon\mathcal G(t)+O(\!\epsilon^2)$ on a sufficiently regular interval.
Its field functional contracts with the cPT across the \emph{whole} support
of the instrument. The resulting continuous Born rule can describe a
measurement or control that acts while the system and environment evolve;
it is not restricted to instantaneous gates at prescribed slots. A finite
auxiliary-space contraction can sum all particle sectors without truncating
Eq.~\eqref{eq:cpt-expansion} one order at a time
\cite{wassnerCPT}.

The regularity assumption matters. In a naive refinement of an ordinary
Choi-state construction, a fresh maximally entangled pair is introduced at
every interval of length $\epsilon$. Taking $\epsilon\to0$ while retaining
arbitrary finite operations at every slot does not give the finite-excitation
continuum limit above. Assigning the do-nothing operation to the vacuum and
scaling instrument generators with their action time instead keeps the
instrument's field excitations finite under the conditions of the continuum
construction. This is the sense in which the cPT avoids the formal
``infinite-energy'' or ultraviolet problem of the naive slot-by-slot limit;
point-time insertions used to recover ideal finite combs may still be
distributional, rather than finite-energy laboratory operations. The cPT
construction does not by itself regularize an unbounded bath Hamiltonian or
an unbounded physical control pulse \cite{wassnerCPT}.

A second pitfall concerns temporal correlations. For bounded joint
Hamiltonians on a fixed interval, Ref.~\cite{wassnerCPT} proves that a
particular relative-entropy non-Markovianity measure of the \emph{normalized
Choi states} of successively finer discrete combs tends to zero. Its proof
bounds each short-step entropy contribution by
$O(\epsilon^2\log(1/\epsilon))$ and sums $O(T/\epsilon)$ steps, producing a
vanishing bound. This is a Zeno-type artifact of that dense Choi-state
construction and that measure, not a theorem that the physical environment
loses its memory. The cPT preserves the identity-operation coarse-graining
and admits correlation measures conditioned on the system at a chosen cut.
It therefore offers a route to studying temporal correlations without
identifying the zero of the naive dense-slot measure with memory loss.

For the compiled processes studied here, physicality is checked on the
continuous Choi image, rather than on signs of the complex coefficients in
Eq.~\eqref{eq:cpt-expansion}. Its open-output finite combs must be positive
and satisfy Eq.~\eqref{eq:background-comb-constraints}; the corresponding
closed testers are positive and normalized on deterministic operations.
Appendix~\ref{app:physicality}
establishes these properties for each control trajectory and its average,
including compatibility with finite-comb recovery.

%% file: appendix_amplitudes.tex
\section{From the fixed-basis formula to a complete moving-basis history}
\label{app:amplitudes}

\subsection{The closed particle amplitude and its boundaries}
Choose a system operator basis $P_0=\id_S,P_1,\ldots,P_{d^2-1}$ with $\Tr(P_\mu^\dagger P_\nu)=d\delta_{\mu\nu}$. The identity insertion is the identity superoperator; a nonvacuum label $\nu=(\mu,\mu')\ne(0,0)$ inserts
\begin{equation}
 \mathcal P_\nu[X]=(P_\mu\otimes\id_E)X(P_{\mu'}^\dagger\otimes\id_E).
 \label{eq:app-insertion}
\end{equation}
This linear operation basis is not a set of physical instruments: an individual off-diagonal $\mathcal P_\nu$ need not be completely positive. If $\mathcal U_\omega(t,s)$ is joint conjugation by the full $SE$ propagator, set $t_0=0$, $t_{n+1}=T$, and $\mathcal U_{\omega,j}=\mathcal U_\omega(t_{j+1},t_j)$ for $0\le j\le n$. For $0<t_1<\cdots<t_n<T$ the trace-boundary coefficient is
\begin{equation}
 \mathcal C^{(n),\mathrm{tr}}_{\boldsymbol\nu,\omega}(\boldsymbol t)
 =\Tr_{SE}\!\left[\overleftarrow{\prod}_{j=1}^{n}(\mathcal U_{\omega,j}\mathcal P_{\nu_j})
 \mathcal U_{\omega,0}[\rho_S\otimes\rho_E]\right].
 \label{eq:app-general-amplitude}
\end{equation}
The ordered product places the factor with the largest $j$ on the left.
Here $n$ counts operational field insertions. It is unrelated to the number of physical particles, the number of interventions in a selected experiment, and the number of numerical time slices. For $n=0$, Eq.~\eqref{eq:app-general-amplitude} is defined as $\Tr_{SE}\mathcal U_\omega(T,0)[\rho_S\otimes\rho_E]=1$. If every insertion is $\mathcal P_{(a,a)}$ with unitary $P_a$, the trace is also one, since each interval propagation and each such insertion preserves trace. No analogous statement holds for general unequal left and right labels. The ideal coefficient is precisely Eq.~\eqref{eq:app-general-amplitude} with $H_{SE}=0$ and the same logical input, instrument coordinates, and trace boundary. This fixes the subtraction in Eq.~\eqref{eq:amplitude-deviation} without setting every ideal nonvacuum coefficient to zero.

For a fixed coupling $gS\otimes B$ with $S|s\rangle=\lambda_s|s\rangle$, insert $\sum_s|s\rangle\langle s|$ between propagators on each Keldysh branch. On an interval $[u_k,u_{k+1})$ the conditional environmental propagator is
\begin{equation}
 V_{s,k}=\mathcal T\exp\!\left[-i\int_{u_k}^{u_{k+1}}
 (H_E+g\lambda_s B)dt\right].
\end{equation}
For two complete eigenvalue histories $\boldsymbol s,\boldsymbol r$, the environmental factor is one trace over the full history,
\begin{equation}
 \mathcal F_{B,\Pi}[\boldsymbol s,\boldsymbol r]
 =\Tr_E[V_{s,N-1}\cdots V_{s,0}\rho_E
 V_{r,0}^\dagger\cdots V_{r,N-1}^\dagger].
 \label{eq:app-full-bath-trace}
\end{equation}
This is the environmental part of the fixed-basis construction leading to Eq.~(76) of Ref.~\cite{wassnerCPT}. The factor is evaluated as one trace over the complete history; factorizing it into interval traces would reset the bath at every grid point.

For one insertion at $t_1$, expand $\rho_S$ in the fixed eigenbasis, with $\rho_{s_0r_0}=\langle s_0|\rho_S|r_0\rangle$. The first interval carries $|s_0\rangle\langle r_0|\otimes\rho_E$ to $|s_0\rangle\langle r_0|\otimes V_{s_0,0}\rho_EV_{r_0,0}^\dagger$. Define $M_\nu(s',r';s,r)=\langle s'|P_\mu|s\rangle\langle r|P_{\mu'}^\dagger|r'\rangle$ for $\nu=(\mu,\mu')$. The insertion changes the system factor according to
\begin{equation}
 P_\mu|s_0\rangle\langle r_0|P_{\mu'}^\dagger
 =\sum_{s_1,r_1}M_\nu(s_1,r_1;s_0,r_0)|s_1\rangle\langle r_1|.
 \label{eq:app-one-insertion-matrix}
\end{equation}
Propagating the second interval, tracing the environment once, and then tracing the system sets $r_1=s_1$ and gives
\begin{equation}
 \mathcal C_{\nu}^{(1),\mathrm{tr}}(t_1)
 =\sum_{s_0,r_0,s_1}\rho_{s_0r_0}M_\nu(s_1,s_1;s_0,r_0)
 \mathcal F_{B,\Pi}[(s_0,s_1),(r_0,s_1)].
 \label{eq:app-one-insertion}
\end{equation}
Repeating this step at $n$ insertions gives
\begin{equation}
 \mathcal C_{\boldsymbol\nu}^{(n),\mathrm{tr}}
 =\sum_{\boldsymbol s,\boldsymbol r}\delta_{s_nr_n}\rho_{s_0r_0}
 \mathcal F_{B,\Pi}[\boldsymbol s,\boldsymbol r]
 \prod_{j=1}^{n}M_{\nu_j}(s_j,r_j;s_{j-1},r_{j-1}).
 \label{eq:app-fixed-basis-amplitude}
\end{equation}
The two labels per interval are forward and backward operator-branch indices, not two physical systems. In a fixed basis these labels do not change inside a particle-event interval. A rotating basis requires further subdivisions because successive basis changes also mix them.

\subsection{What changes under a rotating coupling}
After entering the random-unitary frame, the coupling becomes $S_\omega(t)=U_\omega^\dagger(t)S U_\omega(t)$. The operation basis $P_\mu$ remains fixed in logical coordinates. At each grid point define
\begin{equation}
 \Pi_s(U)=U^\dagger|s\rangle\langle s|U,\qquad
 \mathcal E_{sr}(U)[X]=\Pi_s(U)X\Pi_r(U).
 \label{eq:app-moving-projector}
\end{equation}
Take a grid containing all particle times and freeze $U$ on each half-open interval. Let $\mathcal P_k$ equal the required insertion at $u_k$, or the identity if there is none. For a fixed double history, first form
\begin{equation}
 X_{\boldsymbol s\boldsymbol r}(\boldsymbol U)
 =\mathcal E_{s_{N-1}r_{N-1}}(U_{N-1})\mathcal P_{N-1}
 \cdots\mathcal E_{s_0r_0}(U_0)\mathcal P_0[\rho_S].
 \label{eq:app-projector-chain}
\end{equation}
The grid amplitude is then
\begin{equation}
 \mathcal C_{\boldsymbol\nu,\omega,\Pi}^{(n),\mathrm{tr}}
 =\sum_{\boldsymbol s,\boldsymbol r}
 \mathcal F_{B,\Pi}[\boldsymbol s,\boldsymbol r]
 \Tr_S X_{\boldsymbol s\boldsymbol r}(\boldsymbol U_\omega).
 \label{eq:app-history-sum}
\end{equation}
Equation~\eqref{eq:app-projector-chain} contains the initial-state overlap, every successive basis overlap, and every particle insertion. For example, one left-branch factor between adjacent cells is $\langle s_k|U_kU_{k-1}^\dagger|s_{k-1}\rangle$; replacing it by $\delta_{s_ks_{k-1}}$ is only justified for a fixed basis. The right branch carries the conjugate overlap. The final trace closes the system output only after all operations. If no terminal operation depends on $U_N$, its last transition density integrates to one; internal densities cannot be removed this way.

The Gaussian factor in Eq.~\eqref{eq:app-full-bath-trace} depends on the complete pair of labels, so it must multiply the complete system chain before the label sum and path average. A finite-duration logical instrument is represented by its time-local covariant generator in the chain; an endpoint gate is insufficient. In the fixed-basis limit, adjacent projectors force unchanged labels between particle events, and Eq.~\eqref{eq:app-history-sum} reduces to the finite-interval formula of Ref.~\cite{wassnerCPT}. For a moving basis, this collapse is generally false. The continuum coefficient is the limit of Eq.~\eqref{eq:app-history-sum} under the stated Hamiltonian and bath-kernel convergence conditions, followed by one expectation over complete control histories.

\subsection{Moving-basis connection and the frozen-grid limit}
\label{app:moving-limit}
For a differentiable frame, the overlap between successive instantaneous eigenbases is
\begin{equation}
 U(t+dt)U^\dagger(t)=\id+\dot U(t)U^\dagger(t)dt+o(dt).
 \label{eq:app-basis-overlap}
\end{equation}
Equivalently, transform the full logical operator to the fixed eigenbasis by $X_{\rm eig}=UX_{\rm log}U^\dagger$. Between particle insertions, in the interaction picture of $H_E$,
\begin{equation}
 \dot X_{\rm eig}=-i[h(t)\otimes\id_E,X_{\rm eig}]
 -ig[S\otimes B(t),X_{\rm eig}],\qquad \dot U=-ihU.
 \label{eq:app-moving-connection}
\end{equation}
The first commutator is a coordinate-connection term. Equation~\eqref{eq:app-projector-chain} already contains it through finite adjacent-basis overlaps; adding the commutator while retaining those overlaps would count the same motion twice. For a Brownian frame the corresponding differential equation is Stratonovich. Its It\^o form adds $-(\kappa/2)\sum_c[F_c,[F_c,X_{\rm eig}]]dt$. The finite-overlap calculation does not require a derivative of a Brownian path.

On a fixed grid, Eq.~\eqref{eq:app-history-sum} is exact for the left-endpoint-frozen coupling model. Suppose $B(t)$ is bounded and the finite set of insertions has bounded operator norms. Duhamel comparison of joint propagators and trace-norm contraction give the sufficient pathwise estimate
\begin{multline}
 |\mathcal C_{\boldsymbol\nu,\omega}^{(n),\mathrm{tr}}
 -\mathcal C_{\boldsymbol\nu,\omega,\Pi}^{(n),\mathrm{tr}}|
 \le 2|g|\,\|\rho_{SE}\|_1
 \prod_{j=1}^n\|P_{\mu_j}\|\|P_{\mu'_j}\|\\
 \times\int_0^Tdt\,\|B(t)\|
 \|S_\omega(t)-S_\omega(u_{k(t)})\|.
 \label{eq:app-freezing-error}
\end{multline}
Here $k(t)$ is the cell containing $t$, and $\rho_{SE}$ includes the bath state with its statistical weight; for a normalized physical input $\|\rho_{SE}\|_1=1$. The factor $2$ occurs once, from
$\|[\Delta H(t),X]\|_1\le2\|\Delta H(t)\|_\infty\|X\|_1$ with
$\Delta H(t)=g[S_\omega(t)-S_\omega(u_{k(t)})]\otimes B(t)$.
 For one insertion $\mathcal P_{(\mu,\mu')}[X]=(P_\mu\otimes\id_E)X(P_{\mu'}^\dagger\otimes\id_E)$, the induced trace norm is at most $\|P_\mu\|_\infty\|P_{\mu'}\|_\infty$. Thus $n$ insertions give the displayed product of $n$ left/right pairs, with no extra square or factor $2^n$. Unit-norm Pauli insertions contribute one. Joint unitary propagation preserves the trace norm between insertions, including for non-Hermitian coefficients, and the terminal partial trace is contractive. Duhamel comparison then telescopes the propagation intervals and gives the single integral in Eq.~\eqref{eq:app-freezing-error}. Brownian and OU paths are uniformly continuous on each finite realization, so the right side vanishes under grid refinement when the interaction is bounded. For an unbounded bosonic field, the corresponding check uses the integrated kernel and final-amplitude limits, or a separate cutoff or energy-constrained estimate. Paired-grid agreement provides an independent numerical check alongside this argument.

\subsection{An independent continuous expansion and numerical checks}
The same coefficient can be expanded without instantaneous eigenvectors. In the
interaction picture of $H_E$, expand the joint Liouville propagator between
the fixed particle insertions in a Dyson series. For a factorized initial
state and a centered Gaussian bath coupled through a linear field $B(t)$,
introduce the system-only left and right
branch maps
\begin{equation*}
 \mathcal S_{\omega,L}(t)[X]=-ig\,S_\omega(t)X,\qquad
 \mathcal S_{\omega,R}(t)[X]=+ig\,XS_\omega(t).
\end{equation*}
The system-only insertions are
$\mathcal P_\nu^S[X]=P_\mu X P_{\mu'}^\dagger$ for $\nu=(\mu,\mu')$.
Define bath maps $\mathcal B_L(t)[Y]=B(t)Y$ and
$\mathcal B_R(t)[Y]=YB(t)$. For $t>s$, the pair contraction inherited
from the joint Dyson word is
\begin{equation}
 C^{\epsilon_b\epsilon_a}(t,s)
 =\Tr_E\!\left[\mathcal B_{\epsilon_b}(t)
 \mathcal B_{\epsilon_a}(s)[\rho_E]\right]
 =\begin{cases}
 C(t,s),&\epsilon_a=L,\\
 C(s,t),&\epsilon_a=R,
 \end{cases}
 \label{eq:app-wick-branch-covariance}
\end{equation}
where $C(t,s)=\Tr_E[B(t)B(s)\rho_E]$. For a Hermitian field,
$C(s,t)=C(t,s)^*$. Wick's theorem applies to the complete bath word:
its right-branch factors precede its left-branch factors after cyclically
moving $\rho_E$ to the end of the trace. The fixed coefficient is therefore
\begin{equation}
\begin{aligned}
 \mathcal C_{\boldsymbol\nu,\omega}^{(n),\mathrm{tr}}(\boldsymbol t)
 ={}&\sum_{m\ge0}\int_{0<v_1<\cdots<v_{2m}<T}d\boldsymbol v
 \sum_{\boldsymbol\epsilon\in\{L,R\}^{2m}}
 \sum_{\text{pairings }\pi}
 \prod_{\substack{(a,b)\in\pi\\a<b}}
 C^{\epsilon_b\epsilon_a}(v_b,v_a)\\
 &{}\times\Tr_S\!\left[
 \mathcal T\!\left\{
 \prod_{a=1}^{2m}\mathcal S_{\omega,\epsilon_a}(v_a)
 \prod_{j=1}^n\mathcal P^S_{\nu_j,t_j}
 \right\}[\rho_S]\right].
\end{aligned}
 \label{eq:app-wick}
\end{equation}
Here $\mathcal T$ interleaves all system maps by their actual times, with the
latest map on the left; it does not move a paired block across a particle
insertion. The branch signs are already included in
$\mathcal S_{\omega,L/R}$. At $m=0$ the empty products give the uncoupled
coefficient. For unbounded Gaussian fields the formula is first interpreted
with the cutoff and convergence assumptions of
Appendix~\ref{app:gaussian-L1-sector}. It retains the complex cross-branch
orderings and provides a grid-independent check on the full-history
calculation. A finite-order truncation provides a weak-coupling expansion,
whereas Eq.~\eqref{eq:app-history-sum} evaluates the full arbitrary-coupling history.

 For arbitrary coupling, cache the bath triangle $\eta_{kl}$ and append one double label $(s_k,r_k)$ at each grid step. Attach its new $\eta_{kl}$ links to every earlier cell, apply the current projector and insertion, and then compress the history tensor. Fixed and Poisson paths use their actual holding intervals and event times. Brownian and OU paths retain one complete sampled frame history. Grid size, control substeps, bath quadrature, retained rank, and trajectory count require independent convergence checks. The discarded tensor weight is tracked as a compression diagnostic, while entropy estimates are evaluated separately. The two-time weak-coupling kernel in Appendix~\ref{app:weak} provides a check on the full influence functional at arbitrary $g$.

\subsection{Algorithm for the complete particle-amplitude tensor}
\label{alg:particle-history}
The common full-memory contraction below evaluates every protocol. Appendix~\ref{alg:control-paths} specifies the fixed-refresh, Poisson-refresh, Brownian, and OU path generators in that order. The output is a complex coefficient for each requested operation-label history; no modulus is taken during propagation or control averaging.
\begin{enumerate}
\item \textbf{Input.} Specify $\rho_S$, $g$, $S$ and its eigenvalues $\lambda_s$, the bath correlation $C(t)$ (or its spectral and Matsubara data), $T$, the ordered insertion times $\boldsymbol t$, the operation basis $P_\mu$, the control protocol and strength, the number $N_{\rm tr}$ of complete trajectories, and independent grid, bath, compression, and sampling tolerances. Choose a grid $0=u_0<\cdots<u_M=T$ containing every $t_j$. An index $k$ labels a numerical interval; $n$ labels insertions and need not equal $M$.
\item \textbf{Bath preprocessing.} Integrate $C$ over each ordered pair of grid cells to obtain the complex triangular coefficients $\eta_{kl}$, $0\le l\le k<M$, as in Eq.~\eqref{eq:app-discrete-influence}. Include the diagonal cell integrals and the specified high-frequency tail correction. Cache this triangle once for all trajectories and all operation labels on the same grid. The backward branch uses $\eta_{kl}^*$ in the influence exponent; replacing the triangle by a real symmetrized two-time correlation changes the amplitude.
\item \textbf{Full-memory influence tensor.} Start with the four double eigenlabels $a_0=(s_0,r_0)$. At step $k$, append $a_k=(s_k,r_k)$, multiply its diagonal factor and its pair factor with \emph{every} earlier $a_l$, $l<k$, from Eq.~\eqref{eq:app-discrete-influence}. This constructs $\mathcal F_{B,\Pi}[\boldsymbol s,\boldsymbol r]$ as one history tensor. After attaching all links from the new label, apply a QR factorization sweep ($Q$ has orthonormal columns and $R$ is upper triangular), followed by singular value decomposition (SVD) and truncation. Record the largest retained rank, whether the rank cap was reached, and the discarded squared singular weight. The optional transition reference used for tensor conditioning is inserted and divided out at each system transition, including insertion nodes; it is not a bath approximation.
\item \textbf{One complete control path.} Generate $U_\omega$ from time zero to $T$ by the appropriate procedure in Appendix~\ref{alg:control-paths}. Retain all frames on the grid, including the frame at every insertion. For each requested $\boldsymbol\nu$, contract the cached influence tensor with the \emph{same} complete system chain in Eq.~\eqref{eq:app-projector-chain}. At a grid node apply $\mathcal P_{\nu_j}$ in fixed logical coordinates and the current-frame projectors $\mathcal E_{s_kr_k}(U_{\omega,k})$. Sum the double labels only as the history tensor is contracted, and take the system trace at the closed boundary. The environment has already been eliminated once through the global influence tensor, not interval by interval. A trace-preserving final propagation after the last insertion may be removed only when there is no later operation or measured output.
\item \textbf{Average and compare.} Repeat the preceding contraction for $\omega=1,\ldots,N_{\rm tr}$. For every label history form $\overline{\mathcal C}_{\boldsymbol\nu}=N_{\rm tr}^{-1}\sum_\omega\mathcal C_{\boldsymbol\nu,\omega}$ as a \emph{complex} average. Compute the ideal coefficient with $g=0$ and the same input and insertions. Only then form $\Delta\mathcal C_{\boldsymbol\nu}=\overline{\mathcal C}_{\boldsymbol\nu}-\mathcal C_{\boldsymbol\nu,{\rm ideal}}$ and, if desired, average $|\Delta\mathcal C_{\boldsymbol\nu}|$ or $|\Delta\mathcal C_{\boldsymbol\nu}|^2$ over labels. For the two-insertion qubit calculation, enumerate the 15 nonvacuum labels at each time and return the $15\times15$ matrix indexed by (later label, earlier label).
\item \textbf{Independent checks.} Refine $M$, the control integration substep, the Matsubara cutoff and tail treatment, the MPS tolerance and rank cap, and $N_{\rm tr}$ separately. Compare representative small grids with an uncompressed history sum, and compare final complex amplitudes across refinements. These checks complement the full-history arbitrary-coupling contraction.
\end{enumerate}

%% file: appendix_instrument.tex
\section{Finite-duration instrument contraction and its full particle sum}
\label{app:instrument}

The cPT contains coefficients for all nonvacuum insertions, but an experiment specifies one time-dependent logical instrument. The probability experiments use $m=1$ or $2$ system qubits coupled to one environment qubit, so $d_S=2^m$. Put $Z_j=\sigma_z$ on system qubit $j$ and identity on the other system qubits, and let $U_\omega(t)=\bigotimes_{j=1}^m U_{j,\omega}(t)$. Their common bath model is $H_E=(\omega_E/2)\sigma_z^{(E)}$ and $H_{SE}=g\sum_{j=1}^m Z_j\otimes\sigma_x^{(E)}$. For $m=1$ this is the original single-system-qubit model; for $m=2$ the two qubits share the same $E$. No independent system Hamiltonian is present: $H_{\rm inst}(t)$ implements the target gates. Along a complete frame trajectory the logical-picture process generator is
\begin{equation}
 Q_\omega(t)[X]=-i\left[\id_S\otimes H_E+
 g\sum_{j=1}^m U_\omega^\dagger(t)Z_jU_\omega(t)\otimes\sigma_x^{(E)},X\right],
 \label{eq:app-inst-process-Q}
\end{equation}
and the target instrument generator is $\mathbb L_{\rm inst}(t)[X]=-i[H_{\rm inst}(t),X]$. It is fixed in logical coordinates. The laboratory instrument is its pointwise frame conjugate, as derived in Eqs.~\eqref{eq:app-covariant-map} and \eqref{eq:app-transparent-ode}.

For the numerical Pauli convention, $P_\mu$ runs over the $4^m$ tensor products of $I,X,Y,Z$, with $\Tr(P_\mu^\dagger P_{\mu'})=d_S\delta_{\mu\mu'}$. Define $\mathbb P_{\mu\mu'}[X]=P_\mu XP_{\mu'}^\dagger$. The $d_S^4$ resulting superoperators are orthogonal in the $d_S^2$-dimensional system Liouville space:
\begin{equation}
 \Tr_{\rm super}(\mathbb P_\nu^\dagger\mathbb P_{\nu'})
 =d_S^2\delta_{\nu\nu'}.
 \label{eq:app-inst-orthogonality}
\end{equation}
Thus the instrument expansion coefficients, including the factor from this unscaled basis, are
\begin{equation}
 g_\nu(t)=d_S^{-2}\Tr_{\rm super}
 [\mathbb P_\nu^\dagger\mathbb L_{\rm inst}(t)],\qquad
 \mathbb L_{\rm inst}(t)=\sum_\nu g_\nu(t)\mathbb P_\nu.
 \label{eq:app-inst-coefficients}
\end{equation}
For a commutator Hamiltonian the coefficient of the identity superoperator vanishes. Let $R_\nu=\mathbb P_\nu\otimes\mathrm{id}_E$ and let $c_\nu=g_\nu^*$ be the matching instrument-field coefficient. The measured version of the cPT contraction corresponding to Eq.~(77) of Ref.~\cite{wassnerCPT} is
\begin{equation}
 p_\omega(x)=\langle\!\langle E_x\otimes\id_E|
 \mathcal T\exp\!\left\{\int_0^Tdt
 \left[Q_\omega(t)+\sum_{\nu\ne(0,0)}c_\nu(t)^*R_\nu\right]\right\}
 |\rho_0^{SE}\rangle\!\rangle.
 \label{eq:app-inst-overlap}
\end{equation}
The left boundary means $\langle\!\langle E_x\otimes\id_E|X\rangle\!\rangle=\Tr_{SE}[(E_x\otimes\id_E)X]$. The fixed logical POVM $E_x$ is realized in the laboratory with the final-frame correction of Eq.~\eqref{eq:app-transparent-readout}. Summing all $x$ restores the trace boundary; keeping one $x$ gives a nontrivial experimental probability.

To see how Eq.~\eqref{eq:app-inst-overlap} contains every particle number, let $\mathcal V_\omega(t,s)$ denote propagation by $Q_\omega$ alone and expand the instrument part in a Dyson series:
\begin{multline}
 p_\omega(x)=\sum_{n=0}^\infty
 \sum_{\nu_1,\ldots,\nu_n\ne(0,0)}
 \int_{0<t_1<\cdots<t_n<T}d^n\boldsymbol t\,
 \prod_{j=1}^n c_{\nu_j}(t_j)^*\\
 \times\langle\!\langle E_x\otimes\id_E|
 \mathcal V_\omega(T,t_n)R_{\nu_n}
 \mathcal V_\omega(t_n,t_{n-1})\cdots
 R_{\nu_1}\mathcal V_\omega(t_1,0)
 |\rho_0^{SE}\rangle\!\rangle.
 \label{eq:app-inst-dyson}
\end{multline}
For $n=0$, the term in Eq.~\eqref{eq:app-inst-dyson} means $\langle\!\langle E_x\otimes\id_E|\mathcal V_\omega(T,0)|\rho_0^{SE}\rangle\!\rangle$, with no $t_n$ or insertion label. Every other term contains the corresponding particle kernel with a measured rather than trace-only endpoint. The finite auxiliary-space exponential resums all $n$, so the probability calculations do not truncate the particle expansion. In the code one first verifies $\sum_\nu c_\nu^*R_\nu=-i[H_{\rm inst}\otimes\id_E,\cdot]$ before propagating each trajectory. This algebraic check prevents a factor-of-two error from mixing the unscaled Pauli convention here with the normalized one in Appendix~\ref{app:cmi}.

For the single-system-qubit experiment, put $\ell=T/3=0.4$, $H_{\rm Had}=(\sigma_x+\sigma_z)/\sqrt2$, and $\theta=\pi/4$, with $R_y(\theta)=e^{-i\theta\sigma_y/2}$. The actual piecewise-constant logical Hamiltonian is
\begin{equation}
 H_{\rm inst}(t)=
 \begin{cases}
 \dfrac{\pi}{2\ell}H_{\rm Had},&0\le t<\ell,\\
 0,&\ell\le t<2\ell,\\
 \dfrac{\theta}{2\ell}\sigma_y,&2\ell\le t\le T.
 \end{cases}
 \label{eq:app-single-qubit-generators}
\end{equation}
The first interval gives $-iH_{\rm Had}$ and the last gives $R_y(\pi/4)$; the scalar phase does not affect the logical $Z$-basis readout $E_0=|0\rangle\langle0|$, $E_1=|1\rangle\langle1|$. For the stated initial system state, $p_{0,\mathrm{id}}=(1-\sin\theta)/2$. The idle middle interval has $g_\nu=0$ but the process $Q_\omega$ and its environment continue to evolve. For the Bell experiment, let $P_1=(\id-\sigma_z^{(S_1)})/2$ and $h=T/2$. Its two target Hamiltonians are
\begin{equation}
 H_{\rm inst}^{(1)}=\frac{\pi}{2h}H_{S_1}\otimes\id_{S_2},\qquad
 H_{\rm inst}^{(2)}=\frac{\pi}{2h}P_1\otimes(\id_{S_2}-\sigma_x^{(S_2)}),
 \label{eq:app-bell-generators}
\end{equation}
where $H_{S_1}=(\sigma_x^{(S_1)}+\sigma_z^{(S_1)})/\sqrt2$. Evolution for time $h$ yields $-iH_{S_1}$ in the first interval and exactly $\mathrm{CNOT}_{S_1\to S_2}$ in the second. The Bell readout is $E_0=|\Phi^+\rangle\langle\Phi^+|$ and $E_1=\id_{S_1S_2}-E_0$, so $p_{0,\mathrm{id}}=1$ and $p_0=\langle\Phi^+|\rho_{S_1S_2}(T)|\Phi^+\rangle$. In fixed and Poisson refreshes both local Pauli frames are independently redrawn at each common refresh event; Brownian and OU controls use two independent local paths. In either experiment one complete control history is retained through all instrument windows, and the 500 trajectory probabilities are averaged only after evaluating $p_\omega(x)$. This contraction tests the chosen finite-duration experiments, not the supremum over all allowed instruments in Eq.~\eqref{eq:operational-error}.

\subsection{A finite instrument with an auxiliary system and records}
\label{app:instrument-covariance}
The Hamiltonian examples above are not the definition of a finite-duration instrument. Let $M$ contain its memory, detector ancilla, and any coherent record. A logical interaction $H_{SM}^{\rm inst}(t)$ is implemented, during its entire support, as
\begin{equation}
 H_{SM,\omega}^{\rm lab}(t)
 =(U_\omega(t)\otimes\id_M)H_{SM}^{\rm inst}(t)
 (U_\omega^\dagger(t)\otimes\id_M).
 \label{eq:app-inst-aux-cov}
\end{equation}
The control generator appears separately in the laboratory Hamiltonian and is not included a second time in Eq.~\eqref{eq:app-inst-aux-cov}. With $H_{SE}=0$, let $V_{\rm inst}^{\rm log}(t,a)$ be the joint $SM$ propagator of the target interaction. On a differentiable frame, differentiation of
\begin{equation}
 V_{{\rm inst},\omega}^{\rm lab}(b,a)
 =(U_\omega(b)\otimes\id_M)V_{\rm inst}^{\rm log}(b,a)
 (U_\omega^\dagger(a)\otimes\id_M)
 \label{eq:app-inst-endpoints}
\end{equation}
produces $H_{\rm ctrl}^\omega\otimes\id_M+H_{SM,\omega}^{\rm lab}$ and the identity initial value. Uniqueness of the propagator proves the endpoint formula. Its derivation used the pointwise rule, so an endpoint relation for an isolated black-box channel is not a substitute for specifying what happens while that channel overlaps with a memory-bearing environment.

To separate implementation from its final description, initialize $M$ in the same state for the logical and laboratory devices, and let outcome $x$ correspond to the same effect on $M$. Reducing Eq.~\eqref{eq:app-inst-endpoints} on $M$ gives the \emph{complete instrument map} for that record,
\begin{equation}
 \begin{aligned}
 \mathcal J_{x,\omega}^{\rm lab}(b,a)
 &=\operatorname{Ad}_{U_\omega(b)}\circ
 \mathcal J_x^{\rm log}(b,a)\circ
 \operatorname{Ad}_{U_\omega(a)^\dagger},\\
 \operatorname{Ad}_{V}[X]&=VXV^\dagger.
 \end{aligned}
 \label{eq:app-inst-two-endpoint-map}
\end{equation}
Only the two endpoint unitaries appear in this \emph{resulting map}, although Eq.~\eqref{eq:app-inst-aux-cov} must hold at every intermediate time to derive it. If $U_\omega(t)$ changes during the instrument, applying physical frame rotations only at $a$ and $b$ to an otherwise unmodified laboratory interaction generally does not produce Eq.~\eqref{eq:app-inst-two-endpoint-map}. Moreover, two instruments can share this isolated endpoint channel while differing in their internal timing; a cPT that couples to the system during $[a,b]$ distinguishes them.

If the instrument includes monitored jumps, each logical jump $L_\ell^{\rm log}(t)$ on $SM$ is transformed by the same conjugation as Eq.~\eqref{eq:app-inst-aux-cov}. In a given record branch, the no-jump propagators and jump insertions then have adjacent $U_\omega^\dagger(t)U_\omega(t)$ factors that cancel. Tracing or reading $M$ therefore returns the same logical branch kernel trajectory by trajectory. An instrument with internal non-Markovian memory can be treated by retaining its $M$ explicitly and transforming every system interface; one need not assume that the reduced instrument on $S$ alone has a time-local generator. Fixed and Poisson refresh use the jump identity, Brownian driving uses Stratonovich differentiation, and OU driving uses ordinary differentiation. With $H_{SE}\ne0$, these identities keep the instrument fixed in the logical frame; they do not remove the transformed interaction $gU_\omega^\dagger S U_\omega\otimes B$ that acts during the instrument.

The same point is visible in the process cMPS. Write $\mathbb U_\omega(t)=\operatorname{Ad}_{U_\omega(t)}\otimes\mathrm{id}_E$ for the coordinate change on the joint operator space. This superoperator is distinct from the Pauli target frame $G_j$. For a smooth path the process generators satisfy
\begin{equation}
 Q_{{\rm lab},\omega}=
 \dot{\mathbb U}_\omega\mathbb U_\omega^{-1}
 +\mathbb U_\omega Q_{{\rm log},\omega}\mathbb U_\omega^{-1},
 \qquad R_{\nu,\omega}^{\rm lab,enc}
 =\mathbb U_\omega R_\nu\mathbb U_\omega^{-1}.
 \label{eq:app-inst-cmps-cov}
\end{equation}
The first term is interpreted through the Stratonovich product rule for Brownian control. The initial state is encoded by $\mathbb U_\omega(0)$ and the closed trace boundary satisfies $\ell\mathbb U_\omega(T)=\ell$. Substituting Eq.~\eqref{eq:app-inst-cmps-cov} into any ordered particle chain cancels every internal frame pair, proving equality between its laboratory and logical coefficient at arbitrary coupling. This coordinate identity is separate from the stronger no-coupling transparency proved in Appendix~\ref{app:physicality}.

%% file: appendix_bath.tex
\section{Gaussian influence functional and the Drude bath}
\label{app:bath}

\subsection{Thermal correlation and a common convention}
Let $B=\sum_q(c_qb_q+c_q^*b_q^\dagger)$ and $H_E=\sum_q\omega_qb_q^\dagger b_q$, with $\rho_E\propto e^{-\beta H_E}$. We use $J_B(\Omega)=\sum_q|c_q|^2\delta(\Omega-\omega_q)$, without an implicit factor of $\pi$, and $n_\beta(\Omega)=(e^{\beta\Omega}-1)^{-1}$. Direct evaluation of the thermal two-point function gives
\begin{equation}
 C(v)=\Tr_E[B(v)B(0)\rho_E]
 =\int_0^\infty\!d\Omega\,J_B(\Omega)
 \left[(n_\beta+1)e^{-i\Omega v}+n_\beta e^{i\Omega v}\right].
 \label{eq:app-bath-correlation}
\end{equation}
The identities $C(-v)=C(v)^*$ and the positive-type property follow from the thermal trace. The full Gaussian cumulant expansion of Eq.~\eqref{eq:app-full-bath-trace}, for real forward and backward eigenvalue paths $\lambda_s(t),\lambda_r(t)$, is
\begin{multline}
 \mathcal F_B[s,r]=\exp\!\biggl\{-g^2\int_0^Tdt\int_0^tdu\,
 [\lambda_s(t)-\lambda_r(t)]\\
 \times\bigl[C(t-u)\lambda_s(u)-C(t-u)^*\lambda_r(u)\bigr]\biggr\}.
 \label{eq:app-influence}
\end{multline}
There is no factorization across neighboring cells. For a Gaussian bath, Eq.~\eqref{eq:app-influence} is exact for the frozen scalar-source history and every value of $g$ for which the integrated kernel exists. The $g^2$ exponent resums all connected Gaussian cumulants, so the expression is the exact Gaussian result rather than a second-order truncation. This convention is fixed by the environmental trace, not by a coefficient transcribed from the extracted text of Ref.~\cite{wassnerCPT}; the single-mode check below removes possible half-factor ambiguities.

On a grid $I_k=[u_k,u_{k+1})$, put
\begin{equation}
 \eta_{kl}=\int_{I_k}dt\int_{I_l\cap[0,t]}du\,C(t-u),\quad k\ge l.
 \label{eq:app-eta}
\end{equation}
For $k>l$ the domain is a rectangle, whereas for $k=l$ it is a triangle. Thus
\begin{equation}
 \mathcal F_{B,\Pi}[s,r]=\exp\!\left[-g^2\sum_{k\ge l}
 (\lambda_{s_k}-\lambda_{r_k})
 (\eta_{kl}\lambda_{s_l}-\eta_{kl}^*\lambda_{r_l})\right].
 \label{eq:app-discrete-influence}
\end{equation}
Replacing the diagonal triangle by a square doubles its leading short-time area and generally corrupts the phase as well. The same $\eta_{kl}$ is reused for every history; the label-dependent exponent is assembled only after all cross-time links are retained.

For $\operatorname{Re}z>0$, integration of $e^{-z(t-u)}$ over these domains gives
\begin{equation}
 F_{kl}(z)=\begin{cases}
 e^{-z(u_k-u_{l+1})}
 \dfrac{(1-e^{-z\Delta_k})(1-e^{-z\Delta_l})}{z^2},& k>l,\\[2mm]
 \dfrac{z\Delta_k-1+e^{-z\Delta_k}}{z^2},&k=l.
 \end{cases}
 \label{eq:app-Fkl}
\end{equation}
At $z=0$, the limits are $\Delta_k\Delta_l$ and $\Delta_k^2/2$, respectively. Substituting Eq.~\eqref{eq:app-bath-correlation} gives a direct spectral quadrature,
\begin{equation}
 \eta_{kl}=\int_0^\infty\!d\Omega\,J_B(\Omega)
 \bigl[(n_\beta+1)F_{kl}(i\Omega)+n_\beta F_{kl}(-i\Omega)\bigr].
 \label{eq:app-spectral-eta}
\end{equation}
The integrated formula must satisfy $\mathcal F_B[s,s]=1$, $\mathcal F_B[r,s]=\mathcal F_B[s,r]^*$, and $|\mathcal F_B[s,r]|\le1$. The last inequality is for the complete environmental overlap, not for each cross-cell factor separately. With a finite bath variance, a constant-path short-time expansion is
\begin{equation}
 \mathcal F_B=1-\tfrac12g^2(\lambda_s-\lambda_r)^2
 \langle B^2\rangle_E T^2+o(T^2).
\end{equation}
For one mode of frequency $\Omega$ and coefficient $c$, direct displacement-operator evaluation instead gives the nonperturbative check
\begin{multline}
 \log\mathcal F_B=-g^2(\lambda_s-\lambda_r)^2|c|^2
 \coth(\beta\Omega/2)\frac{1-\cos\Omega T}{\Omega^2}\\
 +ig^2(\lambda_s^2-\lambda_r^2)|c|^2
 \frac{\Omega T-\sin\Omega T}{\Omega^2}.
 \label{eq:app-single-mode}
\end{multline}

\subsection{Drude--Lorentz spectral density and Matsubara integration}
For the finite-temperature example use
\begin{equation}
 J_D(\Omega)=\frac{2\Lambda\gamma_B}{\pi}
 \frac{\Omega}{\Omega^2+\gamma_B^2},\qquad
 \int_0^\infty\frac{J_D(\Omega)}{\Omega}d\Omega=\Lambda.
 \label{eq:app-drude}
\end{equation}
Here $\Lambda$ is the reorganization energy and $\gamma_B^{-1}$ is the bath time scale. Since $J_D(\Omega)\sim\Omega^{-1}$, the equal-time $C(0)$ is logarithmically divergent for the ideal uncut spectrum. The integrated $\eta_{kl}$ can nevertheless have a defined limit because Eq.~\eqref{eq:app-Fkl} supplies additional high-frequency decay. A microscopic ultraviolet regularization $J_D(\Omega)e^{-\Omega/\Omega_{\rm UV}}$ is a distinct bath model; its cutoff must be stated and converged separately. The finite-variance $T^2$ check above applies to the regularized model, whereas the uncut ideal Drude spectrum is treated through its integrated kernels.

For $v>0$, contour integration of Eq.~\eqref{eq:app-bath-correlation} yields the Matsubara representation
\begin{align}
 C(v)&=\sum_{m=0}^\infty c_m e^{-\nu_m v},\\
 \nu_0&=\gamma_B,\qquad
 c_0=\Lambda\gamma_B\left[\cot\left(\frac{\beta\gamma_B}{2}\right)-i\right],\\
 \nu_m&=2\pi m/\beta,\quad
 c_m=\frac{4\Lambda\gamma_B}{\beta}
 \frac{\nu_m}{\nu_m^2-\gamma_B^2}\quad(m\ge1).
 \label{eq:app-matsubara}
\end{align}
Linearity of Eq.~\eqref{eq:app-eta} gives the analytic interval kernel
\begin{equation}
 \eta_{kl}=\sum_{m=0}^\infty c_m F_{kl}(\nu_m).
 \label{eq:app-drude-eta}
\end{equation}
The Matsubara terms represent the thermal correlation as a combined sum rather than as independent positive classical noises; individual $c_m$ may have either sign. If $\nu_m=\gamma_B$, the apparent poles in the separate terms cancel only after their contributions are combined and the limit is taken. In computation we truncate the sum only after checking the integrated kernels and final amplitudes against larger cutoffs, or use a separately verified rational approximation.

Substituting Eq.~\eqref{eq:app-drude-eta} into Eq.~\eqref{eq:app-discrete-influence}, and then into Eq.~\eqref{eq:app-history-sum}, gives the arbitrary-coupling Drude algorithm. Qubit system histories have four double labels per time cell, so uncompressed enumeration grows exponentially with grid size. Fixed and Poisson refresh calculations use actual holding intervals and event-time signs. Brownian and OU calculations use full-memory history-tensor compression on sampled continuous paths. The weak-coupling two-time kernel is used as a diagnostic, while the exact arbitrary-$g$ calculation uses the full influence functional.

%% file: appendix_control.tex
\section{Complete-path control averages and executable propagators}
\label{app:controlalg}

\subsection{Fixed Pauli refresh}
For a single-axis coupling and a constant Pauli frame on a holding interval, $G_j^\dagger\sigma_zG_j=y_j\sigma_z$ with independent, equiprobable $y_j\in\{\pm1\}$. The $\sigma_z$ eigenbasis is now fixed and the system chain $A_{\boldsymbol\nu}[s,r]$ in Eq.~\eqref{eq:app-history-sum} is independent of the signs. The bath factor is Eq.~\eqref{eq:app-discrete-influence} with every $\lambda_{s_k},\lambda_{r_k}$ multiplied by its holding-interval sign. For $K$ fixed holding intervals its exact control average is
\begin{equation}
 \overline{\mathcal F}^{\rm fixed}_B[s,r]
 =2^{-K}\sum_{\boldsymbol y\in\{\pm1\}^K}
 \mathcal F_B[y\lambda_s,y\lambda_r].
 \label{eq:app-fixed-signs}
\end{equation}
Writing the exponent as $-\sum_kD_{kk}-\sum_{k>l}D_{kl}y_ky_l$ gives, for two holding intervals, $\overline{\mathcal F}_B=e^{-D_{00}-D_{11}}\cosh(D_{10})$. The $D_{kl}$ may be complex; this identity is an algebraic check for the combined influence factor, whose physicality is assessed from the complete environmental construction. Enumeration costs $2^K$ before compression, and a transfer or history tensor can retain cross-interval bath links when $K$ is large.

\subsection{Poisson Pauli refresh}
For Poisson rate $\lambda_r$, conditional on exactly $m$ events the ordered
times $0<\tau_1<\cdots<\tau_m<T$ have density $m!/T^m$ on the ordered
simplex. Including the Poisson probability gives the joint ordered-event
(Janossy) density $e^{-\lambda_rT}\lambda_r^m\,d^m\boldsymbol\tau$; it
contains no further $m!$. Including the initial frame and every independent
redraw gives
\begin{equation}
 \overline{\mathcal F}^{\rm Pois}_B[s,r]
 =e^{-\lambda_rT}\sum_{m=0}^\infty\lambda_r^m
 \int_{0<\tau_1<\cdots<\tau_m<T}\!d^m\boldsymbol\tau\,
 2^{-(m+1)}\sum_{y_0,\ldots,y_m=\pm1}
 \mathcal F_B[y\lambda_s,y\lambda_r].
 \label{eq:app-poisson-events}
\end{equation}
The particle times stay fixed when the event times are integrated. Every sampled schedule is split at the union of its events and particle times, and its bath kernel is recomputed for those actual interval lengths. For an independent redraw, even when the same Pauli is sampled again, a nontrivial character obeys $\E[y(t)y(s)]=e^{-\lambda_r|t-s|}$. This provides a weak-coupling two-time check, while the exact nonlinear influence-functional average is evaluated from Eq.~\eqref{eq:app-poisson-events}.

\subsection{Shared auxiliary representation for refresh processes}
The refresh average also has a finite classical-control auxiliary representation \emph{before} the environment is eliminated. Let $\mathcal M_{b,a}(G)$ denote complete cMPS propagation with the same frame $G$ held on $[a,b)$, including the joint $SE$ drift $Q_G$ and all operation-field emission maps $R_\nu$. Independence of the redraws implies the ordered identity
\begin{equation}
 |\overline\Upsilon_T^{\rm tr}\rangle
 =(\Tr_{SE}\otimes\id_F)
 \left[\frac1M\sum_G\mathcal M_{b_K,b_{K-1}}(G)\right]
 \cdots\left[\frac1M\sum_G\mathcal M_{b_1,b_0}(G)\right]
 [\rho_{SE}\otimes|\Omega\rangle],
 \label{eq:app-refresh-cmps-product}
\end{equation}
with $M=4$ for uniform Pauli frames. This is an expectation of \emph{complete interval propagators}; it neither averages their generators before exponentiating nor traces $E$ after each interval. Equivalently, retain a frame register and insert the redraw matrix $J_{GG'}=1/M$ at each fixed boundary, leaving the $SE$ operator untouched. For Poisson redraws the joint register--process generator is
\begin{equation}
 (\mathbb Q_{\rm Pois})_{GG'}
 =\delta_{GG'}Q_G+\lambda_r(M^{-1}-\delta_{GG'})\id_{SE},
 \qquad \mathbb R_\nu=\id_M\otimes R_\nu.
 \label{eq:app-poisson-register}
\end{equation}
Expanding the register jumps in its time-ordered exponential recovers the ordered-event integral in Eq.~\eqref{eq:app-poisson-events}. The register is classical and finite; the environmental memory need not be finite or Markovian after eliminating $E$.

If the Poisson event sum is evaluated deterministically only through $m_{\max}$, the omitted \emph{unnormalized} tail of a closed Pauli-basis particle amplitude obeys
\begin{equation}
 |\overline{\mathcal C}-\overline{\mathcal C}_{m\le m_{\max}}|
 \le\Pr\{\operatorname{Pois}(\lambda_rT)>m_{\max}\},
 \label{eq:app-poisson-tail}
\end{equation}
because complete unitary propagation and each Pauli left/right insertion preserve the relevant trace-norm bound. Renormalizing the retained event probabilities would instead compute a conditional protocol. Equation~\eqref{eq:app-poisson-tail} does not include event-time quadrature, bath-kernel, or Monte Carlo error. The register construction also makes physicality transparent: every finite schedule is a valid causal process and the schedule average is convex, while a finite-width laboratory pulse requires the separate estimate in Appendix~\ref{app:pulse-width}.

\subsection{Brownian conditional density and the two coordinate choices}
Let $U_0$ be Haar distributed, independently of the subsequent noises, and let
\begin{equation}
 dU_t=-i\sqrt\kappa\sum_{c=x,y,z}F_cU_t\circ dW_c(t),
 \qquad F_c=\sigma_c/2.
\end{equation}
The stochastic equation is left invariant. On $[u_k,u_{k+1}]$ its solution is $U_{k+1}=G_kU_k$, where $G_k$ starts at the identity and is driven only by the new Brownian increments. Thus $G_k=U_{k+1}U_k^\dagger$ is independent of the pre-$u_k$ filtration and has Haar density $q_{\Delta_k}$. For a measurable set $A$ and a fixed old frame $V$,
\begin{equation}
 \Pr(U_{k+1}\in A\mid U_k=V)
 =\int_{AV^\dagger}q_{\Delta_k}(G)d\mu(G)
 =\int_Aq_{\Delta_k}(U'V^\dagger)d\mu(U').
 \label{eq:app-conditional-density}
\end{equation}
The last equality is right invariance of Haar measure. Multiplication of these conditional measures gives Eq.~\eqref{eq:path-density}; it does not make the endpoint frames independent.

Fix one double eigenvalue history in Eq.~\eqref{eq:app-history-sum}. In fixed logical coordinates let $f_{k,+}^{\rm log}(U)$ be the accumulated system operator density immediately after the node $u_k$. Its recursion is
\begin{equation}
 f_{k,+}^{\rm log}(U')=\mathcal E_{s_kr_k}(U')\mathcal P_k
 \left[\int d\mu(U)q_{\Delta_{k-1}}(U'U^\dagger)
 f_{k-1,+}^{\rm log}(U)\right].
 \label{eq:app-logic-recursion}
\end{equation}
Start with $f_{0,+}^{\rm log}(U)=\mathcal E_{s_0r_0}(U)\mathcal P_0[\rho_S]$; after the last node take $\int d\mu(U)\Tr_Sf_{N-1,+}^{\rm log}(U)$ and then sum complete label histories with their bath factors. If the bath weight has already been attached incrementally, the corresponding history index must remain in the state. The integral in Eq.~\eqref{eq:app-logic-recursion} is simply a sequential evaluation of the joint integral and retains the current-frame/system-weight correlation.

Equivalently define $f(U)=Uf^{\rm log}(U)U^\dagger$. In this instantaneous eigenbasis the projector is the fixed map $X\mapsto|s\rangle\langle s|X|r\rangle\langle r|$, while the interval propagation becomes
\begin{equation}
 (\mathbb T_\Delta^{\rm B}f)(U')
 =\int d\mu(G)q_\Delta(G)
 Gf(G^\dagger U')G^\dagger.
 \label{eq:app-brown-kernel}
\end{equation}
The particle insertion at a node must also be transformed: $X\mapsto U\mathcal P_k[U^\dagger XU]U^\dagger$. This is why a simpler fixed projector does not make every node independent of $U$. The two recursions are the same change of coordinates. A closed trace boundary needs no transformation back, since $\Tr(U^\dagger f(U)U)=\Tr f(U)$; an open logical output or fixed logical POVM does require the appropriate inverse-frame conversion.

\subsection{SU(2) representation blocks and a closed Brownian interval}
Expand the function of the frame in matrix elements of irreducible representations,
\begin{equation}
 f(U)=\sum_{\Lambda,a,b}D^{(\Lambda)}_{ab}(U)
 f^{(\Lambda)}_{ab},\qquad
 D^{(\Lambda)}(e^{-i\theta F_c})=e^{-i\theta J_c^{(\Lambda)}}.
\end{equation}
The first index $a$ transforms under left multiplication, whereas $b$ is the multiplicity/right index and is unchanged by that multiplication. The representation identity $M(\theta+\eta)=M(\theta)M(\eta)$, differentiated in $\eta$ at zero, gives $M'(\theta)=M(\theta)M'(0)$; with $J_c^{(\Lambda)}=iM'(0)$ and $M(0)=I$ the stated matrix exponential follows. The star in the next equation denotes elementwise complex conjugation, not Hermitian adjoint. For normalized Haar measure, Schur orthogonality gives $\int d\mu(U)D_{ab}^{(\Lambda)}(U)^*D_{cd}^{(\Lambda')}(U)=\delta_{\Lambda\Lambda'}\delta_{ac}\delta_{bd}/d_\Lambda$, so coefficient extraction is $f_{ab}^{(\Lambda)}=d_\Lambda\int d\mu(U)D_{ab}^{(\Lambda)}(U)^*f(U)$. The dimension factor cancels between the coefficient definitions before and after the frame update, giving Eq.~\eqref{eq:app-brown-kernel}:
\begin{equation}
 f_{lb}^{(\Lambda),\rm new}
 =\sum_a\int d\mu(G)q_\Delta(G)
 D_{la}^{(\Lambda)}(G)^*Gf_{ab}^{(\Lambda)}G^\dagger.
 \label{eq:app-block-integral}
\end{equation}
The same $a$ is summed because the old left representation index is contracted with the new one $l$; the right index $b$ stays fixed. For a one-dimensional representation, $d_\Lambda=1$, Eq.~\eqref{eq:app-block-integral} reduces to the averaged adjoint action on a system operator. The density $q_\Delta$ is not discarded: it is the probability law in the integral, and integration of its moments is exactly what produces the finite-dimensional matrix exponential.

With $\mathcal L_{F_c}[X]=[F_c,X]$, define the joint Hermitian generator on the tensor product of representation index and operator space,
\begin{equation}
 \mathsf T_c^{(\Lambda)}=-J_c^{(\Lambda)*}\otimes I
 +I\otimes\mathcal L_{F_c}.
 \label{eq:app-joint-generator}
\end{equation}
For optional deterministic left drift $d_c$, the exact interval average in Eq.~\eqref{eq:app-block-integral} is
\begin{equation}
 \mathsf M_\Delta^{(\Lambda)}
 =\exp\!\left[\Delta\left(-i\sum_cd_c\mathsf T_c^{(\Lambda)}
 -\frac\kappa2\sum_c(\mathsf T_c^{(\Lambda)})^2\right)\right].
 \label{eq:app-brown-exponential}
\end{equation}
The drift is optional; our reported stochastic calculations use $d_c=0$. A constant drift changes the oscillatory phase of control modes. In the isotropic case its generator commutes with the Casimir, so the interval can still be resolved by joint total-spin and drift-axis magnetic-number projectors; a time-dependent or anisotropic generator need not admit that simplification. In the isotropic drift-free case, coupling a spin-$j$ representation to the scalar plus spin-one system-operator adjoint representation yields projectors $\Pi_\ell$ onto the allowed total spins $\ell\in\{|j-1|,j,j+1\}$, with the usual removal of nonexistent sectors. Thus
\begin{equation}
 \mathsf M_\Delta^{(j)}
 =\sum_\ell e^{-\kappa\ell(\ell+1)\Delta/2}\Pi_\ell.
 \label{eq:app-casimir}
\end{equation}
For $j=0$ the blocks are $\ell=0,1$; for $j=1$ their decay factors are $1,e^{-\kappa\Delta},e^{-3\kappa\Delta}$. This is the analytic solution for one Brownian interval. Node projectors and transformed insertions mix $j$ blocks, so applying only the trivial block after every interval would erase needed correlations. One precomputes the finite Haar node matrices, alternates them with Eq.~\eqref{eq:app-brown-exponential}, and reads the trivial representation only at the final Haar integral. A representation cutoff is checked by increasing it until the desired amplitudes converge; there is no universal finite maximum $j$ for every history.

For a constant drift of magnitude $h$ along one fixed axis, let $\Pi_{\ell m}^{(j)}$ project simultaneously onto total spin $\ell$ and magnetic number $m$ of $\mathsf T_z^{(j)}$. The corresponding analytic interval is
\begin{equation}
 \mathsf M_\Delta^{(j)}
 =\sum_{\ell}\sum_{m=-\ell}^{\ell}
 e^{-\kappa\ell(\ell+1)\Delta/2-ihm\Delta}
 \Pi_{\ell m}^{(j)}.
 \label{eq:app-brown-drift-spectral}
\end{equation}
The projectors can be cached and checked against direct exponentiation of Eq.~\eqref{eq:app-brown-exponential}. They solve a control interval, not the entire particle-amplitude contraction: node mixing and the full environmental history remain. In practice, for each desired representation cutoff one computes node matrices by Haar integrating the matrix elements of $D^{(j)}(U)^*$ against the current-frame projector and $U\mathcal P_k[U^\dagger(\cdot)U]U^\dagger$. An SU(2) Euler-angle rule may use periodic quadrature for the two azimuthal angles and Gauss--Legendre quadrature for the cosine of the polar angle. Increasing this quadrature order is a separate check from increasing the representation cutoff.

\subsection{OU: the joint Markov state and Hermite ladder}
For OU driving, $\dot U=-i\sum_c\xi_cF_cU$ and
\begin{equation}
 d\xi_c=-\xi_cdt/\tau+\sqrt{2\sigma^2/\tau}\,dW_c.
\end{equation}
The stationary initial law is Haar in $U$ and independent centered Gaussian of variance $\sigma^2$ in each $\xi_c$. The pair $(U,\boldsymbol\xi)$ is Markov; $U$ alone is not. Relative to the stationary Gaussian measure, the OU differential operator and group-block generator are
\begin{align}
 L_{\rm OU}&=\tau^{-1}\sum_c(-\xi_c\partial_{\xi_c}
 +\sigma^2\partial_{\xi_c}^2),\\
 \mathsf L_{\rm OU}^{(\Lambda)}
 &=L_{\rm OU}\otimes I-i\sum_c\xi_c\otimes\mathsf T_c^{(\Lambda)}.
 \label{eq:app-ou-generator}
\end{align}
An optional drift adds $-i\sum_cd_c\mathsf T_c^{(\Lambda)}$ but is absent from the calculations. In normalized probabilists' Hermite polynomials $\phi_{\boldsymbol q}(\boldsymbol\xi)=\prod_c\mathrm{He}_{q_c}(\xi_c/\sigma)/\sqrt{q_c!}$,
\begin{align}
 L_{\rm OU}\phi_{\boldsymbol q}&=-|\boldsymbol q|\phi_{\boldsymbol q}/\tau,\\
 \xi_c\phi_{\boldsymbol q}&=\sigma\left(\sqrt{q_c+1}\phi_{\boldsymbol q+\boldsymbol e_c}
 +\sqrt{q_c}\phi_{\boldsymbol q-\boldsymbol e_c}\right).
 \label{eq:app-hermite-ladder}
\end{align}
Equations~\eqref{eq:app-ou-generator}--\eqref{eq:app-hermite-ladder} define a sparse matrix: diagonal decay at $-|\boldsymbol q|/\tau$ and neighboring Hermite-level couplings $-i\sigma\mathsf T_c^{(\Lambda)}\sqrt{q_c+1}$ or $-i\sigma\mathsf T_c^{(\Lambda)}\sqrt{q_c}$. For each holding time apply its matrix exponential, then the same current-frame node map as above. The initial stationary noise has only the zero Hermite component; higher levels created during propagation cannot be removed at an intermediate node. At the final integration only the zero Hermite and trivial group modes survive. Increasing both Hermite and group cutoffs checks the result. Eliminating high Hermite levels creates a memory kernel, so the finite-$\tau$ three-axis OU propagator is generally not one scalar exponential. The short-time adjoint correlation starts as $1-\sigma^2\Delta^2+o(\Delta^2)$, in contrast to $1-\kappa\Delta+O(\Delta^2)$ for Brownian motion. A white-noise limit requires a specified scaling of $\sigma^2\tau$; merely refining a time grid does not produce it.

The first two Hermite levels show explicitly what must be retained. Order the levels as $\boldsymbol0,\boldsymbol e_x,\boldsymbol e_y,\boldsymbol e_z$ and put $\mathsf T_d=\sum_cd_c\mathsf T_c^{(j)}$. Truncation at total Hermite order one gives the block matrix
\begin{equation}
 \left.\mathsf L_{\rm OU}^{(j)}\right|_{q_{\max}=1}
 =\begin{pmatrix}
 -i\mathsf T_d&-i\sigma\mathsf T_x&-i\sigma\mathsf T_y&-i\sigma\mathsf T_z\\
 -i\sigma\mathsf T_x&-\tau^{-1}I-i\mathsf T_d&0&0\\
 -i\sigma\mathsf T_y&0&-\tau^{-1}I-i\mathsf T_d&0\\
 -i\sigma\mathsf T_z&0&0&-\tau^{-1}I-i\mathsf T_d
 \end{pmatrix}.
 \label{eq:app-ou-first-block}
\end{equation}
Each entry is a $4(2j+1)$-dimensional joint block for a fixed right representation index. A total-order cutoff $q_{\max}$ contains $\binom{q_{\max}+3}{3}$ Hermite indices for three noise channels. Keeping only the zero level deletes every stochastic coupling; Eq.~\eqref{eq:app-ou-first-block} displays the first nontrivial coupling block. The result is checked by increasing both Hermite and group cutoffs.

For a more precise account of finite-$\tau$ memory, partition the full Hermite generator into zero and nonzero levels, $L_{00},L_{0>},L_{>0},L_{>>}$. Eliminating the latter from the exact linear equations gives
\begin{equation}
 \dot f_0(t)=L_{00}f_0(t)+L_{0>}e^{tL_{>>}}f_{>}(0)
 +\int_0^tds\,L_{0>}e^{(t-s)L_{>>}}L_{>0}f_0(s).
 \label{eq:app-ou-memory}
\end{equation}
The initially stationary state has $f_{>}(0)=0$, but particle nodes generally create nonzero higher levels later. The zero-level corner of $e^{t\mathsf L_{\rm OU}}$ therefore cannot be multiplied across nodes as if it were an independent scalar semigroup. As a single-axis commuting benchmark, a mode of generator eigenvalue $a$ has the exact factor
\begin{equation}
 K_a(t)=e^{-iad_zt}\exp\!\{-a^2\sigma^2[\tau t-\tau^2(1-e^{-t/\tau})]\},
 \label{eq:app-ou-commuting}
\end{equation}
because its integrated OU phase is Gaussian with variance $2\sigma^2[\tau t-\tau^2(1-e^{-t/\tau})]$. For $t\ll\tau$, its envelope begins as $\exp(-a^2\sigma^2t^2/2)$, with a decay scale of order $1/(|a|\sigma)$. For $t\gg\tau$, it has exponential rate $a^2\sigma^2\tau$, giving a scale of order $1/(a^2\sigma^2\tau)$ when that long-time regime is reached. These are coherence scales for one commuting mode, not a total-variation mixing time for the noncommuting joint frame--velocity process. The factor checks the Hermite calculation but does not solve that joint mixing problem. With $\tau\to0$ and $\mathcal D_{\rm OU}:=\sigma^2\tau$ fixed, the integrated velocity covariance tends to $2\mathcal D_{\rm OU}\delta(t)$ and the limiting Stratonovich Brownian strength is $\kappa=2\mathcal D_{\rm OU}$. Holding $\sigma$ fixed instead makes the integrated noise vanish. These are limits of physical models, not consequences of refining the numerical grid.

\subsection{Protocol-specific amplitude algorithms}
\label{alg:control-paths}
Each protocol supplies the \emph{entire} control history to the common contraction in Appendix~\ref{alg:particle-history}. The following steps specify its additional state, its random draw, and the averaging operation. In all four cases the particle times are fixed before drawing a trajectory.
\begin{enumerate}
\item \textbf{Fixed refresh.} Specify the deterministic holding boundaries and draw an independent uniform Pauli at the start of each holding interval, allowing consecutive redraws to coincide. For $S=\sigma_z$ retain its sign $y_j=\pm1$. Split the time axis at the union of holding boundaries and particle times. Apply that sign to both branch coupling eigenvalues on every resulting cell; keep every cross-cell bath link. Average the resulting complete complex amplitudes over sign strings, by exact enumeration when feasible or by sampled strings. Equation~\eqref{eq:app-fixed-signs} is the exact ensemble formula; one must not multiply independently averaged bath factors from different holding intervals.
\item \textbf{Poisson refresh.} Draw the event number with mean $\lambda_r T$ and, conditional on that number, draw and sort uniform event times in $(0,T)$. Independently draw the initial and every replacement Pauli, again permitting identical consecutive frames. Split at the union of actual event times and fixed particle times, integrate each bath-kernel segment at its actual duration, and contract all cross-segment links before averaging complete complex amplitudes. Alternatively, sum the ordered event-time integrals in Eq.~\eqref{eq:app-poisson-events}; their density is $e^{-\lambda_rT}\lambda_r^m$ on the ordered simplex and carries no additional $m!$. Sampling only the event count but replacing actual event times by a regular grid computes a different protocol.
\item \textbf{Brownian frame.} Draw $U_0$ from normalized Haar measure. Draw independent Brownian increments on a control subgrid and propagate the left-multiplicative Stratonovich equation, composing each increment on the left of the current frame. Record one continuous path on a fine subgrid and sample that same path on each coarser bath grid used in a convergence comparison. For every path run Appendix~\ref{alg:particle-history} with $\Pi_s(U_k)$, then average its complex coefficients. A new increment is independent of the past, whereas the successive endpoint frames are not; neither the tensor contraction nor the path sampling replaces them by independent Haar frames.
\item \textbf{OU Hamiltonian frame.} Draw $U_0$ from Haar measure and each $\xi_c(0)$ from its centered stationary Gaussian law of variance $\sigma^2$. Propagate $\boldsymbol\xi$ by its OU transition and propagate $U$ using the time-ordered Hamiltonian $\sum_c\xi_c(t)F_c$ over the same fine subgrid. Preserve the joint state $(U,\boldsymbol\xi)$ throughout the path; drawing a fresh $\boldsymbol\xi$ at each cell changes the model. Freeze the resulting $U$ on the bath cells, contract the complete history, and average complex coefficients. Refine the control step while keeping the chosen $\tau$ and variance scaling explicit.
\end{enumerate}
For the reported single-axis Drude refresh data, the eigenbasis is fixed and only the sign history changes. Consequently the system coefficients for all operation labels and the bath coefficients for each complete sign history can be cached separately. This is the actual refresh calculation; the Brownian and OU data use the full-memory tensor and sampled frame paths above.

\subsection{Alternative deterministic block algorithm}
\label{alg:representation-blocks}
Equations~\eqref{eq:app-brown-exponential} and \eqref{eq:app-ou-generator} permit a different evaluation when the group expansion is tractable. First choose finite representation and, for OU, Hermite cutoffs. Integrate the current-frame node projector and the transformed insertion against representation matrix elements to precompute their finite mixing matrices. Initialize the Haar group mode (and OU zero Hermite mode) together with the full bath-history index. Between nodes apply $\mathsf M_\Delta^{(\Lambda)}$ for Brownian motion, or the sparse joint group--Hermite exponential for OU; at each node apply the precomputed mixing matrix and append all new bath-history links. Retain every resulting mode until the final trace and Haar/noise integral, which selects the trivial group and zero Hermite modes. Increase both cutoffs and the bath-history rank until the requested amplitudes stabilize. The interval exponentials perform the $q_\Delta$ average exactly within a retained block, but no block by itself removes the need to track the bath history. This deterministic route was \emph{not} the solver used for the reported Drude plots.

%% file: appendix_weak.tex
\section{A weak-coupling analytic check for arbitrary fixed particle number}
\label{app:weak}

The exact Gaussian influence functional is an all-coupling result. A separate $g^2$ expansion is useful because it exposes the control correlation that acts between arbitrary particle insertions. Let $t_0=0,t_{n+1}=T$ and put $\mathcal P_{b:a}=\mathcal P_{\nu_b}\cdots\mathcal P_{\nu_a}$, with the identity understood when $b<a$, and $X_j=\mathcal P_{j:1}[\rho_S]$. Expand $S_\omega(t)=\sum_as_{a,\omega}(t)P_a$ in a Hermitian basis, with real $s_{a,\omega}$. Define
\begin{equation}
 K_{ab}(t,u)=\E_\omega[s_{a,\omega}(t)s_{b,\omega}(u)].
 \label{eq:app-weak-K}
\end{equation}
For $t\in(t_i,t_{i+1})$ and $u\in(t_j,t_{j+1})$ with $t>u$, denote the rectangular domain for $i>j$ or triangular domain for $i=j$ by $\Omega_{ij}$. The early interaction vertex acts as $-ig(S_\omega(u)X_j\otimes B(u)\rho_E-X_jS_\omega(u)\otimes\rho_E B(u))$. Acting with the later commutator and tracing the bath gives $C(t,u)$ for the first term and $C(t,u)^*$ for the second; the intervening particle insertions remain between the two vertices. Therefore
\begin{multline}
 \overline{\mathcal C}_{\boldsymbol\nu}^{(n),\mathrm{tr}}
 =\Tr_S\mathcal P_{n:1}[\rho_S]
 -g^2\sum_{i=0}^n\sum_{j=0}^i\int_{\Omega_{ij}}dt\,du
 \sum_{a,b}K_{ab}(t,u)\\
 \times\Tr_S\mathcal P_{n:i+1}\!\left\{
 C(t,u)[P_a,\mathcal P_{i:j+1}[P_bX_j]]
 -C(t,u)^*[P_a,\mathcal P_{i:j+1}[X_jP_b]]\right\}
 +O(g^4).
 \label{eq:app-weak-amplitude}
\end{multline}
This formula retains all $n$ particle legs and uses only one two-time control kernel because there are exactly two interaction vertices at this order. If the later vertex lies after the last particle event, the outer trace of its commutator vanishes, a useful causal check. Higher coupling orders require higher control moments and the full-history algorithm.

The refresh protocols give the first two explicit kernels. For $S=\sigma_z$, a uniform Pauli frame changes only its sign, so $s_{z,\omega}(t)=y_\omega(t)$ and all other coefficients vanish. Under fixed refresh, $K_{zz}(t,u)=1$ when $t$ and $u$ belong to the same holding interval and is zero when a refresh boundary separates them. Under Poisson refresh, the probability of no intervening redraw gives $K_{zz}(t,u)=e^{-\lambda_r(t-u)}$ for $t\ge u$. Thus the fixed-refresh $Z_{ij}^{zz}$ is obtained by integrating $C(t,u)$ only over portions of $\Omega_{ij}$ within a common holding interval. For Poisson refresh and $C(v)=\sum_m c_me^{-\nu_mv}$, it is $\sum_m c_mF_{ij}^{\rm event}(\nu_m+\lambda_r)$. These are second-order checks; the exact refresh averages retain all higher control moments and bath links in Eqs.~\eqref{eq:app-fixed-signs} and \eqref{eq:app-poisson-events}.

For $S=\sigma_z$ and an initial Haar frame with isotropic drift-free qubit Brownian motion, common-frame Haar averaging gives $K_{ab}(t,u)=\delta_{ab}e^{-\kappa(t-u)}/3$ for $a,b=x,y,z$. With the Drude expansion $C(v)=\sum_m c_me^{-\nu_mv}$, the only nonzero scalar coefficients of Eq.~\eqref{eq:app-weak-amplitude} are
\begin{equation}
 Z_{ij}^{ab}=\int_{\Omega_{ij}}C(t,u)K_{ab}(t,u)dt\,du
 =\frac{\delta_{ab}}3\sum_m c_mF_{ij}^{\rm event}(\nu_m+\kappa),
 \label{eq:app-weak-brown}
\end{equation}
where $F^{\rm event}$ is Eq.~\eqref{eq:app-Fkl} evaluated on the particle-event intervals. Thus the Brownian damping shifts every bath pole by $\kappa$ in this second-order coefficient, for any fixed $n$. This is not a universal multiplicative damping of each complete particle amplitude.

For optional constant Brownian drift $h\boldsymbol n\cdot\boldsymbol F$, let $\theta$ be the angle between $\boldsymbol n$ and the original $z$ coupling direction. Haar initialization and isotropic diffusion give
\begin{equation}
 K_{ab}(t,u)=\frac{\delta_{ab}}3e^{-\kappa v}
 [\cos^2\theta+\sin^2\theta\cos(hv)],\qquad v=t-u\ge0.
 \label{eq:app-weak-drift-kernel}
\end{equation}
Thus each $F_{ij}^{\rm event}(\nu_m+\kappa)$ in Eq.~\eqref{eq:app-weak-brown} becomes
\begin{equation}
 \cos^2\theta F_{ij}^{\rm event}(\nu_m+\kappa)
 +\frac{\sin^2\theta}{2}
 \bigl[F_{ij}^{\rm event}(\nu_m+\kappa+ih)
 +F_{ij}^{\rm event}(\nu_m+\kappa-ih)\bigr].
 \label{eq:app-weak-drift-poles}
\end{equation}
This is a second-order control-kernel result; the reported simulations have no deterministic drift. The bath coefficients $c_m$ retain their complex quantum part.

For OU control, let $L_{\rm OU}^{\rm ad}$ be the Hermite ladder of Eq.~\eqref{eq:app-ou-generator} restricted to the spin-one adjoint representation, and let $v_0=|\boldsymbol0\rangle\otimes e_z$. Its stationary two-time autocorrelation is
\begin{equation}
 k_{\rm OU}(v)=v_0^\dagger e^{vL_{\rm OU}^{\rm ad}}v_0,
 \qquad K_{ab}(t,u)=\frac{\delta_{ab}}3k_{\rm OU}(t-u).
 \label{eq:app-weak-ou}
\end{equation}
After a finite converged Hermite truncation, put $H_m=\nu_m I-L_{\rm OU}^{\rm ad}$. Analytic functional calculus in Eq.~\eqref{eq:app-Fkl} then gives
\begin{equation}
 Z_{ij}^{ab}=\frac{\delta_{ab}}3\sum_m c_m
 v_0^\dagger F_{ij}^{\rm event}(H_m)v_0.
 \label{eq:app-weak-ou-matrix}
\end{equation}
For a diagonal event interval of length $a$, $F(H)=H^{-2}(Ha-I+e^{-Ha})$; a cross interval is the analogous product of two $(I-e^{-Ha})$ factors and its gap exponential. Near zero eigenvalues these functions are evaluated by stable $\varphi$-functions or series, not by subtracting nearly equal matrices after an explicit inverse. Hermite and Matsubara cutoffs are increased independently. The finite-$\tau$ OU coefficient generally has several rates; replacing it by one Brownian exponential is justified only after a specified white-noise scaling limit.

\subsection{A static-frame and terminal-observable benchmark}
\label{app:static-benchmark}
These checks deliberately measure a final system observable, rather than calling the closed zero-particle coefficient a decay curve. Let $A_0(T)=\int_0^T(T-v)\operatorname{Re}C(v)dv$. For a static but Haar-random coupling axis, the exact Gaussian dephasing factor perpendicular to that axis is $q(T)=e^{-4g^2A_0(T)}$. Averaging axes sends an input Bloch vector $\boldsymbol r_0$ to
\begin{equation}
 \overline\rho_{\rm static}(T)=\frac12\left[
 \id+\frac{1+2e^{-4g^2A_0(T)}}3
 \boldsymbol r_0\cdot\boldsymbol\sigma\right].
 \label{eq:app-static-haar}
\end{equation}
The residual $1/3$ Bloch factor is an exact arbitrary-coupling check and shows why a single Haar draw is not rapid decoupling. For the uncut Drude bath, $A_0$ is evaluated from the integrated Matsubara kernel; its divergent point value $C(0)$ is not used.

For a moving isotropic Haar-initialized frame, a separate $g^2$ expansion of the final \emph{open} state gives
\begin{equation}
 \overline\rho_X(T)=\rho_S-\frac{g^2A_X(T)}3
 \sum_{a=x,y,z}[\sigma_a,[\sigma_a,\rho_S]]+O(g^4),
 \quad A_X(T)=\int_0^T(T-v)\operatorname{Re}C(v)k_X(v)dv,
 \label{eq:app-weak-terminal}
\end{equation}
where $k_X(v)=e^{-\kappa v}$ for Brownian motion and $k_X(v)=k_{\rm OU}(v)$ from Eq.~\eqref{eq:app-weak-ou} for OU driving. This isotropic formula does not apply unchanged to Pauli refresh, whose coupling remains on the fixed $z$ axis; its kernel is given above. For Brownian motion, $A_{\rm B}(T)=\operatorname{Re}\sum_mc_mF_{00}(\nu_m+\kappa)$ on the full interval. OU driving instead uses $v_0^\dagger F_{00}(\nu_m I-L_{\rm OU}^{\rm ad})v_0$. The imaginary bath response cancels in this symmetric final-state check but not in general off-diagonal multi-particle amplitudes. Equation~\eqref{eq:app-weak-terminal} is a weak-coupling terminal-state expansion. Arbitrary-coupling observables are evaluated from the full-history influence functional, with no additional Markov approximation or post hoc clipping.

%% file: appendix_physicality.tex
\section{Positivity, causality, finite-comb recovery, and single-trajectory transparency}
\label{app:physicality}

\subsection{Single-trajectory transparency with no independent $H_S$}
The laboratory Hamiltonian used in this paper is
\begin{equation}
 H_{\rm lab}^{\omega}(t)=H_E+H_{SE}+H_{\rm ctrl}^{\omega}(t)+H_{\rm inst}^{\omega}(t).
 \label{eq:app-lab-hamiltonian}
\end{equation}
There is no extra system drift $H_S$. A target logical instrument with Hamiltonian $h_{\rm inst}(t)$ is implemented in the laboratory by its frame-covariant version. Let $U_\omega(t)$ be the control frame, including its initial Haar draw, and choose the control Hamiltonian so that its propagator satisfies $\dot U_\omega=-iH_{\rm ctrl}^{\omega}U_\omega$. For $H_{SE}=0$, the compiled system propagator through an instrument interval obeys
\begin{equation}
 \dot V_\omega(t)=-i\left[H_{\rm ctrl}^{\omega}(t)
 +U_\omega(t)h_{\rm inst}(t)U_\omega^\dagger(t)\right]V_\omega(t).
 \label{eq:app-compiled-propagator}
\end{equation}
Set $W_\omega(t)=U_\omega^\dagger(t)V_\omega(t)$. Using the product rule and $\dot U_\omega^\dagger=iU_\omega^\dagger H_{\rm ctrl}^{\omega}$,
\begin{align}
 \dot W_\omega
 &=iU_\omega^\dagger H_{\rm ctrl}^{\omega}V_\omega
 +U_\omega^\dagger\dot V_\omega\\
 &=-i h_{\rm inst}(t)W_\omega.
 \label{eq:app-transparent-ode}
\end{align}
Thus $W_\omega(t)$ is exactly the target logical propagator $W_{\rm id}(t)$, up to the common boundary value. If the input is encoded as $\rho_{\rm lab}(0)=U_\omega(0)\rho_{\rm log}U_\omega^\dagger(0)$, then at the end of each interval
\begin{equation}
 \rho_{\rm lab}(t)=U_\omega(t)\rho_{\rm log}(t)U_\omega^\dagger(t),
 \qquad
 \rho_{\rm log}(t)=W_{\rm id}(t)\rho_{\rm log}(0)W_{\rm id}^\dagger(t).
 \label{eq:app-transparent-state}
\end{equation}
For fixed and Poisson refresh, the frame is constant between jumps. At either kind of jump $U^-\mapsto U^+$, the laboratory update is $J=U^+(U^-)^\dagger$, so $\rho_{\rm lab}^+=J\rho_{\rm lab}^-J^\dagger=U^+\rho_{\rm log}(U^+)^\dagger$. The logical state is unchanged. A pulse of nonzero duration requires pointwise covariance throughout its support and the separate finite-pulse analysis below. For Brownian driving, the same cancellation follows from the Stratonovich product rule; $H_{\rm ctrl}^{\omega}$ then denotes a stochastic generator, not a bounded pointwise Hamiltonian. For OU driving, the equations are ordinary differential equations along each continuous velocity path. Thus the transparency proof applies to each protocol without assuming differentiable Brownian paths or instantaneous physical pulses.
The same calculation works for a general covariant completely positive, trace-preserving (CPTP) instrument: if $\mathcal A_t$ is its logical infinitesimal map, define
\begin{equation}
 \mathcal A_t^\omega[X]=U_\omega(t)\mathcal A_t[U_\omega^\dagger(t)XU_\omega(t)]U_\omega^\dagger(t).
 \label{eq:app-covariant-map}
\end{equation}
At every intermediate time, applying $U_\omega^\dagger(t)$ maps the laboratory state and the compiled instrument back to the same logical pair $(\rho_{\rm log},\mathcal A_t)$. Therefore the complete composition of finite-duration maps is the target composition for each fixed $\omega$, not only after an ensemble average. Endpoint covariance alone would not imply this cancellation when the frame varies inside the pulse.

For a fixed laboratory POVM $M_x$, the physical readout that returns the logical probability is
\begin{equation}
 M_{x,\omega}^{\rm lab}=U_\omega(T)M_xU_\omega^\dagger(T),\qquad
 p_\omega(x)=\Tr[M_{x,\omega}^{\rm lab}\rho_{\rm lab}(T)]
 =\Tr[M_x\rho_{\rm log}(T)].
 \label{eq:app-transparent-readout}
\end{equation}
If the final boundary is a total trace, the same result follows without an explicit detector correction by cyclicity. A fixed uncorrected POVM instead measures a random-frame observable and is not the logical experiment. This distinction is why the finite-duration numerical cases use a fixed logical POVM together with the inverse-frame or equivalent covariant readout.

\subsection{Constructive positivity of the continuous Choi object}
The process-side Choi map is not obtained by identifying the process Fock space with an operator space through a local rule for every field species. As shown in Ref.~\cite{fuksa2026dilation}, the continuum Fock space does not factorize along the forward/backward contour, and the cross-species amplitudes cannot in general be inserted as ordinary Choi-operator kernels. We use the process map $\mathsf P$ of that work,
\[
 \mathsf P:\Gamma_T^{(d_S^4-1)}\longrightarrow
 \mathrm{HS}\!\left(\gamma_T\right),
 \qquad \gamma_T:=\Gamma_T^{(d_S^2-1)},
\]
which discards cross-species wavefunctions and reshapes the pure forward and backward amplitudes into an operator kernel on the Choi Fock space. For ordered particle sectors this kernel is
\begin{equation}
 [\mathsf P|\Upsilon_T\rangle]^{(n,n')}_{\boldsymbol\mu;\boldsymbol\mu'}
 (\boldsymbol t;\boldsymbol t')
 =\mathcal C^{(n+n')}_{\boldsymbol\nu}
 (\operatorname{sort}(\boldsymbol t,\boldsymbol t')),
 \label{eq:app-choi-process-map}
\end{equation}
where the sorted label string $\boldsymbol\nu$ assigns $(\mu_j,0)$ to each forward time $t_j$ and $(0,\mu'_k)$ to each backward time $t'_k$, in their respective temporal orders. The resulting kernel is Hilbert--Schmidt by construction. The recovery map has domain $\mathcal S\subset\mathfrak S_1(\gamma_T)$ and codomain $\Gamma_T^{(d_S^4-1)}$; on physical cPTs it satisfies $\mathsf R\mathsf P|\Upsilon_T\rangle=|\Upsilon_T\rangle$. The instrument-side map, denoted $\mathsf V$, acts on instrument effects and retains cross-species information; its output is generally a sesquilinear kernel form rather than an operator. The process and instrument are paired through the kernel pairing of Ref.~\cite{fuksa2026dilation}; in particular, $\mathsf P$, $\mathsf R$, and $\mathsf V$ have distinct domains and roles, and $\mathsf V$ is not used in the Gram-positivity argument below.

For a constructive proof, let $\widetilde\Omega$ be the Choi vacuum and define
the free logical Hamiltonian by
\begin{equation*}
 H_{\log,\omega}(t)=\id_S\otimes H_E(t)+
 (U_\omega^\dagger(t)\otimes\id_E)H_{SE}(t)
 (U_\omega(t)\otimes\id_E).
\end{equation*}
For a process-canonical representation, the positive operator $\mathsf P|\Upsilon_{T,\omega}\rangle$ has the following single-branch Gram construction, equivalent to Eqs.~(46)--(48) of Ref.~\cite{wassnerCPT} on the regular process domain and to the process map of Ref.~\cite{fuksa2026dilation}. The field regularization and generator are the same in both branches:
\begin{equation}
 \mathsf K_{T,\omega}=\mathcal T\exp\!\int_0^Tdt\,
 \left[-iH_{\log,\omega}(t)\otimes\id_F
 +\sum_{\mu>0}(P_\mu\otimes\id_E)\otimes\phi_\mu^\dagger(t)\right].
 \label{eq:app-choi-single-branch}
\end{equation}
It is a field-generating operator, not a unitary on the field. Expanding the
two branches shows that its reduced Gram kernel contains precisely the pure
forward/backward amplitudes in Eq.~\eqref{eq:app-choi-process-map}; it does
not insert cross-species amplitudes as a separate diagonal field block. Thus,
on this regular process domain, the trace-closed process Choi operator is
\begin{equation}
 \widehat\Upsilon_{T,\omega}^{\rm tr}
 =\Tr_{SE}\left[\mathsf K_{T,\omega}
 \bigl(\rho_{SE}\otimes|\widetilde\Omega\rangle\langle\widetilde\Omega|\bigr)
 \mathsf K_{T,\omega}^\dagger\right].
 \label{eq:app-choi-positive}
\end{equation}
To expose positivity, write $\rho_{SE}=\sum_i\lambda_i|i\rangle\langle i|$
and use an orthonormal basis $\{|j\rangle\}$ of the final $SE$ space:
\begin{equation}
 |u_{ji,\omega}\rangle=(\langle j|\otimes\id_F)\,
 \mathsf K_{T,\omega}(|i\rangle\otimes|\widetilde\Omega\rangle).
 \label{eq:app-choi-vectors}
\end{equation}
\begin{equation}
 \widehat\Upsilon_{T,\omega}^{\rm tr}
 =\sum_{i,j}\lambda_i|u_{ji,\omega}\rangle\langle u_{ji,\omega}|\succeq0.
 \label{eq:app-choi-positive-sum}
\end{equation}
The vectors in this positive decomposition need not be orthogonal; it is
not asserted to be a spectral decomposition. For every allowed Fock test
vector, the trajectory average therefore satisfies
\begin{equation}
 \langle\varphi|\overline{\widehat\Upsilon}_T^{\rm tr}|\varphi\rangle
 =\mathbb E_\omega\sum_{i,j}\lambda_i
 |\langle\varphi|u_{ji,\omega}\rangle|^2\ge0.
 \label{eq:app-choi-average-positive}
\end{equation}
To verify the identification and recovery explicitly, write $V_\omega(t,s)$
for the free joint propagator and, for a forward history, define
\begin{equation}
 A_{\boldsymbol\mu,\omega}(\boldsymbol t)
 =V_\omega(T,t_n)(P_{\mu_n}\otimes\id_E)
 V_\omega(t_n,t_{n-1})\cdots
 (P_{\mu_1}\otimes\id_E)V_\omega(t_1,0).
 \label{eq:app-choi-branch-word}
\end{equation}
Its empty history is $V_\omega(T,0)$. The $(n,n')$ block of the Gram
operator is exactly
\begin{equation}
 \Tr_{SE}\!\left[
 A_{\boldsymbol\mu,\omega}(\boldsymbol t)\rho_{SE}
 A_{\boldsymbol\mu',\omega}(\boldsymbol t')^\dagger\right],
 \label{eq:app-choi-gram-kernel}
\end{equation}
which is the coefficient obtained by interleaving the corresponding pure
left and pure right insertions. This proves
$\widehat\Upsilon_{T,\omega}^{\rm tr}=\mathsf P|\Upsilon_{T,\omega}\rangle$.
For a general label history $\nu_j=(\mu_j,\mu'_j)$, the forward word
contains every nonzero $\mu_j$ and the backward word every nonzero
$\mu'_j$, each at its original time. The recovery map $\mathsf R$
reconstructs these coefficients from the process Choi operator as in
Appendix~H, Eq.~(H8), of Ref.~\cite{fuksa2026dilation}. For the regular
Gram kernels generated here, strong continuity of the joint propagator
permits the equal-time restriction of Eq.~(H9) in that reference: the
left and right insertion at a shared time compose to
$X\mapsto P_{\mu_j}XP_{\mu'_j}^\dagger$. This establishes
$\mathsf R\mathsf P|\Upsilon_{T,\omega}\rangle
=|\Upsilon_{T,\omega}\rangle$ on this regular domain. For a general
physical cPT, recovery uses the operator-domain definition of
$\mathsf R$, rather than pointwise evaluation of an arbitrary $L^2$
kernel. In neither case is a separate delta-function block added to
the process Choi operator.
If $p=d^2-1$ and $L_P=\max_{\mu>0}\|P_\mu\|_\infty$, unitarity and
normalization give
\begin{equation}
 \Tr_F\widehat\Upsilon_{T,\omega}^{\rm tr}
 \le\sum_{n\ge0}\frac{(pL_P^2T)^n}{n!}
 =e^{pL_P^2T}.
 \label{eq:app-choi-trace-bound}
\end{equation}
Thus the Gram construction is trace class with a path- and cutoff-independent
bound, in the fixed basis convention used here.
No product-state assumption is needed for this Gram factorization. Retaining
the final system and tracing only the environment gives the corresponding
open-output Choi object; tracing the final system as well gives the
trace-closed tester. For an infinite bosonic bath, first apply the construction
to physical regularized processes. The cutoff limit in Theorem~\ref{thm:fock-convergence}
is also a physical process: for every finite number of slots, smearing the
coefficient kernels against nonnegative normalized $L^2$ test functions is a
continuous linear functional of the process vector. Identity slots reduce to
lower-order kernels by integration over the finite time interval. Hence the
finite-slot positivity conditions, the affine causal constraints, and the
normalized vacuum condition pass to the Fock-norm limit, since the finite
positive cones and causal affine subspaces are closed. The bounded process
map $\mathsf P$ carries Fock-norm convergence to Hilbert--Schmidt
convergence of the Choi operators. Since the Hilbert--Schmidt norm
bounds the operator norm, every nonnegative quadratic form of the
regularized Choi operators remains nonnegative in the limit. Fatou's
lemma applied to their diagonal entries in any orthonormal basis,
together with Eq.~\eqref{eq:app-choi-trace-bound}, also makes the
limiting positive operator trace class. Thus positivity passes through
convergence of the process Choi images, without assuming point evaluation
or continuity of recovery on arbitrary kernels.

The positivity proof uses the process map $\mathsf P$ and its single-branch
Gram representation. A transfer expression for an instrument belongs to
$\mathsf V$ and must be paired with $\mathsf P$ through the kernel pairing;
it is not a second local definition of the process Choi operator. Finite
truncations of a Fourier, Hermite, or bath-history representation likewise
require their own positivity checks.

\subsection{Causality and finite-comb marginals}
Let $\ell[X]=\Tr_{SE}$ and let $Q_\omega(t)$ be the trace-preserving no-insertion generator. Then
\begin{equation}
 \ell Q_\omega(t)=0.
 \label{eq:app-trace-identities}
\end{equation}
The left/right Pauli maps used in the amplitude and instrument calculations are
\begin{equation}
 \mathbb P_{\mu\mu'}[X]=(P_\mu\otimes\id_E)X
 (P_{\mu'}^\dagger\otimes\id_E).
 \label{eq:app-physical-insertion}
\end{equation}
They must be distinguished from the rank-one Liouville maps used in the
causality basis of Ref.~\cite{wassnerCPT}. With
$\langle\!\langle P_\alpha|P_{\alpha'}\rangle\!\rangle
=d_S\delta_{\alpha\alpha'}$, define
\begin{equation}
 \widetilde R_{\alpha\alpha'}[X]
 =\frac{1}{d_S}(P_\alpha\otimes\id_E)
 \Tr_S\!\left[(P_{\alpha'}^\dagger\otimes\id_E)X\right].
 \label{eq:app-rank-one-insertion}
\end{equation}
This is the action of
$d_S^{-1}|P_\alpha\rangle\rangle\langle\!\langle P_{\alpha'}|$
on the system Liouville factor. Consequently,
\begin{equation}
 (\ell\widetilde R_{\alpha\alpha'})(X)
 =\frac{\Tr P_\alpha}{d_S}
 \Tr_E\!\left[\Tr_S\!\left[(P_{\alpha'}^\dagger\otimes\id_E)X\right]\right]=0
 \quad(P_\alpha\ne\id_S).
 \label{eq:app-rank-one-trace}
\end{equation}
By contrast, $\ell\mathbb P_{\alpha\alpha}=\ell$ for an unscaled Pauli
operator. The zero-trace relation therefore applies to $\widetilde R$, not to
the physical left/right map $\mathbb P$.

The cPT causality condition now follows the construction of Ref.~\cite{wassnerCPT}.
The basis change from $\mathbb P_{\mu\mu'}$ to $\widetilde R_{\alpha\alpha'}$
induces the corresponding field transformation
$\xi_{\mu\mu'}\mapsto\widetilde\xi_{\alpha\alpha'}$. Let
$\mathsf P_{(t,T)}$ project onto the subspace with no field particles after
$t$. Trace preservation of the process-canonical no-insertion evolution,
together with Eq.~\eqref{eq:app-rank-one-trace}, gives
\begin{equation}
 \langle\Omega|\Upsilon_{T,\omega}\rangle=1,
 \qquad
 \mathsf P_{(t,T)}\widetilde\xi_{\alpha\alpha'}(t)
 |\Upsilon_{T,\omega}\rangle=0
 \quad(P_\alpha\ne\id_S).
 \label{eq:app-fock-causality}
\end{equation}
This is the representation-independent continuum causality condition used in
the reference construction. The frame protocol changes $Q_\omega$ and the
coefficient functions, but not the trace-preserving argument. Linearity then
preserves the condition under trajectory averaging.

For a finite set of operation slots, let $W_{\boldsymbol t}^{\rm op}$ denote
the Choi operator obtained from the same joint dilation while retaining the
final system output. As in Sec.~\ref{sec:background}, $i_j$ is the system
returned by the preceding intervention and $o_j$ is the terminal output of
the $j$th open prefix. Trace preservation of the ensuing joint propagation
gives
\begin{equation}
 \Tr_{o_j}W_{j:0}^{\rm op}
 =\id_{i_j}\otimes W_{j-1:0}^{\rm op},
 \label{eq:app-finite-comb-recursion}
\end{equation}
Tracing $o_j$ leaves the shorter open comb with terminal output $o_{j-1}$.
Closing the terminal output $S_f$ of the full comb with a deterministic trace
effect produces the tester
\begin{equation}
 W_{\boldsymbol t}^{\rm tr}=\Tr_{S_f}W_{\boldsymbol t}^{\rm op}.
 \label{eq:app-closed-tester}
\end{equation}
For any sequence of deterministic CPTP slot maps $\mathcal A_j$, it obeys
\begin{equation}
 \Tr\!\left[W_{\boldsymbol t}^{\rm tr}
 \left(\bigotimes_{j=1}^k J(\mathcal A_j)\right)^{\mathsf T}\right]=1.
 \label{eq:app-tester-normalization}
\end{equation}
The same normalization holds for a complete adaptive instrument after summing
all its outcomes. This is a positive, trace-closed tester, not an open-output
comb. Its causal normalization is inherited from
Eq.~\eqref{eq:app-finite-comb-recursion}.
Finite-duration covariant instrument contraction commutes with finite-marginal
recovery because it acts on the full time interval, including all intermediate
system maps.

To make this recovery explicit, choose distinct fixed times $t_1<\cdots<t_k$ and extend the operation labels by $\xi_{00}(t)=\id$ while $\xi_{\mu\mu'}(t)$ for $(\mu,\mu')\ne(0,0)$ extracts the corresponding nonvacuum field coefficient. Then
\begin{equation}
 C_{\boldsymbol\mu,\boldsymbol\mu'}^{\rm tr}(\boldsymbol t)
 =\langle\Omega|\xi_{\mu_k\mu'_k}(t_k)\cdots
 \xi_{\mu_1\mu'_1}(t_1)|\Upsilon_T^{\rm tr}\rangle.
 \label{eq:app-finite-coefficients}
\end{equation}
The double-zero label joins neighboring propagators by the composition law and adds no field particle. Thus the highest $k$-particle sector alone does not reconstruct the finite comb: all choices of identity slots contribute. In the column-vectorization convention $|P_\mu\rangle\rangle$, where $\langle\!\langle P_\mu|P_{\mu'}\rangle\!\rangle=d\delta_{\mu\mu'}$, the closed finite-slot Choi operator is
\begin{equation}
 W_{\boldsymbol t}^{\rm tr}
 =d^{-2k}\sum_{\boldsymbol\mu,\boldsymbol\mu'}
 C_{\boldsymbol\mu,\boldsymbol\mu'}^{\rm tr}(\boldsymbol t)
 \bigotimes_{j=1}^k|P_{\mu_j}^*\rangle\!\rangle
 \langle\!\langle P_{\mu'_j}^*|.
 \label{eq:app-finite-choi}
\end{equation}
Indeed $J(\mathcal P_{\mu\mu'})=|P_\mu\rangle\!\rangle\langle\!\langle P_{\mu'}|$, so contraction of Eq.~\eqref{eq:app-finite-choi} with the tensor of $J(\mathcal P_{\mu\mu'})^{\mathsf T}$ returns Eq.~\eqref{eq:app-finite-coefficients}; the factor $d^{-2k}$ follows from the $k$ Hilbert--Schmidt basis inner products. The final system has already been traced. An open-output comb retains that factor separately.

The coefficients have a positive Gram representation. Purify the initial
$SE$ state and, for each forward operation string $\boldsymbol\mu$, propagate
its ket branch through the complete joint dynamics to obtain
$|v_{\boldsymbol\mu}\rangle$. Then
$C_{\boldsymbol\mu,\boldsymbol\mu'}^{\rm tr}
=\langle v_{\boldsymbol\mu'}|v_{\boldsymbol\mu}\rangle$. Inserting this Gram
matrix into Eq.~\eqref{eq:app-finite-choi} proves positivity of the closed
tester, while the open-output Gram operator proves positivity of
$W_{\boldsymbol t}^{\rm op}$. Removing an intermediate time means linking its
slot with the identity channel, not taking an arbitrary partial trace of its
Choi legs. Because field extraction, basis conversion, and identity linking
are linear, recovery also commutes with averaging complete control
trajectories. Point-time extraction assumes the physically generated kernels
admit values at the chosen distinct times; an arbitrary abstract $L^2$ Fock
vector has no automatic point value. Finite-duration instruments are instead
contracted over their full time supports.

\subsection{Independent-interval classicalization and its limits}
\label{app:classicalization}
Assume the initial state factorizes as
$\rho_{SE}(0)=\rho_S\otimes\rho_E$ and is independent of all sampled frames.
Let $W_j$ be the full $SE$ unitary on prescribed interval $j$, independent of
the Pauli frame subsequently drawn for that interval. Expand it in the
complete system Pauli basis as
\begin{equation}
 W_j=\sum_aP_a\otimes B_{j,a},\qquad
 B_{j,a}=d^{-1}\Tr_S[(P_a^\dagger\otimes\id_E)W_j].
 \label{eq:app-interval-pauli-expansion}
\end{equation}
The $B_{j,a}$ are \emph{environmental expansion operators} for a whole finite interval, not the single bath coupling operator $B$. The protocol assumes that one independent Pauli frame $G_j$ conjugates the entire interval unitary and that its boundary operation slots are compiled covariantly. Thus $G_j^\dagger P_aG_j=\chi_a(G_j)P_a$. Expanding the complete forward and backward process histories before averaging gives
\begin{equation}
 \E\left[\prod_j\chi_{a_j}(G_j)\chi_{b_j}(G_j)\right]
 =\prod_j\delta_{a_jb_j}.
 \label{eq:app-pauli-proof}
\end{equation}
The surviving history $\boldsymbol a=(a_1,\ldots,a_K)$ has weight
\begin{equation}
 p(\boldsymbol a)=\Tr_E[B_{K,a_K}\cdots B_{1,a_1}\rho_E
 B_{1,a_1}^\dagger\cdots B_{K,a_K}^\dagger]\ge0.
 \label{eq:app-classical-history-prob}
\end{equation}
Unitarity of $W_j$ and Pauli orthogonality give $\sum_aB_{j,a}^\dagger B_{j,a}=\id_E$. Summing Eq.~\eqref{eq:app-classical-history-prob} from its latest label backward therefore yields $\sum_{\boldsymbol a}p(\boldsymbol a)=1$. For arbitrary operations $\mathcal A_j$ at the specified boundaries, the averaged process has the same action as
\begin{equation}
 \sum_{\boldsymbol a}p(\boldsymbol a)
 \mathcal P_{a_K}^{\rm err}\circ\mathcal A_{K-1}
 \circ\mathcal P_{a_{K-1}}^{\rm err}\circ\cdots
 \circ\mathcal A_1\circ\mathcal P_{a_1}^{\rm err}[\rho_S],
 \quad \mathcal P_a^{\rm err}[X]=P_aXP_a^\dagger.
 \label{eq:app-classical-history-action}
\end{equation}
This is process-level equality on the stated slots, rather than equality of one output state. The associated interval-pair Choi factors $|P_a\rangle\!\rangle\langle\!\langle P_a|/d$ are mutually orthogonal, so the inter-interval error history is diagonal. The weights generally fail to factorize because Eq.~\eqref{eq:app-classical-history-prob} propagates one environment through all intervals. The individual input--output pair can still transmit quantum information; ``classical history'' concerns correlations among the interval error labels.

For a correlated initial state, the character identity still eliminates
unequal forward--backward labels. The initial boundary of a surviving history
is instead the positive system operator
$\tau_{\boldsymbol a}=\Tr_E[(\id_S\otimes B_{\boldsymbol a})\rho_{SE}(0)
(\id_S\otimes B_{\boldsymbol a}^\dagger)]$, where
$B_{\boldsymbol a}=B_{K,a_K}\cdots B_{1,a_1}$.
In general $\tau_{\boldsymbol a}$ is not proportional to a common $\rho_S$,
so the scalar history mixture in
Eq.~\eqref{eq:app-classical-history-action} requires the stated product
initial state.

For a full-system Haar twirl on one paired Choi factor, $U\otimes\overline U$ splits into a one-dimensional identity sector $\Pi_0$ and an irreducible, multiplicity-free adjoint sector $\Pi_{\rm ad}=\id-\Pi_0$. Schur orthogonality gives the local conditional expectation
\begin{equation}
 \mathsf T_{\rm H}[X]=\Tr(\Pi_0X)\Pi_0
 +\frac{\Tr(\Pi_{\rm ad}X)}{d^2-1}\Pi_{\rm ad}.
 \label{eq:app-haar-interval-twirl}
\end{equation}
When $X$ also has other interval and environmental indices, the two traces here are partial traces and leave operators on those other factors; they are not prematurely independent scalar probabilities. Applying independent copies of Eq.~\eqref{eq:app-haar-interval-twirl} to every interval, then expanding $\Pi_{\rm ad}$ in a fixed orthonormal operator basis, again gives a diagonal interval-error history. For a different group representation $\bigoplus_\lambda V_\lambda\otimes M_\lambda$, a Haar twirl acts trivially on each multiplicity space $M_\lambda$. Coherence in those spaces can remain, so the conclusion is not valid for every group twirl or every pairing of temporal legs.

If two operation cuts lie inside one holding interval, the same $G_j$ appears on both pieces, and two nontrivial characters can multiply to one instead of vanishing. For Poisson refreshes Eq.~\eqref{eq:app-pauli-proof} holds conditional on a realized schedule at event-aligned slots; at two predetermined times separated by $\Delta$, no refresh occurs with probability $e^{-\lambda_r\Delta}$, leaving that nontrivial character correlation. Different event schedules first require a common fixed-slot cPT or an explicit classical schedule register before their finite combs can be averaged. These correlations demonstrate why the independent-interval proof does not apply at arbitrary continuous cuts; they alone do not prove that a particular resulting process has quantum temporal entanglement.

The premise concerns each \emph{whole} interval. Independent endpoint samples do not ensure that a continuously driven interval propagator is conjugated by one sampled frame. For a concrete comparison, take $H_{SE}=g\sigma_z^S\otimes\sigma_z^E$ and $\rho_E=\id_E/2$. Twirling the fixed-duration unitary as one Haar-conjugated block gives the depolarizing Bloch multiplier $[1+2\cos(2g\Delta)]/3$. Driving the coupling continuously with a drift-free isotropic Brownian frame instead gives the uncoupled channel as $\kappa\to\infty$ at fixed $g,\Delta$. The bounded-interaction estimate of Appendix~\ref{app:bounds} vanishes because the control mixing time scales as $1/\kappa$. These limits differ whenever the block multiplier is not one. OU paths also retain neighboring-frame correlations at finite strength. Refining either path's numerical mesh does not create independent Haar frames.

\subsection{Finite-width implementation of an ideal refresh}
\label{app:pulse-width}
The piecewise-constant protocol uses an ideal instantaneous frame update at each boundary. To compare it with a physical pulse, move to the pulse-control frame and let $H_{\rm err}(t)$ be the Hamiltonian term omitted by the ideal-jump description, including any untracked interaction and instrument mismatch. If $\|H_{\rm err}(t)\|\le g_{\rm p}$ uniformly over the allowed experiment class on the union $P$ of pulse supports, Duhamel's formula for the joint unitary, followed by trace-norm contraction of all later channels, gives the conservative complete-experiment estimate
\begin{equation}
 \|p_{\rm pulse}-p_{\rm jump}\|_{\rm TV}
 \le 2\int_Pdt\,\|H_{\rm err}(t)\|
 \le 2g_{\rm p}|P|.
 \label{eq:app-pulse-error}
\end{equation}
The factor two is a convenient upper bound for the state trace-norm difference and is not an optimized probability constant. The norm $g_{\rm p}$ concerns the \emph{difference after the intended pulse has been factored out}; it should not be identified with the interaction bound $g_*$ in Eq.~\eqref{eq:static-record-bound}. If the only mismatch is between the pulse-frame and ideal-jump conjugates of $H_{SE}$, then $g_{\rm p}\le2g_*$ follows by the triangle inequality, while additional instrument or control errors need separate bounds. A short pulse relative to a well-defined bath correlation time $\tau_E$ helps preserve the intended holding-interval picture, but does not by itself make $g_{\rm p}|P|$ small. Its instrument action must remain pointwise covariant, or its mismatch must be included in $H_{\rm err}$. For an unbounded bosonic field Eq.~\eqref{eq:app-pulse-error} needs a cutoff or energy-constrained replacement.

If $N_{\rm p}$ nonoverlapping fixed-refresh pulses each have width at most $\delta$, Eq.~\eqref{eq:app-pulse-error} gives $2g_{\rm p}N_{\rm p}\delta$. Adding this to the ideal-jump bound yields Eq.~\eqref{eq:fixed-finite-pulse-bound}. For Poisson rate $\lambda_r$, $m\sim\operatorname{Pois}(\lambda_rT)$ and at most $m+1$ pulse supports give the averaged estimate $2g_{\rm p}\delta(\lambda_rT+1)$ on schedules whose pulses can be implemented without overlap. The $+1$ covers an initial frame-setting pulse. For a schedule outside that set, its output distribution can differ by at most one in total variation; multiplying by its probability $p_{\rm ov}$ gives Eq.~\eqref{eq:poisson-finite-pulse-bound}.
For the equally spaced case in Eq.~\eqref{eq:fixed-pulse-envelope}, direct differentiation gives $F'(h)=\frac32g_*^2T-2g_{\rm p}T\delta/h^2$ and $F''(h)=4g_{\rm p}T\delta/h^3>0$. This proves the stated unique interior minimizer of the envelope when it lies in the feasible range. If $K=T/h$ must be integral, the smaller of the two feasible integer values nearest $T/h^*$ minimizes this convex envelope over the grid. The $\min\{1,F(h)\}$ cap may make several settings equally uninformative at value one; Eq.~\eqref{eq:fixed-pulse-optimum} concerns the uncapped informative branch. It also presumes $g_{\rm p}$ does not itself increase with refresh rate or pulse shortening.

The Poisson factorial moment formula bounds the expected number of event pairs closer than $\delta$ by
\begin{equation}
 \lambda_r^2\int_0^Tdt\int_t^{\min(t+\delta,T)}ds
 \le\lambda_r^2T\delta.
 \label{eq:app-poisson-overlap}
\end{equation}
Hence a union bound gives $p_{\rm ov}\le\lambda_r^2T\delta+2\lambda_r\delta$ if overlap with either endpoint pulse is also excluded. A simultaneous $\lambda_r\to\infty$ limit therefore needs control of this probability as well as $g_{\rm p}\lambda_rT\delta$. Refusing or delaying close Poisson events changes the schedule law. A terminal inverse-frame pulse belongs to an open-output or fixed-detector implementation and must be counted separately in $|P|$; the closed trace boundary has no such pulse. These are sufficient bounds uniform over the admissible schedules. Pulse designs with explicit cancellation may have smaller errors, but a diverging residual norm $g_{\rm p}$ cannot be hidden by taking $\delta\to0$ alone.

%% file: appendix_bounds.tex
\section{Fast-refresh limits and the complete-experiment bound}
\label{app:bounds}

\subsection{Fixed-sector limit for a bounded Gaussian correlation}
The first statement concerns the single-axis Gaussian model in Appendices~\ref{app:bath} and \ref{app:controlalg}. Assume $|C(t,u)|\le C_{\max}$ on $[0,T]^2$, hold $g,T$ and a particle sector fixed, and let $\Delta_k$ be the independent-sign holding lengths. Write $h_*=\max_k\Delta_k$ and $c_*=4g^2C_{\max}$. For $\lambda_s,\lambda_r\in\{\pm1\}$, the exponent of Eq.~\eqref{eq:app-discrete-influence} can be grouped as
\begin{equation}
 X_{\boldsymbol y}=\sum_kD_{kk}+\sum_{k>l}D_{kl}y_ky_l,
 \qquad \mathcal F_B[y\lambda_s,y\lambda_r]=e^{-X_{\boldsymbol y}}.
 \label{eq:app-X-signs}
\end{equation}
The diagonal triangle and cross-cell rectangle give $|D_{kk}|\le c_*\Delta_k^2/2$ and $|D_{kl}|\le c_*\Delta_k\Delta_l$. Independent centered signs make different unordered products $y_ky_l$ orthogonal. Consequently
\begin{align}
 \E_{\boldsymbol y}|X_{\boldsymbol y}|^2
 &=\left|\sum_kD_{kk}\right|^2+\sum_{k>l}|D_{kl}|^2\\
 &\le c_*^2\left(\sum_k\Delta_k^2\right)^2
 \le c_*^2T^2h_*^2.
 \label{eq:app-X-second-moment}
\end{align}
For a physical Gaussian influence functional $\operatorname{Re}X_{\boldsymbol y}\ge0$, because its real part is the nonnegative quadratic form of the forward-minus-backward source in the symmetrized bath correlation. The integral identity $1-e^{-X}=X\int_0^1e^{-aX}da$ therefore implies $|1-e^{-X}|\le|X|$. Cauchy--Schwarz and Eq.~\eqref{eq:app-X-second-moment} give the uniform history bound
\begin{equation}
 |\E_{\boldsymbol y}\mathcal F_B[y\lambda_s,y\lambda_r]-1|
 \le c_*T h_*.
 \label{eq:app-sector-bound}
\end{equation}
For any fixed number and times of system insertions, the finite system history sum can be taken through this limit. The resulting coefficient is the uncoupled ideal coefficient of Eq.~\eqref{eq:app-general-amplitude}. This is an all-$g$ statement under the bounded-correlation condition, while the constants depend on the fixed sector.

For Poisson refreshes, condition on the event schedule and apply Eq.~\eqref{eq:app-sector-bound} with its random largest gap. For every $\varepsilon>0$, cover $[0,T]$ by cells of length at most $\varepsilon/2$. If a gap exceeds $\varepsilon$, one such cell has no event. A union bound yields
\begin{equation}
 \Pr(h_*>\varepsilon)\le
 \left(\left\lceil\frac{2T}{\varepsilon}\right\rceil+1\right)
 e^{-\lambda_r\varepsilon/2}.
 \label{eq:app-max-gap}
\end{equation}
Since $0\le h_*\le T$, this proves $\E h_*\to0$ as $\lambda_r\to\infty$ and hence
\begin{equation}
 |\overline{\mathcal F}_B^{\rm Pois}[s,r]-1|
 \le c_*T\E h_*\longrightarrow0.
 \label{eq:app-poisson-sector}
\end{equation}
The uncut Drude spectrum has divergent $C(0)$ and lies outside this bounded-correlation premise. Appendix~\ref{app:fixed-sector-general} replaces it with a finite-time $L^1$ covariance condition, proves fixed-sector convergence and the full-process Fock-norm limit of Theorem~\ref{thm:fock-convergence}, and obtains a sharper fixed-sector rate from its logarithmic short-time envelope. The full-process result controls the particle sum, but extending it to an infinite-field-space CMI still requires an entropy-continuity estimate.

\subsection{Nonautonomous complete-record theorem}
Theorem~\ref{thm:complete-experiment} allows prescribed $H_E(t)$ and $H_{SE}(t)$ and an autonomous Markov control state $Z_t$ with time-dependent law $\pi_t$. The control is initially independent of the joint state, and its transition law is not altered by feedback from the system, instruments, or previous records; it may be statistically correlated with the records it generates. Its laws obey $\pi_sP_{s,t}=\pi_t$. The environment Hamiltonian generates a common time-ordered propagator; its interaction-picture transformation preserves the norm of a bounded coupling. The logical instruments are pointwise covariant, and their complete record-preserving channels are independent of $z$. We assume a nonnegative $g\in L^1[0,T]$ such that $\|H_{{\rm int},z}(t)\|_\infty\le g(t)$ and
\begin{equation}
 \int\pi_t(dz)H_{{\rm int},z}(t)=0
 \label{eq:app-zero-mean}
\end{equation}
for almost every $t$. The system-identity part of $H_{SE}(t)$ is included in $H_E(t)$. The ideal experiment uses the same environment propagation, logical instruments and initial state, with the centered interaction removed; there is no environment reset. These operator-norm assumptions do not follow from an integrable covariance of an unbounded bosonic field.

Let $P_{s,t}(z,dz')$ be the control transition kernel. Define $\|\mu-\nu\|_{\rm TV}=\frac12\int|d\mu-d\nu|$ and the average mixing envelope
\begin{equation}
 b(s,t)=\int\pi_s(dz)\|P_{s,t}(z,\cdot)-\pi_t\|_{\rm TV},
 \qquad 0\le b(s,t)\le1.
 \label{eq:app-mixing}
\end{equation}
For a stationary time-homogeneous process write $b(u)$ and, when finite, $t_{\rm mix}=\int_0^\infty b(u)du$. This is a property of the full control state; for OU it is the pair $(U,\boldsymbol\xi)$, not the frame alone.

Keep the physical system, bath, instrument auxiliary system, and an orthogonal register for all outcome histories in the conditional state $\rho_t(z)$. Complete instruments act as trace-preserving maps on this enlarged state and are trace-norm contractive on Hermitian inputs. Put
\begin{equation}
 \overline\rho_t=\int\pi_t(dz)\rho_t(z),\qquad
 q_t(z)=\rho_t(z)-\overline\rho_t,\qquad
 \|q_t\|_{\pi_t,1}=\int\pi_t(dz)\|q_t(z)\|_1.
 \label{eq:app-q-control}
\end{equation}
Independent preparation gives $q_0=0$. The probability kernel induces a forward operator-valued transfer $\mathcal K_{t,s}$ and a projection $\mathsf P_{t\leftarrow s}$ onto the control-independent component. Their consistency follows from $\pi_sP_{s,t}=\pi_t$. Writing $\mathcal V_z(t)[X]=-i[H_{{\rm int},z}(t),X]$ gives $\|\mathcal V_z(t)[X]\|_1\le2g(t)\|X\|_1$.

Here is the operator-valued kernel estimate needed below. Let $\mathcal X$ be the trace-class operators on the joint system, bath, instrument memory, and classical record. For $f\in L^\infty(\pi_s;\mathcal X)$, define
\begin{equation}
 \|f\|_{\infty,1}=\operatorname*{ess\,sup}_{z\sim\pi_s}\|f(z)\|_1,
 \qquad \|f\|_{\pi_s,1}=\int\pi_s(dz)\|f(z)\|_1.
 \label{eq:app-kernel-spaces}
\end{equation}
Let $\Gamma_{s,t}(dz,dz')=\pi_s(dz)P_{s,t}(z,dz')$ be the joint control law. Its second marginal is $\pi_t$. For $f\in L^\infty(\pi_s;\mathcal X)$, the trace-class-valued measure
\begin{equation}
 \mu_f(A)=\int\pi_s(dz)P_{s,t}(z,A)f(z)
 \label{eq:app-kernel-vector-measure}
\end{equation}
obeys $|\mu_f|(A)\le\|f\|_{\infty,1}\pi_t(A)$, so it vanishes whenever $\pi_t(A)=0$. On the standard Borel control spaces and separable trace class used here, its density is the Bochner conditional expectation $\E_{\Gamma_{s,t}}[f(Z_s)\mid Z_t=z']$, defined for $\pi_t$-almost every $z'$. This construction requires no absolute continuity of the individual transition laws $P_{s,t}(z,\cdot)$. The forward transfer and constant projection from $L^\infty(\pi_s;\mathcal X)$ into $L^1(\pi_t;\mathcal X)$ are therefore
\begin{align}
 (\mathcal K_{t,s}f)(z')&=\E_{\Gamma_{s,t}}[f(Z_s)\mid Z_t=z'],\\
 (\mathsf P_{t\leftarrow s}f)(z')&=\int\pi_s(dz)f(z).
 \label{eq:app-kernel-definition}
\end{align}
The vector measure of $(\mathcal K_{t,s}-\mathsf P_{t\leftarrow s})f$ is $\int\pi_s(dz)[P_{s,t}(z,\cdot)-\pi_t]f(z)$. The triangle inequality for its trace-class variation gives
\begin{align}
 \|(\mathcal K_{t,s}-\mathsf P_{t\leftarrow s})f\|_{\pi_t,1}
 &\le\int\pi_s(dz)\|f(z)\|_1
       \int|P_{s,t}(z,dz')-\pi_t(dz')|\\
 &\le 2b(s,t)\|f\|_{\infty,1}.
 \label{eq:app-operator-kernel-lemma}
\end{align}
This is an $L^\infty(\pi_s;\mathcal X)\to L^1(\pi_t;\mathcal X)$ bound. An unrestricted $L^1\to L^1$ estimate would require the stronger envelope $\operatorname*{ess\,sup}_z\|P_{s,t}(z,\cdot)-\pi_t\|_{\rm TV}$. A common trace-preserving logical propagator, including a complete record-preserving instrument and prescribed $H_E(t)$, acts independently of $z$ and is trace-norm contractive. It intertwines with both transfers, so Eq.~\eqref{eq:app-operator-kernel-lemma} also applies when such a map is inserted before or after the control transfer.

The control is autonomous and retains marginal $\pi_s$ even after the instrument record is included. Consequently each conditional $\rho_s(z)$ is a positive, trace-one joint-record state for $\pi_s$-almost every $z$, and the exact source $f_s(z)=\mathcal V_z(s)[\rho_s(z)]$ satisfies $\|f_s\|_{\infty,1}\le2g(s)$. Duhamel's formula for the true conditional evolution, projected onto its control-dependent component, is
\begin{equation}
 q_t=\int_0^t ds\,\Phi_0(t,s)
       (\mathcal K_{t,s}-\mathsf P_{t\leftarrow s})f_s,
 \label{eq:app-q-duhamel}
\end{equation}
with the common logical propagator $\Phi_0$ understood to act on the trace-class fibre. Applying Eq.~\eqref{eq:app-operator-kernel-lemma} gives
\begin{equation}
 \|q_t\|_{\pi_t,1}\le4\int_0^t ds\,g(s)b(s,t).
 \label{eq:app-q-bound}
\end{equation}
This is not a second-order Dyson truncation; the conditional state inside the source is the exact interacting state and includes all past control and record correlations.

By Eq.~\eqref{eq:app-zero-mean}, $\int\pi_t(dz)\mathcal V_z(t)[\overline\rho_t]=0$. The exact averaged-state difference from the ideal experiment is therefore
\begin{equation}
 \overline\rho_T-\rho_{{\rm id},T}
 =\int_0^Tdt\,\Phi_0(T,t)
 \left[\int\pi_t(dz)\mathcal V_z(t)[q_t(z)]\right],
 \label{eq:app-mean-difference}
\end{equation}
where $\Phi_0$ propagates the uncoupled experiment, including its instruments. Combining the commutator bound with Eq.~\eqref{eq:app-q-bound} gives
\begin{equation}
 \|\overline\rho_T-\rho_{{\rm id},T}\|_1
 \le8\int_0^Tdt\int_0^t ds\,g(t)g(s)b(s,t).
 \label{eq:app-trace-bound}
\end{equation}
The final bath trace and POVM are channels, so their total-variation distance is at most half this trace norm. This proves Eq.~\eqref{eq:general-bound}, uniformly over the stated experiment class and its number of finite instruments.

\subsection{Uniform envelope and protocol specializations}
If $g(t)\le g_*$, Eq.~\eqref{eq:general-bound} immediately gives Corollary~\ref{cor:static-record}. A stationary control has $\pi_t=\pi$ and $b(s,t)=b(t-s)$, so its unweighted time integral equals $\int_0^T(T-u)b(u)du$ and is at most $Tt_{\rm mix}$ when $t_{\rm mix}<\infty$. Scheduled fixed refresh and Cayley use the same corollary with their time-dependent control laws.

For independent four-Pauli redraws, $b(s,t)=3/4$ within a holding interval and zero after a redraw. Thus the triangle on $I_j$ contributes $(3/8)a_j^2$ to the double integral, giving the fixed line of Eq.~\eqref{eq:weighted-protocol-bounds}. For Brownian motion, $b(s,t)=b_1(\kappa(t-s))$ and $C_B=\int_0^\infty b_1(u)du$. Young's inequality on the half-line bounds the weighted triangle by $C_B\|g\|_2^2/\kappa$. The same argument gives $t_{\rm mix}^{\rm OU}\|g\|_2^2$ for the joint frame--velocity envelope. The Cayley calculation uses the nonstationary cell law and the same-cell and neighboring-cell mixing distances of Eq.~\eqref{eq:cayley-two-frame-mixing}; its weighted integral is displayed in Appendix~\ref{app:cayley}. Poisson refresh is treated in Appendix~\ref{app:poisson}.

For a finite control set of $M$ equiprobable frames, let $\pi$ be the uniform marginal law. The distance of a fixed frame from $\pi$ is
\begin{equation}
 \|\delta_G-\pi\|_{\rm TV}
 =\frac12\left[(1-M^{-1})+(M-1)M^{-1}\right]
 =1-M^{-1}.
 \label{eq:app-frame-tv}
\end{equation}
For fixed independent redraws, the transition kernel is $\delta_G$ if $s,t$ lie in the same holding interval, and $\pi$ once at least one refresh boundary has been crossed. Hence $b(s,t)=1-M^{-1}$ within the same interval and zero otherwise, and
\begin{equation}
 \int_0^Tdt\int_0^t ds\,b(s,t)
 =\tfrac12(1-M^{-1})\sum_j\Delta_j^2
 \le\tfrac12(1-M^{-1})Th_*.
 \label{eq:app-fixed-triangle}
\end{equation}
For the four Pauli frames, $M=4$, substitution in Eq.~\eqref{eq:static-record-bound} gives $\epsilon^{\rm fixed}\le(3/2)g_*^2Th_*$, as in Eq.~\eqref{eq:protocol-bounds}. The frames are independently redrawn, even if a new draw equals the previous frame.

The order $g_*^2Th_*$ cannot be improved uniformly under these assumptions. To see this, take one system qubit, no environment, and $H_{\rm int}=g\sigma_z$ before compilation. On each of $K$ equal intervals of length $h=T/K$, a uniform Pauli frame gives an independent sign $y_j=\chi_z(G_j)$. Prepare $|+\rangle$ in the logical frame and measure in the $\sigma_x$ basis at $T$. For a sign history, the total phase is $\theta=g h\sum_jy_j$, and the probability of the $|+\rangle$ outcome is $\cos^2\theta$. Independence gives
\begin{equation}
 \E_{\boldsymbol y}e^{2i\theta}=[\cos(2gh)]^K,
 \qquad p_+=\frac{1+[\cos(2gh)]^K}{2}.
 \label{eq:app-fixed-tight-example}
\end{equation}
The ideal outcome is certain, so its complete-record total variation is $1-p_+$. Expanding for $gh\ll1$ and $g^2Th\ll1$ gives Eq.~\eqref{eq:fixed-bound-example}. The theorem's factor $3/2$ is conservative in this example, while its $g^2Th$ dependence is attained. The prefactor is specific to the instrument and control model used here.

For a Poisson refresh process, condition on whether there was at least one event in a gap $u=t-s$. With probability $e^{-\lambda_ru}$ there was none and the frame is still $G$; otherwise the last redraw is independent and uniform. Thus
\begin{equation}
 P_{s,t}(G,\cdot)=e^{-\lambda_ru}\delta_G
 +(1-e^{-\lambda_ru})\pi,\qquad
 b(u)=(1-M^{-1})e^{-\lambda_ru}.
 \label{eq:app-poisson-kernel}
\end{equation}
The triangular integral is $\int_0^T(T-u)e^{-\lambda_ru}du=T/\lambda_r-(1-e^{-\lambda_rT})/\lambda_r^2$. For four Pauli frames this gives the sharper finite-time inequality
\begin{equation}
 \epsilon^{\rm Pois}\le\min\!\left\{1,3g_*^2
 \left[\frac{T}{\lambda_r}
 -\frac{1-e^{-\lambda_rT}}{\lambda_r^2}\right]\right\}
 \le\min\{1,3g_*^2T/\lambda_r\}.
 \label{eq:app-poisson-exact-bound}
\end{equation}
At zero rate the bracket tends to $T^2/2$, matching a single static random frame rather than complete decoupling. Equations~\eqref{eq:app-fixed-triangle} and \eqref{eq:app-poisson-kernel} compute control mixing, while the bath state is retained in the quantum conditional state throughout.

For isotropic drift-free Brownian control on a compact connected group, generator scaling gives $b^{\rm Br}(u;\kappa)=b^{\rm Br}(\kappa u;1)$. With $C_B=\int_0^\infty b^{\rm Br}(v;1)dv<\infty$, the full-group density yields $b^{\rm Br}(u;\kappa)=\frac12\int|q_u(G)-1|d\mu(G)$. Substitution gives the Brownian line of Eq.~\eqref{eq:protocol-bounds}. If only adjoint actions are physically distinguished, the quotient rotation group may be used, but its mixing constant must be computed for that state space.

For OU control the relevant state is $(U,\boldsymbol\xi)$, and a finite integrated total-variation mixing time is an additional hypothesis. It does not follow from the velocity correlation time $\tau$ alone. If a verified baseline joint process is uniformly accelerated by $U^{(r)}(t)=U^{(1)}(rt)$ and $\boldsymbol\xi^{(r)}(t)=r\boldsymbol\xi^{(1)}(rt)$, then $\tau_r=\tau_1/r$, $\sigma_r^2=r^2\sigma_1^2$, and $t_{{\rm mix},r}=t_{{\rm mix},1}/r$. This differs from both a fixed-strength OU white-noise limit and the fixed-$\tau$ particle scan.

%% file: appendix_fixed_sector.tex
\section{Fixed-sector and Fock-norm convergence for finite environments and Gaussian baths}
\label{app:fixed-sector-general}

The particle number in this appendix counts insertions in the cPT field expansion, not excitations of the physical environment. We fix an insertion order $n$, its times $0<\tau_1<\cdots<\tau_n<T$, and bounded system-side insertion maps. The bounds below may depend on $n$ and the chosen labels, but their constants are independent of the control rate. Pointwise covariance makes the uncoupled logical propagation and all insertion maps independent of the sampled frame. This observation is essential: averaging a control-dependent intended operation would not isolate the coupling error.

\subsection{A finite-environment Dyson theorem}
\label{app:finite-env-sector}
Let $\mathcal H_S$ and $\mathcal H_E$ be finite dimensional and let $\rho_{SE}(0)$ be an arbitrary initial joint state, possibly correlated. The autonomous control is sampled independently of this state. In the logical frame and the interaction picture of the free environment, write
\begin{equation}
 H_{{\rm int},r}(t)=g\sum_{\alpha=1}^{m}
 W_{\alpha,r}(t)\otimes B_\alpha(t),\qquad
 W_{\alpha,r}(t)=U_r(t)^\dagger S_\alpha U_r(t),
 \label{eq:app-finite-interaction}
\end{equation}
where $S_\alpha$ and $B_\alpha(t)$ are bounded. Put
\begin{equation}
 G_E=\sum_\alpha\|S_\alpha\|_\infty
 \sup_{0\le t\le T}\|B_\alpha(t)\|_\infty<\infty.
 \label{eq:app-finite-G}
\end{equation}
If $B_\alpha(t)$ is obtained by conjugating a fixed $B_\alpha$ with the environment propagator, its norm is unchanged. No Gaussian assumption or environment reset is made.

\begin{proposition}[Bounded-environment fixed-sector limit]
\label{prop:app-finite-sector}
Suppose that, for every $q\ge1$, every selection $\alpha_1,\ldots,\alpha_q$, and almost every set of distinct times $t_1,\ldots,t_q\in(0,T)$,
\begin{equation}
 \lim_{r\to\infty}\left\|\E_r
 \bigotimes_{j=1}^{q}W_{\alpha_j,r}(t_j)\right\|=0.
 \label{eq:app-FM-finite}
\end{equation}
Assume also that the product of uncoupled logical propagators and the fixed insertion maps has a trace-class norm bound $K_n$ uniform over every further subdivision of $[0,T]$. Then every trajectory-averaged fixed cPT coefficient obeys
\begin{equation}
 \lim_{r\to\infty}\overline{\mathcal C}^{(n)}_{\boldsymbol\nu,r}
 (\boldsymbol\tau)=\mathcal C^{(n)}_{\boldsymbol\nu,{\rm id}}
 (\boldsymbol\tau).
 \label{eq:app-finite-sector-limit}
\end{equation}
\end{proposition}

\noindent\emph{Proof.} Let $\mathcal V_0(t,s)$ propagate the uncoupled environment and the logical instrument between two times. Define $\mathcal T_{\boldsymbol\nu}[\mathcal A_q(t_q),\ldots,\mathcal A_1(t_1)]$ by placing the $\mathcal A_j$ and all fixed particle insertions in their actual temporal order, with $\mathcal V_0$ between adjacent events, then applying the bounded final coefficient functional to $\rho_{SE}(0)$. This is multilinear and satisfies
\begin{equation}
 |\mathcal T_{\boldsymbol\nu}[\mathcal A_q,\ldots,\mathcal A_1]|
 \le K_n\prod_{j=1}^{q}\|\mathcal A_j\|_{1\to1}.
 \label{eq:app-finite-word-bound}
\end{equation}
For $\mathcal L_r(t)X=-i[H_{{\rm int},r}(t),X]$, the $q$th Duhamel--Dyson contribution is exactly
\begin{equation}
 \mathcal D^{(q)}_r=\int_{0<t_1<\cdots<t_q<T}
 \E_r\mathcal T_{\boldsymbol\nu}
 [\mathcal L_r(t_q),\ldots,\mathcal L_r(t_1)]\,d^qt.
 \label{eq:app-finite-Dyson}
\end{equation}
Expanding each commutator gives a finite sum over channel labels and left/right actions. For fixed distinct times, every term is a multilinear functional of the tensor in Eq.~\eqref{eq:app-FM-finite}; all environment operators and intended insertion maps are control independent. Thus its control average tends to zero. Coincident-time sets have zero integration measure. Moreover,
\begin{equation}
 \|\mathcal L_r(t)X\|_1\le 2|g|G_E\|X\|_1,
 \qquad
 |\mathcal D^{(q)}_r|\le K_n\frac{(2|g|G_ET)^q}{q!}.
 \label{eq:app-finite-majorant}
\end{equation}
The sum of this majorant over $q\ge1$ is $K_n(e^{2|g|G_ET}-1)$. Dominated convergence therefore removes every positive Dyson order while retaining the zero-order ideal coefficient. This proves Eq.~\eqref{eq:app-finite-sector-limit}.\hfill$\square$

The same proof permits prescribed $H_E(t)$ and $H_{SE}(t)$ without a uniform-in-time coupling norm. Let the environment have a well-defined time-ordered propagator $V_E(t,0)$ and expand the system-traceless coupling in a fixed finite system-operator basis as $H_{SE}(t)=\sum_\alpha S_\alpha\otimes B_\alpha(t)$. In the environment interaction picture $B_\alpha^I(t)=V_E(t,0)^\dagger B_\alpha(t)V_E(t,0)$ has the same norm. If $G(t)=\sum_\alpha\|S_\alpha\|_\infty\|B_\alpha(t)\|_\infty$ belongs to $L^1[0,T]$, the $q$th Dyson term is dominated instead by
\begin{equation}
 |\mathcal D^{(q)}_r|\le K_n\frac{\bigl(2\int_0^T G(t)dt\bigr)^q}{q!}.
 \label{eq:app-time-dependent-sector-majorant}
\end{equation}
The finite control-moment condition in Eq.~\eqref{eq:app-FM-finite} still removes each positive order at distinct times. Time-dependent system factors may be expanded in the same fixed traceless basis, with their deterministic coefficients absorbed into $B_\alpha(t)$. The limit therefore holds for this integrable nonautonomous coupling as well.

The condition concerns \emph{all} finite control moments; centering one time or checking only a two-time correlation does not suffice at arbitrary coupling. For independent fixed Pauli redraws, distinct fixed sampling times eventually occupy different cells, and each traceless $S_\alpha$ has zero Pauli twirl. Poisson redraws have the same limit because the chance of no event between two fixed times tends to zero. In a full-cell Cayley path, the frame in one cell depends only on its adjacent target pair; at sufficiently fine nodes, the pairs attached to distinct fixed times are disjoint. The old uniform target frame gives zero Pauli twirl conditional on the relative target. Brownian heat-kernel mixing gives asymptotically independent Haar frames. The time-accelerated three-axis OU conclusion is conditional on mixing of its \emph{joint} frame--velocity state, as specified below. Under that condition these observations verify Eq.~\eqref{eq:app-FM-finite} for the single-axis qubit calculation. They do not cover a fixed-$\tau$ OU variance scan merely by identifying a large variance with time acceleration.

For the qubit environment in Sec.~\ref{sec:numerical}, $S=\sigma_z^{(S)}$, $B=\sigma_x^{(E)}$, $g=0.8$, and $H_E=(1.7/2)\sigma_z^{(E)}$, so $G_E=1$. Equation~\eqref{eq:app-finite-sector-limit} applies separately to all $15^2=225$ two-particle labels at the specified pair of times. Since this label set is finite, its largest coefficient deviation also tends to zero. Equation~\eqref{eq:app-finite-majorant} controls the fixed-sector convergence. It is distinct from the complete-record bound of Theorem~\ref{thm:complete-experiment}, which assumes a quantitative mixing envelope and bounds a different observable. Boundedness, rather than finite environment dimension itself, is what the proof uses.

\subsection{A time-integrable Gaussian covariance theorem}
\label{app:gaussian-L1-sector}
Assume $\rho_{SE}(0)=\rho_S\otimes\rho_E$ is independent of the
autonomous control, with $\rho_E$ centered Gaussian (or quasifree).
The coupling operators $B_\alpha(t)$ are linear bath fields, so their
ordered moments obey Wick's rule. Write their matrix covariance and
finite-time norm as
\begin{equation}
 C_{\alpha\beta}(t,s)=\Tr_E[B_\alpha(t)B_\beta(s)\rho_E],
 \quad
 A_T=\int_0^T\!\int_0^T\|C(t,s)\|\,dt\,ds.
 \label{eq:app-general-covariance}
\end{equation}
The bath need not be stationary. A prescribed time-dependent quadratic $H_E(t)$ is admissible when its evolution preserves the initial state's quasifree structure; the operators $B_\alpha(t)$ then include that time-ordered interaction-picture evolution. Bounded deterministic coupling coefficients may be absorbed into these operators and their covariance. We use finitely many Hermitian system factors $S_\alpha$ and Hermitian bath fields $B_\alpha(t)$, and put $s_\Sigma=\sum_\alpha\|S_\alpha\|_\infty$. In the multichannel estimates, $\|C(t,s)\|$ is the matrix spectral norm; hence $|C_{\alpha\beta}(t,s)|\le\|C(t,s)\|$ and $\|C(s,t)\|=\|C(t,s)\|$. At each ultraviolet cutoff $\Lambda$, assume a Fock representation of the purified quasifree state with vacuum $\Omega_\Lambda$ and fields
\begin{equation}
 B_{\alpha,\Lambda}(t)=a(f_{\alpha,\Lambda}(t))+
 a^\dagger(f_{\alpha,\Lambda}(t)),\qquad
 L_\Lambda=2|g|\int_0^T\sum_\alpha
 \|S_\alpha\|_\infty\|f_{\alpha,\Lambda}(t)\|\,dt<\infty.
 \label{eq:app-gaussian-cutoff-realization}
\end{equation}
The one-particle space may be doubled for a thermal or mixed quasifree state. The form factors are strongly measurable; $L_\Lambda$ need not be uniform as the cutoff is removed. This assumption holds, for example, when each cutoff has finite uniformly bounded equal-time covariance on the chosen finite interval. A frequency cutoff does not make the field operator bounded; the following construction supplies its unitary interaction-picture propagation. Separately, assume $C_\Lambda\to C$ in $L^1([0,T]^2)$ with $\sup_\Lambda A_{T,\Lambda}<\infty$. For uncut fields, the limiting fixed coefficient is defined by the physical cutoff processes; no uncut Hamiltonian realization is asserted.

\begin{proposition}[Time-integrable Gaussian fixed-sector limit]
\label{prop:app-gaussian-sector}
Under the preceding factorized Gaussian-state, linear-field, and cutoff
assumptions, suppose $A_T<\infty$ and the rotated system factors satisfy
the multichannel finite-moment condition of Eq.~\eqref{eq:app-FM-finite}.
Then every trajectory-averaged fixed cPT coefficient converges to its
uncoupled value as the control is accelerated.
The convergence is uniform under removal of the stated ultraviolet cutoff, so
the cutoff and fast-control limits commute for each fixed sector.
\end{proposition}

\noindent\emph{Proof.} We first construct the regularized physical propagator. Let $N_\Lambda$ be the bath number operator and $Q_J=\mathbf1_{N_\Lambda\le J}$. Replace each field by $Q_JB_{\alpha,\Lambda}(t)Q_J$, retaining the bounded Hermitian system factors. The resulting interaction Hamiltonian is bounded, norm integrable, and self-adjoint, so it generates a unitary propagator $V_{\Lambda,J}(t,s)$. For a vector $\psi_k$ supported on bath particle numbers at most $k$, a linear field changes particle number by at most one and obeys $\|B_{\alpha,\Lambda}(t)\psi_k\|\le2\|f_{\alpha,\Lambda}(t)\|\sqrt{k+1}\|\psi_k\|$. Thus its $q$th ket Dyson term, with any finite bounded system insertions of total operator-norm product $M_n$, satisfies
\begin{equation}
 \|D_{q,\Lambda,J}\psi_k\|
 \le M_n\frac{L_\Lambda^q}{q!}
       \sqrt{\frac{(k+q)!}{k!}}\,\|\psi_k\|.
 \label{eq:app-gaussian-vector-majorant}
\end{equation}
This estimate is uniform in $J$ and in the control path, since conjugation does not change a system factor's norm. The ratio of successive majorants is $L_\Lambda\sqrt{k+q+1}/(q+1)\to0$. Every fixed-order term agrees with the unprojected field word once $J\ge k+q$. Here $J$ is an auxiliary particle-number truncation used to construct the cutoff propagator; it is distinct from the physical ultraviolet cutoff $\Lambda$. Consequently $V_{\Lambda,J}(t,s)$ has a strong limit on the dense finite-particle subspace, uniformly over $0\le s\le t\le T$. The same argument for the adjoints gives a strong inverse limit. Since all approximants are unitary, the limits extend to a unitary cocycle $V_\Lambda(t,s)$ on the full Hilbert space. This defines the physical cutoff propagation for the measurable interaction of Eq.~\eqref{eq:app-gaussian-cutoff-realization}, including its bounded logical interaction-picture factors.

On the purified Gaussian input, each ket Dyson series, including finitely many system insertions, converges in norm. The corresponding bra--ket matrix element is therefore the absolutely convergent double series of these terms. Applying Wick's rule to the unprojected field words and regrouping by total coupling order identifies the physical insertion coefficients with the Dyson--Wick expansion of Eq.~\eqref{eq:app-wick}. The number-projected bath is used only to construct the propagator; Wick's rule is applied after its strong limit, not to the projected, non-Gaussian field operators. This closes the propagation-to-coefficient step at each cutoff without assuming that frequency regularization bounds the field operator.

We may now take the bath trace and average the control. At Dyson order $2m$, Wick's rule pairs the bath operators. Each pair is an entry of $C_\Lambda$ or its adjoint; odd orders vanish. With $a=2g^2s_\Sigma^2$ and $A=A_{T,\Lambda}$, the commutator branches, pairings, and ordered time integral have the absolute majorant
\begin{equation}
 K_n\frac{(aA)^m}{m!}.
 \label{eq:app-Wick-majorant}
\end{equation}
Indeed, extending the symmetric pairing sum from the ordered simplex to the whole time cube gives $A^m/(2^m m!)$; the $2^{2m}$ left/right choices supply the displayed factor in $a$. The sum is bounded by $K_ne^{aA}$ uniformly in control strength and cutoff. At each fixed set of distinct coupling times, the averaged system factor is a multilinear functional of the tensor in Eq.~\eqref{eq:app-FM-finite}, and therefore vanishes in the fast limit. Dominated convergence applies first to each Wick order and then to their uniformly convergent sum.

For two cutoffs $\Lambda,\Lambda'$, telescope each product of pair covariances once and let $A_*$ bound their two $A$ values. The same pairing count yields
\begin{equation}
 \sup_r|\overline{\mathcal C}^{(n)}_{\Lambda,r}
 -\overline{\mathcal C}^{(n)}_{\Lambda',r}|
 \le K_n a e^{aA_*}
 \|C_\Lambda-C_{\Lambda'}\|_{L^1([0,T]^2)}.
 \label{eq:app-cutoff-stability}
\end{equation}
The bound is uniform in the control law. It makes the cutoff family Cauchy for every fixed coefficient and permits either order of limits.\hfill$\square$

The covariance condition includes matrix-valued cross correlations and nonstationary Gaussian baths. It is weaker than a pointwise bound on $C(t,s)$. The explicitly stated matrix spectral norm and Hermitian-channel convention are used in the entrywise Wick estimate above; a non-Hermitian channel decomposition must first be replaced by a Hermitian one or use a symmetric covariance envelope with the corresponding constants recomputed. The preceding propositions concern fixed field-insertion order. For the fixed operation basis of the free, trace-closed process, the insertion bound can be chosen uniformly across orders as shown next, allowing the particle sum to be controlled. Entropy continuity requires additional assumptions.

\subsection{Fock-norm convergence of the trajectory-averaged process}
\label{app:fock-convergence}
We prove Theorem~\ref{thm:fock-convergence}. The norm is the ordered-simplex
Fock norm of Ref.~\cite{fuksa2026dilation}; its Appendix~B,
Eqs.~(B9)--(B10), establishes factorial sector bounds for physical cPTs.
Here we derive a coefficient majorant directly in the operation basis
used in this paper, including uniformity in the Gaussian regularization.
The PCR density theorem in that reference is not itself a fast-control
convergence result.

\paragraph{Fixed basis and the coefficient majorant.}
Let $q=d_S^4-1$, $\Delta_n=\{0<t_1<\cdots<t_n<T\}$, and
let $\mathcal P_\nu$ be the fixed insertion maps in
Eq.~\eqref{eq:app-general-amplitude}. These are left/right maps
$X\mapsto(P_\mu\otimes\id_E)X(P_{\mu'}^\dagger\otimes\id_E)$.
Set
\begin{equation}
 R=\max\!\left\{1,\max_{\nu=(\mu,\mu')}
          \|P_\mu\|_\infty\|P_{\mu'}\|_\infty\right\}.
 \label{eq:app-fock-insertion-bound}
\end{equation}
Then $\|\mathcal P_\nu[X]\|_1\le R\|X\|_1$, also after extension
by any environment identity. The constant is independent of the control
rate, cutoff, insertion times, and environment dimension.
For the unscaled Pauli basis used in the qubit calculations, $R=1$.
The basis is fixed throughout; a rate-dependent rescaling would change
the norm under consideration.

In case (i), joint unitary conjugation is an isometry on trace-class
operators, including non-Hermitian operators produced by off-diagonal
insertions. Starting from $\|\rho_{SE}(0)\|_1=1$ and using
$|\Tr X|\le\|X\|_1$ at the trace boundary gives, path by path,
\begin{equation}
 |\mathcal C^{(n),\mathrm{tr}}_{\boldsymbol\nu,r,\omega}
   (\boldsymbol t)|\le R^n.
 \label{eq:app-fock-finite-majorant}
\end{equation}
The same bound holds for the trajectory average and the ideal process.
Initial system--environment correlations do not change this estimate.

In case (ii), use the regularized Dyson--Wick construction of
Proposition~\ref{prop:app-gaussian-sector}. Its uncoupled propagation is
unitary, and its $n$ fixed insertions have total trace-norm bound $R^n$
under every subdivision by coupling events. Thus the constant $K_n$
in Eqs.~\eqref{eq:app-finite-word-bound}
and \eqref{eq:app-Wick-majorant} may be taken to be $R^n$ for these
free-process coefficients. With $a=2g^2s_\Sigma^2$ and
$A_*=\sup_\Lambda A_{T,\Lambda}<\infty$, summing the absolute Wick
majorant gives
\begin{equation}
 |\overline{\mathcal C}^{(n),\mathrm{tr}}_{
       \boldsymbol\nu,\Lambda,r}(\boldsymbol t)|
 \le R^n\sum_{m\ge0}\frac{(aA_*)^m}{m!}
 =R^ne^{aA_*}.
 \label{eq:app-fock-gaussian-majorant}
\end{equation}
No cutoff-uniform pointwise bound on $C_\Lambda(t,s)$ or operator-norm bound on the field is used in this Wick estimate; the finite-cutoff realization was established separately above.
Consequently, in both cases,
\begin{equation}
 |\overline{\mathcal C}^{(n),\mathrm{tr}}_{
       \boldsymbol\nu,r}(\boldsymbol t)|,
 \quad |\mathcal C^{(n),\mathrm{tr}}_{
       \boldsymbol\nu,{\rm id}}(\boldsymbol t)|
 \le MR^n,
 \qquad
 M=\begin{cases}1,&\text{case (i)},\\
 e^{aA_*},&\text{case (ii)}.
 \end{cases}
 \label{eq:app-fock-uniform-majorant}
\end{equation}
In case (ii) this estimate first applies at each cutoff and then passes
to the limiting coefficients. The normalized vacuum coefficient is one.
These bounds also ensure that the process vectors and their measurable
trajectory averages lie in $\Gamma$.

\paragraph{Finite sectors and the uniform tail.}
Write $|D_r\rangle=|\overline\Upsilon_{T,r}\rangle
-|\Upsilon_{T,{\rm id}}\rangle$ and let $\Pi_{\le N}$ project onto
particle orders at most $N$. In the normalization of
Eq.~\eqref{eq:cpt-expansion},
\begin{equation}
 \|D_r\|_\Gamma^2
 =\sum_{n=0}^{\infty}\sum_{\boldsymbol\nu\in[q]^n}
  \int_{\Delta_n}
  \left|\overline{\mathcal C}^{(n),\mathrm{tr}}_{
       \boldsymbol\nu,r}(\boldsymbol t)
       -\mathcal C^{(n),\mathrm{tr}}_{
       \boldsymbol\nu,{\rm id}}(\boldsymbol t)\right|^2
  d\boldsymbol t.
 \label{eq:app-fock-norm}
\end{equation}
For each fixed $n$ and label string, the integrand tends to zero for
almost every insertion-time tuple by
Proposition~\ref{prop:app-finite-sector} or
\ref{prop:app-gaussian-sector}. It is bounded by $4M^2R^{2n}$.
Since $|\Delta_n|=T^n/n!$ and there are $q^n$ label strings,
dominated convergence gives
$\|\Pi_{\le N}D_r\|_\Gamma\to0$ for every fixed $N$.
The complementary part satisfies the rate-independent estimate
\begin{equation}
 \sup_r\|\Pi_{>N}D_r\|_\Gamma^2
 \le4M^2\sum_{n>N}\frac{(qR^2T)^n}{n!}
 \xrightarrow[N\to\infty]{}0.
 \label{eq:app-fock-tail}
\end{equation}
The bound is also cutoff independent in case (ii).
Given any error tolerance, choose $N$ to control this tail, and then
choose $r$ to control the finitely many lower sectors. Orthogonality of
the sector projections proves Eq.~\eqref{eq:fock-convergence}.
The estimate is for fixed $T$ and establishes no uniform long-time limit.

\paragraph{Gaussian cutoff removal in the full norm.}
The stability estimate \eqref{eq:app-cutoff-stability}, now with
$K_n=R^n$, gives
\begin{equation}
 \sup_r\left|
  \overline{\mathcal C}^{(n),\mathrm{tr}}_{
    \boldsymbol\nu,\Lambda,r}
 -\overline{\mathcal C}^{(n),\mathrm{tr}}_{
    \boldsymbol\nu,\Lambda',r}\right|
 \le R^n a e^{aA_*}
       \|C_\Lambda-C_{\Lambda'}\|_{L^1([0,T]^2)}.
 \label{eq:app-fock-cutoff-coefficient}
\end{equation}
Squaring, integrating over insertion times, and summing all particle
orders yields
\begin{equation}
 \sup_r\bigl\||\overline\Upsilon_{T,\Lambda,r}\rangle
       -|\overline\Upsilon_{T,\Lambda',r}\rangle\bigr\|_\Gamma
 \le a e^{aA_*}e^{qR^2T/2}
       \|C_\Lambda-C_{\Lambda'}\|_{L^1([0,T]^2)}.
 \label{eq:app-fock-cutoff-stability}
\end{equation}
The vacuum difference vanishes; retaining the zeroth term in the
exponential only weakens the bound. Completeness of $\Gamma$ and the
assumed $L^1$ convergence of $C_\Lambda$ therefore define an uncut
averaged process $|\overline\Upsilon_{T,r}\rangle$, uniformly in $r$.
For a fixed cutoff, the preceding argument gives convergence to
$|\Upsilon_{T,{\rm id}}\rangle$. Given a tolerance, first choose a
cutoff for which Eq.~\eqref{eq:app-fock-cutoff-stability} controls its
distance to the uncut process uniformly in $r$, and then take $r$ large
at this fixed cutoff. The triangle inequality proves convergence of
the uncut process and shows that both iterated limits equal the ideal
vector. This argument also controls simultaneous limits with
$\Lambda\to\infty$ and $r\to\infty$.
It remains to verify that this norm limit is a physical cPT rather than
only a Fock vector. Use the finite-slot positivity and causality definition
of Ref.~\cite{fuksa2026dilation}. For a fixed number of slots and a
nonnegative normalized test function $f\in L^2(\Delta_k)$, every entry of
the smeared finite-slot Choi matrix is a bounded linear functional of
finitely many process sectors. A nonidentity label uses the corresponding
$L^2$ pairing. An identity slot composes adjacent propagators and removes
that field insertion; integration over the unused slot times produces
an $L^2$ kernel on the lower simplex by Cauchy--Schwarz on the finite
interval. Thus these smeared matrices converge entrywise with the Fock
vector. Their finite-dimensional positive cones, causal affine constraints,
and normalized vacuum condition are closed. Each cutoff process is physical
by the unitary construction above, and therefore the limit satisfies the
same smeared positivity and causality conditions. No point evaluation of
an arbitrary $L^2$ kernel is required. The logarithmic short-time covariance
of the uncut Drude model is integrable on a finite time square, so it is
included whenever the stated physical cutoff realization and $L^1$
covariance convergence hold.\hfill$\square$

Equation~\eqref{eq:fock-instrument-bound} follows from the induced
instrument norm of Ref.~\cite{fuksa2026dilation}, Eq.~(8).
The process-vector convergence does not provide an error bound uniform
over instruments of arbitrarily large norm. In particular, it does not
replace Theorem~\ref{thm:complete-experiment}, which has a quantitative
mixing premise and controls complete records directly. Nor does Fock-norm
convergence alone imply continuity of operator CMI in infinite field
space. The fixed-sector rates below remain statements at prescribed
particle order unless their order-dependent constants are also summed.

\subsection{Protocol-specific fixed-sector rates}

For the single-axis, stationary fixed-refresh case, one can say more without estimating the entire Dyson series. Let $C(v)=C(t,t-v)$ and $y_k=\pm1$ be independent Pauli signs on cells of length $h$. The exact influence exponent of Eq.~\eqref{eq:app-discrete-influence} is $X_y=A_h+Z_y$, where $A_h$ sums same-cell triangles and $Z_y$ sums $y_ky_lD_{kl}$ over distinct cells. Since $\operatorname{Re}X_y\ge0$, $|1-e^{-X_y}|\le|X_y|$. If $C\in L^2(0,T)$, Cauchy--Schwarz on each cross-cell rectangle and orthogonality of distinct sign products give
\begin{equation}
 |\E_y e^{-X_y}-1|\le\delta(h):=
 4g^2\!\left[T\int_0^h|C(v)|dv
 +h\left(\int_0^T(T-v)|C(v)|^2dv\right)^{1/2}\right].
 \label{eq:app-L2-refresh}
\end{equation}
Thus $\delta(h)=O(\sqrt h)$ for a fixed finite $L^2$ norm, with a constant independent of an ultraviolet cutoff whose $L^2$ norms are uniformly bounded. Mere $L^1$ integrability still gives qualitative convergence: the same-cell integral tends to zero by absolute continuity, while $\E|Z_y|\le4g^2\sqrt{Th}\,\|C\|_{L^1(0,T)}$. A bounded short-time kernel recovers $O(h)$. Each fixed coefficient is a finite sum of influence factors multiplied by bounded insertion weights; for normalized qubit Pauli insertions a coarse bound is $4^{n+1}\delta(h)$. The bath is traced only after the full temporal history has been assembled.

\subsection{Spectral tests and the uncut Drude example}
\label{app:drude-L1-sector}
For a stationary thermal bath write $C(v)=\int_{\mathbb R}f(\omega)e^{-i\omega v}d\omega$, with the thermal absorption and emission factors included in $f$. If $f\in L^1(\mathbb R)$, then $C$ is bounded. The useful new case is $f\in L^2(\mathbb R)$ but $f\notin L^1$: Plancherel gives $C\in L^2(0,T)\subset L^1(0,T)$ although $C(0)$ may diverge. These spectral criteria provide sufficient verification of the stated convergence class. An ultraviolet law $f(\omega)=O(|\omega|^{-\alpha})$ has an integrable high-frequency tail for $\alpha>1$ and a square-integrable one for $\alpha>1/2$, provided its infrared and thermal factors also pass the corresponding test. Exponentially or Gaussian-cutoff Ohmic families, and regular Lorentzian sums, are typically in the bounded-correlation class after their infrared behavior is checked. The theorem itself uses the time covariance and also covers nonstationary Gaussian kernels with no single spectral density.

For the uncut Drude spectrum of Eq.~\eqref{eq:app-drude}, let $a_D=2\Lambda\gamma_B/\pi$. Its positive-frequency thermal density obeys $f(\omega)=a_D/\omega+O(\omega^{-3})+O(e^{-\beta\omega}/\omega)$ at high frequency; the negative-frequency tail is exponentially suppressed, and both sides are bounded near $\omega=0$. Hence $f\in L^2(\mathbb R)$ but $f\notin L^1(\mathbb R)$. With $f_\Omega=e^{-|\omega|/\Omega}f$, Plancherel gives $\sup_\Omega\|C_\Omega\|_2\le\sqrt{2\pi}\|f\|_2$ and $C_\Omega\to C$ in $L^2$, hence in $L^1$ on $[0,T]$. At each fixed $\Omega$, $f_\Omega\in L^1$ and the purified thermal-field form factor has finite norm squared $C_\Omega(0)$; stationarity then gives $L_\Omega<\infty$ in Eq.~\eqref{eq:app-gaussian-cutoff-realization}. Together these facts verify the physical-cutoff and covariance assumptions of Proposition~\ref{prop:app-gaussian-sector}, while the uncut $C(0)$ remains logarithmically divergent.

The $1/\omega$ tail also controls the short-time rate. For any fixed reference time $\tau_0>0$, splitting $f=a_D\mathbf1_{\omega>1/\tau_0}/\omega+f_{\rm rem}$ with $f_{\rm rem}\in L^1$ and integrating the oscillatory tail by parts gives a cutoff-independent envelope
\begin{equation}
 |C_\Omega(v)|\le A_D+B_D\log_+(\tau_0/v),\qquad v>0,
 \label{eq:app-Drude-envelope}
\end{equation}
where $A_D,B_D<\infty$ depend on bath parameters, not $\Omega$. Substitution in Eq.~\eqref{eq:app-L2-refresh} yields, for $h\le\tau_0$,
\begin{equation}
 |\overline{\mathcal C}^{(n)}_{\boldsymbol\nu,h}
 -\mathcal C^{(n)}_{\boldsymbol\nu,{\rm id}}|
 =O_n\!\left(g^2Th\log(\tau_0/h)\right).
 \label{eq:app-Drude-refresh-rate}
\end{equation}
The logarithm is a consequence of this ultraviolet tail, not a universal rate for Gaussian baths. For this fixed sector, removing the Drude cutoff and taking $h\to0$ commute by Eq.~\eqref{eq:app-cutoff-stability}.

\subsection{Brownian and OU rates for the uncut bath}
\label{app:continuous-frame-Drude}
For the continuous single-axis protocols let $W_r(t)=U_r(t)^\dagger S U_r(t)$, $s=\|S\|_\infty$, and $Z_r(t)=Z_1(rt)$ be the full stationary Markov control state. For Brownian control $Z=U$; for OU control $Z=(U,\boldsymbol\xi/r)$ under the overall time acceleration defined in Appendix~\ref{app:controlalg}. Define the average total-variation envelope of the baseline chain as $b(u)=\int\pi(dz)\|P_u(z,\cdot)-\pi\|_{\rm TV}$. Consecutive conditioning gives, at ordered times $t_1<\cdots<t_q$,
\begin{equation}
 \left\|\E\bigotimes_{j=1}^q W_r(t_j)\right\|
 \le 2s^q\sum_{j=1}^{q-1}b(r(t_{j+1}-t_j)).
 \label{eq:app-finite-mixing-moments}
\end{equation}
The right side tends to zero at distinct fixed times if $b(u)\to0$, establishing the moment condition. It does not replace the full trajectory average by a mean generator.

Put $A=\int_{[0,T]^2}|C(t-u)|dtdu$, $w_* =\sup_t\int_0^T|C(t-u)|du$, $a=2g^2s^2$, and
\begin{equation}
 J_r=2\int_0^T(T-v)|C(v)|b(rv)dv,\qquad
 D_r=2\int_0^T(T-v)b(rv)dv.
 \label{eq:app-control-bath-weights}
\end{equation}
To bound a Wick term, enlarge the sum of adjacent time gaps in Eq.~\eqref{eq:app-finite-mixing-moments} to all time pairs. If the selected pair is itself a Wick contraction, its two integrations give $J_r$; otherwise its two Wick partners give at most $w_*^2D_r$. Counting these possibilities and summing the majorant of Eq.~\eqref{eq:app-Wick-majorant} yields
\begin{equation}
 |\overline{\mathcal C}^{(n)}_{\boldsymbol\nu,r}
 -\mathcal C^{(n)}_{\boldsymbol\nu,{\rm id}}|
 \le 2K_ne^{aA}\left[aJ_r+2a^2w_*^2D_r\right].
 \label{eq:app-control-Gaussian-rate}
\end{equation}
If $B_1=\int_0^\infty b(u)du<\infty$ and $C\in L^2(0,T)$, Cauchy--Schwarz gives $J_r\le2T\|C\|_2\sqrt{B_1/r}$ and $D_r\le2TB_1/r$, so the fixed-sector bound is $O_n(r^{-1/2})$. For an exponential envelope $b(u)\le M e^{-\gamma_cu}$ and Eq.~\eqref{eq:app-Drude-envelope}, direct integration improves this to
\begin{equation}
 |\overline{\mathcal C}^{(n)}_{\boldsymbol\nu,r}
 -\mathcal C^{(n)}_{\boldsymbol\nu,{\rm id}}|
 =O_n\!\left(\frac{\log(r\gamma_c\tau_0)}{r\gamma_c}\right)
 \quad(r\gamma_c\tau_0\ge1).
 \label{eq:app-Drude-continuous-rate}
\end{equation}
Here $r$ is a dimensionless acceleration factor, and $\gamma_c$ is a control mixing rate, distinct from the Drude parameter $\gamma_B$.

The Brownian envelope is exponentially integrable by the spectral gap of the compact-group heat semigroup \cite{nechita2013,onorati2017}. For OU, the relevant state is the frame together with its velocity. Set $\boldsymbol\eta=\boldsymbol\xi/\sigma$ for a baseline process with $\sigma,\tau>0$ and all three noise directions present. Its equations and invariant law are
\begin{equation}
 \dot U_t=-i\sigma\sum_{c=1}^3\eta_{c,t}F_cU_t,\qquad
 d\eta_{c,t}=-\tau^{-1}\eta_{c,t}dt+\sqrt{2/\tau}\,dW_{c,t},
 \qquad \pi=\mu_{\rm Haar}\otimes N(0,I_3).
 \label{eq:app-OU-joint-state}
\end{equation}
The OU generator preserves the Gaussian factor and, for each fixed velocity, the group flow preserves Haar measure. For $V(U,\boldsymbol\eta)=1+|\boldsymbol\eta|^2$, direct calculation gives
\begin{equation}
 LV=-2\tau^{-1}V+8\tau^{-1},\qquad
 P_tV=e^{-2t/\tau}V+4(1-e^{-2t/\tau}).
 \label{eq:app-OU-Lyapunov}
\end{equation}
The noise fields are $Y_c=\sqrt{2/\tau}\,\partial_{\eta_c}$. If $X_c^G$ are the three group directions, their brackets with the drift contain
\begin{equation}
 [\partial_{\eta_c},X_0]=\sigma X_c^G-\tau^{-1}\partial_{\eta_c}.
 \label{eq:app-OU-brackets}
\end{equation}
Thus the noise fields and their first brackets span the six-dimensional tangent space at every point, giving the H\"ormander condition for smooth transition densities \cite{hormander1967}. This establishes regularity, but does not by itself give a common minorization measure.

To obtain an exponential mixing estimate from this drift and bracket calculation, one additionally needs a small-set minorization on a sufficiently large velocity sublevel set. We state it as an explicit condition: for some sufficiently large $R$, a time $t_R>0$, a neighborhood $O$ in $SU(2)\times\mathbb R^3$, and $\varepsilon_R>0$,
\begin{equation}
 P_{t_R}(z,\cdot)\ge\varepsilon_R\nu_O(\cdot)
 \quad\text{for }|\boldsymbol\eta(z)|\le R,
 \label{eq:app-OU-small-set}
\end{equation}
where $\nu_O$ is normalized Haar--Lebesgue volume on $O$. A possible verification route uses smooth controlled velocity paths, endpoint-map nondegeneracy, a support theorem \cite{stroockvaradhan1972}, and a positive-density theorem. The bracket and accessibility checks identify the route to the small-set condition. Under Eq.~\eqref{eq:app-OU-small-set}, the drift in Eq.~\eqref{eq:app-OU-Lyapunov} and a geometric Harris theorem \cite{meyntweedie2009} give constants $M_{\rm OU},\gamma_c>0$ such that
\begin{equation}
 \|P_u(z,\cdot)-\pi\|_{\rm TV}
 \le M_{\rm OU}V(z)e^{-\gamma_cu},\qquad
 b(u)\le4M_{\rm OU}e^{-\gamma_cu}.
 \label{eq:app-OU-mixing}
\end{equation}
Thus the OU fixed-sector rates follow whenever Eq.~\eqref{eq:app-OU-small-set}, or directly an integrable joint-state mixing envelope, is available. The displayed drift and bracket checks provide the structural route to this condition; the mixing rate can then be obtained from the joint-process analysis or numerical calibration.

The OU limit here accelerates the \emph{whole} baseline process: $U_r(t)=U_1(rt)$ and $\boldsymbol\xi_r(t)=r\boldsymbol\xi_1(rt)$, so $\tau_r=\tau_1/r$ and $\sigma_r^2=r^2\sigma_1^2$. The Drude numerical OU sweep in Sec.~\ref{sec:numerical} keeps $\tau$ fixed while changing its variance. Equation~\eqref{eq:app-Drude-continuous-rate} makes no claim about that separate parameter family. Theorem~\ref{thm:fock-convergence} controls all field sectors through the additional uniform tail estimate. The fixed-sector rate bounds alone do not give a full-process rate, a complete-record total-variation bound uniform over arbitrary instruments for unbounded baths, or continuum operator CMI convergence.

%% file: appendix_cayley.tex
\section{Bounded Cayley paths for fixed Pauli refresh}
\label{app:cayley}

This appendix gives a bounded-control realization of fixed refresh. The random variables remain the independently sampled Pauli targets at the prescribed nodes; only the instantaneous jump between targets is replaced by a path generated by bounded Hamiltonians. This distinction is relevant for finite pulse amplitude or bandwidth and connects the construction to bounded-control dynamical decoupling and geometric control \cite{violaKnill2003,chen2006}. The random-walk picture follows the group-valued control framework of random decoupling \cite{viola2005}; CPRC additionally imposes pointwise covariance on every finite-duration instrument.

\subsection{Target frames and the Cayley graph}
Let $\mathcal P=\{I,X,Y,Z\}$ denote Pauli frames modulo global phase. Divide $[0,T]$ into $K$ equal intervals, with nodes $b_j=jh$ and $h=T/K$. Draw
\begin{equation}
 G_0,G_1,\ldots,G_{K-1}\quad\text{independently and uniformly from }\mathcal P,
 \label{eq:cayley-targets}
\end{equation}
where the last target frame is held through the final interval. Consecutive draws are allowed to coincide. The target frame on the holding interval is $G_j$; hence the target-frame process is exactly the fixed-refresh process used in the interval classicalization proof.

The physical transition is specified on the Cayley graph generated by $X$ and $Z$. For a relative transformation $Q=G_{j+1}G_j$ (global phases ignored), let
\begin{equation}
 Q=L_{j,m}\cdots L_{j,1},\qquad L_{j,r}\in\{X,Z\},
 \label{eq:cayley-word}
\end{equation}
be a shortest word. Thus $Q=I$ has length zero, $Q=X$ or $Z$ has length one, and $Q=Y$ is represented by two edges. An edge labelled $L$ is generated by a bounded Hamiltonian proportional to the corresponding Pauli operator. Choosing the pulse duration $\delta_{j,r}$ so that
\begin{equation}
 \exp[-i\,\vartheta_{j,r}L_{j,r}/2]
 \quad\text{implements the required edge,}\qquad
 |\vartheta_{j,r}|\le\pi,
 \label{eq:cayley-edge}
\end{equation}
gives a continuous path from $G_j$ to $G_{j+1}$. The required edge amplitude depends on the available duration and is not held fixed as $K$ changes. Increasing $K$ therefore refines the node schedule while increasing the required control resource; the sequence describes a control-speed refinement.

For a path $V_j(t)$ inside interval $j$, define the implemented frame by
\begin{equation}
 U(t)=V_j(t)G_j,\qquad b_j\le t<b_{j+1},
 \label{eq:cayley-frame}
\end{equation}
where $V_j(b_j)=I$ and $V_j(b_{j+1})=G_{j+1}G_j$ whenever $j<K-1$. The instrument Hamiltonian is transformed pointwise according to Eq.~\eqref{eq:instrument-hamiltonian-covariance}. Therefore the path changes the physical control trajectory but does not change the uncoupled logical instrument on any individual trajectory, by Proposition~\ref{prop:transparency}.

\subsection{Full-interval node-matched schedule}
The Cayley word is traversed uniformly over $[b_j,b_{j+1}]$; the last target frame is held through the final interval. If $n_j$ is the shortest-word length, each edge lasts $h/n_j$. The associated peak angular velocity is proportional to $n_j/h$. Thus the path reaches the independently sampled Pauli target exactly at each node. Completing a nontrivial edge in time $h=T/K$ requires increasing control speed, so increasing $K$ implements a control-speed refinement.

\subsection{Propagation and error accounting}
On a Cayley edge, distinguish the laboratory Hamiltonian from its logical-frame generator. Tensor products with identities on untouched spaces are implicit:
\begin{equation}
 \begin{aligned}
 H_{\rm lab}(t)&=H_{\rm ctrl}(t)+H_E(t)+H_{SE}(t)\\
 &\quad+U(t)H_{\rm inst}^{\rm log}(t)U^\dagger(t),\\
 H_{\rm log}(t)&=H_E(t)+U^\dagger(t)H_{SE}(t)U(t)\\
 &\quad+H_{\rm inst}^{\rm log}(t).
 \end{aligned}
 \label{eq:cayley-frame-generators}
\end{equation}
The control generator cancels in the logical frame. Because the Cayley frame, and possibly the environment and coupling, vary during an edge, its exact joint propagator is
\begin{equation}
\mathcal U_{j,r}^{\rm log}=\mathcal T\exp\!\left[-i\int_{t_{j,r}}^{t_{j,r+1}}
 H_{\rm log}(t)\,dt\right],
 \label{eq:cayley-propagator}
\end{equation}
and the complete interval propagator is the ordered product of the edge and instrument pieces. In the finite-dimensional Cayley simulations, each edge is subdivided at control and instrument boundaries, $t_*=(a+b)/2$ is used on a substep $[a,b]$, and $e^{-iH_{\rm log}(t_*)(b-a)}$ approximates Eq.~\eqref{eq:cayley-propagator}. For a Hamiltonian Lipschitz with constant $L$ on that substep, the unitary Duhamel estimate bounds its local operator-norm error by $L(b-a)^2/4$; without this regularity the exact equation remains the time-ordered one. Edge-substep refinements quantify the local numerical approximation at the stated control resolution. The environment is carried through every substep; it is neither reset at a node nor averaged separately on each edge.

The exact interval classicalization identity of Eq.~\eqref{eq:pauli-diagonal} applies to the target-frame construction when a whole interval is represented by one constant-frame conjugation. A Cayley path generally violates that premise during the transition window. To compare it with the ideal jump, factor out the intended edge rotation and write the residual Hamiltonian as $H_{\rm err}(t)$. If $\|H_{\rm err}(t)\|\le g_{\rm p}$ on the union $P$ of transition supports, the Duhamel estimate gives
\begin{equation}
 \|p_{\rm Cayley}-p_{\rm jump}\|_{\rm TV}
 \le 2\int_P\!dt\,\|H_{\rm err}(t)\|
 \le 2g_{\rm p}|P|,
 \label{eq:cayley-pulse-bound}
\end{equation}
which is the fixed-refresh specialization of Eq.~\eqref{eq:app-pulse-error}. For $N_{\rm p}$ nonoverlapping transitions of width at most $\delta$, $|P|\le N_{\rm p}\delta$. Thus the bounded path is physically meaningful only when the cumulative residual action is within the allowed output error, in addition to the single-edge condition $g_{\rm p}\delta\ll1$. The bath correlation-time requirement $\delta\ll\tau_E$ helps justify the piecewise holding picture but does not replace Eq.~\eqref{eq:cayley-pulse-bound}.

Pathwise transparency is stronger than this comparison. When the instrument is covariant throughout the transition, the ideal uncoupled logical output is unchanged exactly. Equation~\eqref{eq:cayley-pulse-bound} measures only the difference between a finite transition and the ideal instantaneous-refresh model caused by residual interaction or implementation error.

\subsection{Direct complete-experiment bound for full-interval paths}
The pulse-support comparison can be uninformative when each transition occupies an entire interval: then $|P|$ is of order $T$ even as $h=T/K$ decreases. A separate bound uses independence of the \emph{target draws}, rather than a small total transition duration. It applies to the full-interval path for bounded, centered interactions and the pointwise covariant instrument class of Theorem~\ref{thm:complete-experiment}.

For $j<K-1$, use the full state $z_t=(j,u_t,G_j,G_{j+1})$, where $u_t=(t-b_j)/h\in[0,1)$ is the deterministic position on the selected Cayley word. On the final holding cell use $z_t=(K-1,u_t,G_{K-1})$; no artificial terminal target is drawn. The state spaces $\mathsf Z_t$ are consequently time dependent. Let $\pi_t$ put unit mass on the current $(j,u_t)$ and give the target pair the uniform law on 16 ordered Pauli pairs, or the final target the uniform law on four frames. The shortest-word rule depends on the two targets and the phase only. For a traceless system coupling, averaging $U^\dagger S U$ over $\pi_t$ gives zero at every physical time. Indeed, $Q=G_{j+1}G_j^\dagger$ is independent of uniform $G_j$, and $U=V(Q,u_t)G_j$. The inner average over $G_j$ is the Pauli twirl of $V^\dagger S V$, which vanishes. The last cell has the ordinary Pauli twirl. A system-identity component can instead be incorporated into the common environment Hamiltonian.

Write $P_{s,t}:\mathsf Z_s\to\mathcal P(\mathsf Z_t)$ for the time-inhomogeneous transition kernel: it advances the deterministic phase, retains the current targets until the next node, then shifts the next target into the old-target position and independently draws the subsequent target. At entry to the terminal cell it retains only $G_{K-1}$. Uniform independent target draws imply the consistency relation $\pi_sP_{s,t}=\pi_t$. Conditional on a complete old pair, the next pair retains one known frame and redraws the other, so its total-variation distance from $\pi_t$ is $3/4$. After two cell crossings the new pair has no target draw in common with the old pair and has law $\pi_t$; its distance is zero. Within one nonterminal cell the old pair is fixed and its distance from $\pi_t$ is $15/16$. In the terminal cell the corresponding distance is $3/4$. Consequently an average-mixing envelope valid for every cell is
\begin{equation}
 b_{\rm Cayley}(s,t)\le
 \begin{cases}
  15/16,&s,t\text{ in one nonterminal cell},\\
  3/4,&s,t\text{ in the terminal cell},\\
  3/4,&s,t\text{ in adjacent cells},\\
  0,&\text{at least two nodes lie strictly between }s\text{ and }t.
 \end{cases}
 \label{eq:cayley-two-frame-mixing}
\end{equation}
At cell boundaries the value assigned to a single time point has no effect on the integral. This envelope concerns the hidden pair of targets, not the marginal distribution of the physical frame at one time.

For clarity, the nonstationary operator-kernel statement can be made on its precise spaces. Let $\mathcal X$ be the trace class of the joint system, bath, instrument memory, and complete classical record. At fixed $s<t$, suppress the deterministic phase coordinates and denote the finite frame-label spaces by $\mathsf S_s,\mathsf S_t$. For $f\in L^\infty(\pi_s;\mathcal X)$, define the conditional-state transfer and its constant projection on every $z'\in\mathsf S_t$ with $\pi_t(z')>0$ by
\begin{equation}
 (\mathcal K_{t,s}f)(z')=
 \frac{\sum_{z\in\mathsf S_s}\pi_s(z)P_{s,t}(z,z')f(z)}{\pi_t(z')},
 \qquad
 (\mathsf P_{t\leftarrow s}f)(z')=
 \sum_{z\in\mathsf S_s}\pi_s(z)f(z).
 \label{eq:cayley-kernel-definition}
\end{equation}
These maps take $L^\infty(\pi_s;\mathcal X)$ to $L^1(\pi_t;\mathcal X)$. The consistency relation above and the trace-norm triangle inequality give
\begin{equation}
 \|\mathcal K_{t,s}f-\mathsf P_{t\leftarrow s}f\|_{\pi_t,1}
 \le 2b_{\rm Cayley}(s,t)\|f\|_{\infty,1},\qquad
 b_{\rm Cayley}(s,t)=\int\pi_s(dz)\|P_{s,t}(z,\cdot)-\pi_t\|_{\rm TV}.
 \label{eq:cayley-time-kernel}
\end{equation}
The deterministic phase causes no singularity: both conditional and reference laws have the same phase coordinate at time $t$. Explicitly, multiply the first difference in Eq.~\eqref{eq:cayley-kernel-definition} by $\pi_t(z')$, sum its trace norm over $z'$, exchange the finite sums, and use $\sum_{z'}|P_{s,t}(z,z')-\pi_t(z')|=2\|P_{s,t}(z,\cdot)-\pi_t\|_{\rm TV}$. This proves Eq.~\eqref{eq:cayley-time-kernel} without a stationary law. Pointwise centering, autonomous target selection and control-independent logical instruments then give the same Duhamel complete-record estimate as Eq.~\eqref{eq:general-bound}, with $b_{\rm Cayley}$. The $K-1$ nonterminal same-cell triangles each have area $h^2/2$ and weight $15/16$; the final triangle has weight $3/4$. The $K-1$ adjacent-cell rectangles each have area $h^2$ and weight $3/4$.
For a prescribed coupling envelope $g(t)$, put $a_j=\int_{b_j}^{b_{j+1}}g(t)dt$. Symmetry of $g(t)g(s)$ makes each same-cell triangle one half of $a_j^2$, while an adjacent-cell rectangle gives $a_ja_{j+1}$. The weighted bound is therefore
\begin{equation}
 \epsilon^{\rm Cayley}_{n_{\rm inst}}(T)\le\min\!\left\{1,4\left[
 \frac{15}{32}\sum_{j=0}^{K-2}a_j^2
 +\frac38a_{K-1}^2
 +\frac34\sum_{j=0}^{K-2}a_ja_{j+1}\right]\right\}.
 \label{eq:cayley-weighted-bound}
\end{equation}
The target-frame law and its pointwise centering are unchanged for deterministic $H_{SE}(t)$; a prescribed $H_E(t)$ enters the common ideal propagator. With $g(t)=g_*$ and equal cells, Eq.~\eqref{eq:cayley-weighted-bound} reduces to the static result
\begin{equation}
 \epsilon^{\rm Cayley}_{n_{\rm inst}}(T)
 \le\min\!\left\{1,\frac{39}{8}g_*^2Th-\frac{27}{8}g_*^2h^2\right\}.
 \label{eq:cayley-direct-bound}
\end{equation}
The bound is a sufficient upper bound obtained from the exact terminal-cell mixing distance and triangle inequalities. At $K=1$ it reduces to the single static-frame bound $\min\{1,(3/2)g_*^2T^2\}$. Completing a nontrivial edge in time $h$ requires speed proportional to $1/h$, so the sequence is interpreted as a control-speed refinement. A transition-containing interval also remains distinct from a Pauli-diagonal error-history slot. The Drude coefficient deviations are separate observables from the complete-record distance bounded here.

\subsection{Numerical implementation}
\label{app-cayley-numerical}
The node-matched calculation uses the qubit-environment model of Sec.~\ref{sec:numerical}, with $T=1.2$, $H_E=(1.7/2)\sigma_z^{(E)}$, and $H_{SE}=0.8\sigma_z^{(S)}\otimes\sigma_x^{(E)}$. We use 500 trajectories at each $K$ in the scan $\{1,6,24,100,200,500,768\}$. The paired comparisons use common target frames. Each Cayley edge is subdivided into 12 propagation steps; the reported operational values also have an exact finite-frame ensemble check. The logical experiments are the $H$--idle--$R_y(\pi/4)$ task and Bell preparation with a shared environment qubit, both shown in Fig.~\ref{fig:operational}.

For the operator CMI, the same complete target-frame paths are used in both averaging orders, with 80 temporal blocks and field rank 128. The two-particle calculation enumerates all 225 label pairs with the Drude--Matsubara history tensor. The original uniform-grid Cayley Drude calculation covers $K=1,6,24$; a separate event-aligned full-interval extension reaches $K=100,200,768$ with the system-leg resolution stated below. The operational and CMI scans use the same target-frame paths and are reported alongside the particle-sector calculation.

At $K=768$, the exact frame-ensemble $H$--idle--$R_y$ probabilities are $0.147164785$ for the full-interval Cayley path and $0.147082152$ for instantaneous refresh. The data show a small transition correction consistent with the finite-duration path construction.

\subsection{Drude-bath two-particle amplitudes}
The particle-sector calculation uses a separate thermal Gaussian bath with $\beta=5$, reorganization energy $\Lambda=1$, bath parameter $\gamma_B=1$, and coupling $g=1.5$, with no additional ultraviolet cutoff in the target Drude spectrum. The insertion times are $(t_1,t_2)=(1/3,2/3)$. For every trajectory we contract the complete bath history and all 225 ordered pairs of nonvacuum operation labels. We average the complex amplitudes over 500 target-frame paths before subtracting the uncoupled coefficient. Writing this difference as $\Delta C_{\nu_2\nu_1}$, the reported mean absolute and RMS deviations are
\begin{equation}
 d_1=\frac{1}{225}\sum_{\nu_1,\nu_2}|\Delta C_{\nu_2\nu_1}|,
 \qquad
 d_2=\left(\frac{1}{225}\sum_{\nu_1,\nu_2}|\Delta C_{\nu_2\nu_1}|^2\right)^{1/2}.
 \label{eq:cayley-particle-metrics}
\end{equation}
These are deviations of complex cPT coefficients, not output probabilities or operational distances.

The full-interval values from the original 384-cell bath grid are $d_1=0.188936$, $0.116848$, and $0.055611$ at $K=1,6,24$, respectively. The event-aligned contraction gives the closely corresponding $P=4$ values $0.188931$, $0.116862$, and $0.055636$; these are the low-$K$ points in Fig.~\ref{fig:drude-particles-main}. The two calculations have different temporal quadratures, so their small differences are retained as a convergence check rather than averaged away.

The Drude correlation used 2048 Matsubara terms. The original history tensor used 384 time cells, truncation tolerance $10^{-7}$, and maximum bond rank 64. Paired calculations at 192 cells, tolerance $10^{-6}$, and maximum rank 48 changed $d_1$ by at most about $9.2\times10^{-5}$ over the three full-interval cases. The paired calculation provides a sensitivity check for the reported full-interval values. A uniform 384-cell grid does not resolve each edge at high $K$; the event-aligned extension separates system and bath quadratures and is reported below with its corresponding resolution.

The Cayley construction provides a controlled physical realization of the fixed-refresh schedule and resolves the finite-duration path dependence.

\subsection{Node-matched validation and high-$K$ Drude extension}
The three-qubit Bell calculation uses two independent local Pauli target frames and one shared environment qubit. The laboratory gate Hamiltonians are covariant pointwise during a duration-$0.6$ Hadamard and a duration-$0.6$ CNOT; the final logical readout is the $|\Phi^+\rangle$ projector. At zero system--bath coupling the largest pathwise deviation from unit Bell fidelity was below $10^{-11}$. Figure~\ref{fig:operational} shows all seven matched $K$ values for the full-interval Cayley path and instantaneous refresh. The target frames and shared bath are identical in each paired comparison.

The event-aligned full-interval Drude calculation extends the mean absolute coefficient deviation over 225 labels to $K=100,200,768$. Against matched instantaneous fixed-refresh values, the pairs $(d_1^{\rm fixed},d_1^{\rm Cayley})$ are $(0.018438,0.021103)$, $(0.010857,0.012135)$, and $(0.003477,0.003997)$, respectively. These points appear in Fig.~\ref{fig:drude-particles-main}. The $K=100$ Cayley point uses four system quadrature legs per transition; the $K=200$ and $768$ points use two. The sequence decreases with $K$, and the displayed system-leg resolution is specified for reproducibility.

%% file: appendix_poisson.tex
\section{Poisson refresh: event schedules, bounds and numerical results}
\label{app:poisson}
This appendix gives the event-time extension of the fixed-node protocol. It is included for completeness, while the main text uses fixed refresh and its bounded Cayley realization as the scheduled-node comparison.

\subsection{Ordered event schedules}
Let $N_T$ be a Poisson random variable of mean $\lambda_rT$. Conditional on
$N_T=m$, the ordered event times satisfy $0<\tau_1<\cdots<\tau_m<T$ and have
conditional joint density $m!/T^m$ on the ordered simplex. Multiplying by
$\Pr\{N_T=m\}=e^{-\lambda_rT}(\lambda_rT)^m/m!$ gives the joint (Janossy)
density $e^{-\lambda_rT}\lambda_r^m$ used below. If $y_0,\ldots,y_m$ are
independent Pauli characters, the complete Gaussian influence average is
\begin{equation}
 \overline{\mathcal F}^{\rm P}[s,r]=e^{-\lambda_rT}\sum_{m=0}^{\infty}\lambda_r^m\int_{0<\tau_1<\cdots<\tau_m<T}d^m\boldsymbol\tau\,2^{-(m+1)}\sum_{y_0,\ldots,y_m=\pm1}\mathcal F_B[y\lambda_s,y\lambda_r].
 \label{eq:poisson-complete-average}
\end{equation}
The particle insertion times remain fixed. For each sampled schedule the time axis is split at the union of event and insertion times; all bath links across those pieces are retained before the final environment contraction.

\subsection{Conditional classicalization and residual correlations}
Under the factorized, frame-independent initial-state assumption used for
interval classicalization, a fixed event schedule has independent frames on
the realized holding intervals, and the same character identity as
Eq.~\eqref{eq:pauli-diagonal} gives a diagonal error-history mixture at the
event-aligned boundaries. Averaging the schedule does not turn arbitrary
preassigned cuts into independent slots. For a nontrivial character $x(t)$,
\begin{equation}
 \mathbb E[x(t_2)x(t_1)]=\Pr\{N([t_1,t_2])=0\}=e^{-\lambda_r(t_2-t_1)}.
 \label{eq:poisson-character-correlation}
\end{equation}
Thus Poisson refresh is conditionally classical at its random boundaries but retains a finite-time correlation after the event times are averaged.

\subsection{Register representation and physicality}
Before eliminating the environment, introduce a finite classical register $G$ for the current Pauli frame. If $Q_G$ is the complete cMPS generator with $G$ held fixed, the joint register--process generator is
\begin{equation}
 (\mathbb Q_{\rm P})_{GG'}=\delta_{GG'}Q_G+\lambda_r(M^{-1}-\delta_{GG'})\id_{SE},\qquad M=4.
 \label{eq:poisson-register-generator}
\end{equation}
The register jumps are classical and leave the environment state untouched. Every finite schedule is a valid causal dilation, and the event-time average is a convex combination of such processes. Positivity, causality and finite-comb recovery therefore follow from the same trajectory argument as Proposition~\ref{prop:averaged-cpt}. Truncating the event sum at $m_{max}$ introduces the unnormalized tail bound
\begin{equation}
 |\overline{\mathcal C}-\overline{\mathcal C}_{m\le m_{\max}}|\le \Pr\{\operatorname{Pois}(\lambda_rT)>m_{\max}\}.
 \label{eq:poisson-tail}
\end{equation}

\subsection{Operational and particle-amplitude results}
The Poisson Bell experiment used the same two logical qubits, shared environment qubit, $T=1.2$, $g=0.8$, and 500 trajectories as the fixed-node and Cayley scans. The Bell infidelity was $0.39630$, $0.20392$, $0.07677$, $0.02282$, $0.00579$, and $0.00268$ for rates $\lambda_r=0,5,20,80,320,640$, respectively. These values are an event-time average and should not be read as a fixed-node result at $K=\lambda_rT$.

For the Drude two-particle coefficients at $(t_1,t_2)=(T/3,2T/3)$, the strongest Poisson point used rate $\lambda_r=512$ and gave mean absolute deviation $0.00830$ over all 225 ordered operation-label pairs at those times. This is not a numerical integration over the continuous two-particle sector. The fixed-sector bound applies after conditioning on the schedule and then averaging its random maximum holding interval; it does not assert a finite-time exact independence at preassigned cuts.

\subsection{Finite pulse duration}
If each event is implemented by a pulse of width $\delta$, with residual Hamiltonian bounded by $g_{\rm p}$, the ensemble comparison acquires the conservative correction
\begin{equation}
 \epsilon^{\rm P,pulse}(T)\le \min\left\{1,\frac{3g_*^2T}{\lambda_r}+2g_{\rm p}\delta(\lambda_rT+1)+p_{\rm ov}\right\},
 \label{eq:poisson-pulse-bound}
\end{equation}
where $p_{\rm ov}$ accounts for schedules with overlapping nonzero-width pulses. The event-time rate and pulse-width limits therefore cannot be optimized independently.

\subsection{Separate Poisson figures}
The following plots isolate the event-time protocol. Their control-rate axes must not be identified with the prescribed-node count or the continuous-drive strength parameters in the main-text figures. The single-qubit curve in Fig.~\ref{fig:poisson-operational}(a) uses the exact event-and-frame ensemble contraction; the Bell, CMI, and particle-coefficient results use sampled complete Poisson trajectories. The remaining protocols are shown in the main text.

\begin{figure}[t]
\includegraphics[width=\linewidth]{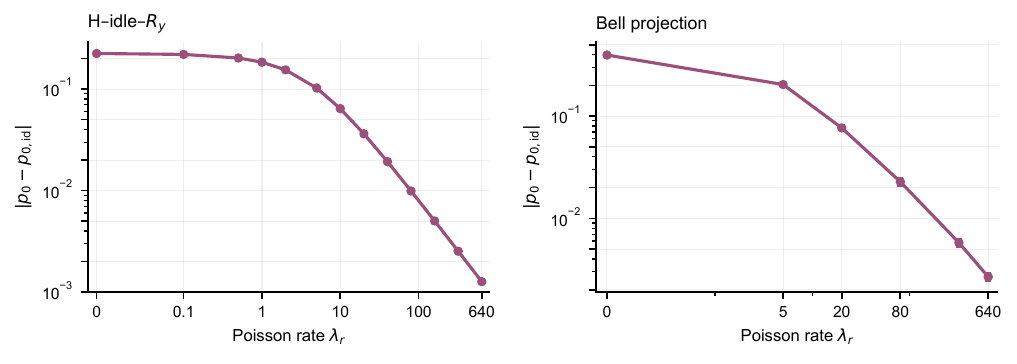}
\caption{Poisson event-time operational recovery for the $H$--idle--$R_y$ task (left, 13 sampled rates) and Bell projection with a shared environment qubit (right, six sampled rates). The ordinate is $|p_0-p_{0,\mathrm{id}}|$. The left curve uses the exact event-and-frame ensemble contraction and has no sampling bars. The Bell points average 500 trajectories; their bars show pointwise 95\% sampling intervals and exclude time-grid error. Zero rate is a static random frame, not the deterministic uncompiled process.}
\label{fig:poisson-operational}
\end{figure}

At $\lambda_r=640$, the 500-trajectory estimate for the first task gives $p_0-p_{0,\mathrm{id}}=3.93\times10^{-4}$ with a 95\% half-width of $8.08\times10^{-4}$. That sampled interval includes zero, so its absolute value does not have a positive lower confidence limit. The independent event-and-frame ensemble contraction gives $|p_0-p_{0,\mathrm{id}}|=1.27\times10^{-3}$; Fig.~\ref{fig:poisson-operational}(a) uses these exact-ensemble values at every rate. The sampled values remain available as a numerical cross-check, without being interpreted as a resolved residual at the strongest rate.

\begin{figure}[t]
\includegraphics[width=\linewidth]{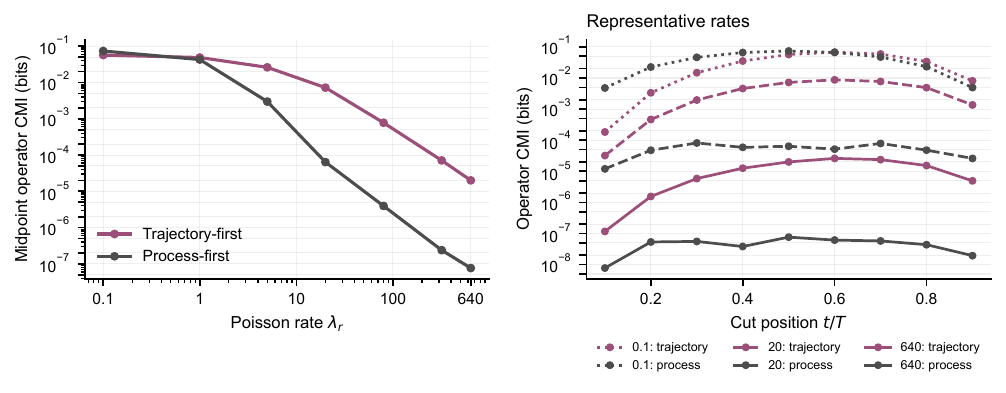}
\caption{Poisson event-time operator CMI. The left panel shows both averaging orders at the midpoint for all seven calculated rates, including four additional 500-trajectory scans. The right panel shows nine cut positions at three representative rates ($0.1$, $20$, and $640$); line style identifies rate and color identifies averaging order. These values use the same 80-block, rank-128 physical opening as the main-text CMI figures.}
\label{fig:poisson-cmi}
\end{figure}

\begin{figure}[t]
\includegraphics[width=0.65\linewidth]{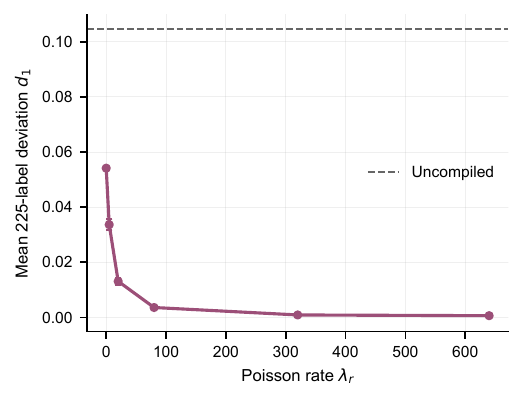}
\caption{Poisson qubit-environment two-particle coefficients at $(t_1,t_2)=(T/3,2T/3)$. Each point is the mean modulus difference between the averaged complex amplitude and its zero-coupling value over all 225 operation-label pairs at those times; the label enumeration does not integrate over the continuous insertion-time domain. The dashed line gives the deterministic uncompiled $U(t)=\id$, $g=0.8$ baseline. Bars are pointwise 95\% bootstrap sampling intervals over 500 paths.}
\label{fig:poisson-qubit-particles}
\end{figure}

\begin{figure}[t]
\includegraphics[width=0.65\linewidth]{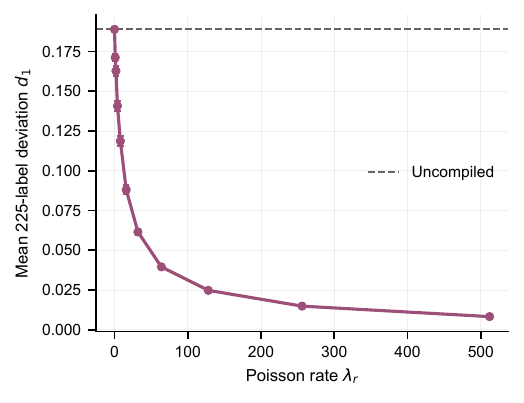}
\caption{Poisson Drude-bath two-particle coefficients at $(t_1,t_2)=(T/3,2T/3)$. The plotted mean absolute deviation uses all 225 complex operation-label coefficients at this time pair, rather than integrating over the continuous insertion-time domain. The dashed line is the deterministic uncompiled $U(t)=\id$, $g=1.5$ baseline. Bars are pointwise 95\% bootstrap sampling intervals from 500 complete event-time trajectories.}
\label{fig:poisson-drude-particles}
\end{figure}

\subsection{Pulse-support counting}
For Poisson refresh, let $m$ be the event count and $p_{\rm ov}$ the probability of a schedule that cannot be implemented by the specified nonoverlapping pulses. With at most $m+1$ pulse supports away from that event, the ensemble-averaged bound is
\begin{equation}
 \epsilon_{n_{\rm inst}}^{\rm Pois,pulse}(T)
 \le\min\!\left\{1,\frac{3g_*^2T}{\lambda_r}
 +2g_{\rm p}\delta(\lambda_rT+1)+p_{\rm ov}\right\}.
 \label{eq:poisson-finite-pulse-bound}
\end{equation}
The $+1$ allows an initial frame-setting pulse; a terminal inverse-frame pulse, if required by the readout, must also be counted. The Poisson expression is an ensemble-averaged sufficient bound. If pulses cannot overlap and events separated by less than $\delta$ are treated as exceptional, $p_{\rm ov}=O(\lambda_r^2T\delta+\lambda_r\delta)$ on a fixed interval, including possible boundary overlaps. The value of this term depends on the specified overlap-handling rule; dropping or delaying such events defines a different ideal Poisson protocol. Appendix~\ref{app:pulse-width} gives the underlying comparison.

\subsection{Rate-dependent bound and pulse optimum}
The complete-experiment mixing bound gives the event-time result
\begin{equation}
 \epsilon^{\rm P}(T)\le\min\{1,3g_*^2T/\lambda_r\},
 \label{eq:poisson-operational-bound}
\end{equation}
because the no-refresh probability is $e^{-\lambda_r(t-s)}$ and its triangular time integral is at most $T/\lambda_r$. This uses the complete trajectory average, not a reset of the environment at each event.

For a prescribed time-dependent coupling envelope $g\in L^1[0,T]$, the same kernel calculation gives
\begin{equation}
 \epsilon^{\rm P}(T)\le\min\!\left\{1,
 3\int_0^Tdt\int_0^t ds\,g(t)g(s)e^{-\lambda_r(t-s)}\right\}.
 \label{eq:poisson-weighted-bound}
\end{equation}
Young's inequality reduces this to $\epsilon^{\rm P}\le\min\{1,3\|g\|_2^2/\lambda_r\}$ when $g\in L^2[0,T]$. The static expression above follows on setting $g(t)=g_*$. The environment propagation may remain time ordered and is never reset at an event.

For an allowed additional pulse error $\eta$, the sufficient ensemble condition is $2g_{\rm p}\delta(\lambda_rT+1)+p_{\rm ov}\le\eta$, where $p_{\rm ov}$ is the probability of an event schedule that cannot be realized with nonoverlapping pulses. A small single-event action $g_{\rm p}\delta$ alone does not control the accumulated error. If the pulse correction must be smaller than the ideal $O(g_*^2T/\lambda_r)$ term, it is sufficient that $g_{\rm p}\lambda_r\delta=o(g_*^2/\lambda_r)$ and $p_{\rm ov}=o(g_*^2T/\lambda_r)$ at fixed $g_*$ and $T$. When the overlap term is controlled by the implementation, the conservative envelope has the formal stationary point $\lambda_r^*=\sqrt{3g_*^2/(2g_{\rm p}\delta)}$; the realized optimum is obtained by including the specified overlap-handling rule.

%% file: appendix_cmi.tex
\section{The two operator-CMI implementations}
\label{app:cmi}

\subsection{Interface, normalization, and the entropy used in the data}
The CMI calculation is a diagnostic on an opened continuous process, not the final POVM probability in Sec.~\ref{sec:numerical}. We use the normalized Pauli basis $p_\mu=\sigma_\mu/\sqrt2$ for a qubit, so $\Tr(p_\mu^\dagger p_\nu)=\delta_{\mu\nu}$. A nonvacuum field insertion is the unscaled left/right Pauli operation multiplied by the corresponding basis normalization, namely
\begin{equation}
 R_{\mu\mu'}[X]=(p_\mu\otimes\id_E)X(p_{\mu'}^\dagger\otimes\id_E),
 \qquad (\mu,\mu')\ne(0,0).
 \label{eq:app-cmi-R}
\end{equation}
This convention is specific to the operator entropy and is not the unnormalized Pauli convention used for the Drude amplitude plots. Mixing these scales changes the operator state and its entropy.

At a cut $t$, let $p$ and $f$ denote orthonormal past and future Fock coordinates. Let $i,j=1,\ldots,d_S^2$ be the two physical system Liouville indices at the opened interface and $e=1,\ldots,d_E^2$ the environment Liouville index. For one complete control trajectory, propagation from the initial boundary to the cut and from the cut to the terminal trace gives arrays
\begin{equation}
 A_{p;ie},\qquad B_{f;je}.
 \label{eq:app-cmi-AB}
\end{equation}
The closed process coefficient is the bilinear contraction
\begin{equation}
 M_{fp}=\sum_{i,j,e}B_{f;je}A_{p;ie}\delta_{ji}=(BA^T)_{fp}.
 \label{eq:app-cmi-M}
\end{equation}
The transpose is ordinary: a process amplitude is bilinear in past and future coefficients. Complex conjugation is introduced only when constructing a Gram matrix or the final density operator. To open the physical system while keeping the environment connection closed, define
\begin{equation}
 E_{ji}=|j\rangle\langle i|_{S,\mathrm L}\otimes\id_{E,\mathrm L},
 \qquad Z^{ij}=BE_{ji}A^T,
 \quad Z^{ij}_{fp}=\sum_eB_{f;je}A_{p;ie}.
 \label{eq:app-cmi-Z}
\end{equation}
The same $e$ is contracted across the cut; $i,j$ remain open. This is an opening of the physical system port, not a separate environment trace and not a bath reset. Summing $i=j$ recovers $M$ and is a direct closure check.

Stacking the $Z^{ij}$ gives a vector $|\zeta\rangle_{PSF}$. We normalize it globally, $|\widehat\zeta\rangle=|\zeta\rangle/\sqrt{\langle\zeta|\zeta\rangle}$, form $\rho_{PSF}=|\widehat\zeta\rangle\langle\widehat\zeta|$, and use
\begin{equation}
 \mathcal N_{\rm op}(t)=S(P S)+S(S F)-S(S)-S(P S F),
 \label{eq:app-cmi-definition}
\end{equation}
with base-two entropy $S(\rho)=-\Tr(\rho\log_2\rho)$. Since the normalized operator vector is pure, $S(PSF)=0$, $S(PS)=S(F)$, and $S(SF)=S(P)$, but the four-term expression fixes the convention and remains the implementation definition.

\subsection{First implementation: physical interface and Gram propagation}
For the qubit model use $H_E=(1.7/2)\sigma_z^{(E)}$, $H_{SE}=0.8\sigma_z^{(S)}\otimes\sigma_x^{(E)}$, initial $|0\rangle_S|+\rangle_E$, and no $X/H$ logical instruments or terminal POVM in this diagnostic. The process has its trace boundary, and the system port is opened only at the chosen cut. For one control trajectory define the Liouville generator
\begin{equation}
 Q_\omega(t)|X\rangle\rangle=|-i[H_\omega(t),X]\rangle\rangle,
 \label{eq:app-cmi-Q}
\end{equation}
and use the two $16\times16$ auxiliary Gram matrices $r$ and $l$ obeying
\begin{align}
 \dot r&=Qr+rQ^\dagger+\sum_\nu R_\nu rR_\nu^\dagger,
 &r(0)&=|\rho_{SE}\rangle\rangle\langle\langle\rho_{SE}|,\\
 -\dot l&=Q^\dagger l+lQ+\sum_\nu R_\nu^\dagger lR_\nu,
 &l(T)&=|\id_{SE}\rangle\rangle\langle\langle\id_{SE}|.
 \label{eq:app-cmi-gram}
\end{align}
The two equations are forward and backward contractions of the same full cMPS. They do not trace the environment at the cut. For a grid with pure-emission blocks, the sum over arbitrary particle number is performed analytically by the linear emission transfer. A convenient representation is $A=(\sqrt r)^T$, $B=\sqrt l$, so $BA^T=\sqrt l\sqrt r$; only the field support used by the numerical compression is finite. The pure-emission auxiliary map can be stabilized by subtracting its common scalar decay, which cancels after the final normalization. This is a computational gauge, not an additional physical depolarizing channel.

For each cut construct $Z^{ij}$ from Eq.~\eqref{eq:app-cmi-Z}, normalize the stacked vector, build the three reduced matrices by ordinary partial traces, and diagonalize them. For example,
\begin{align}
 (\rho_P)_{pp'}&=\sum_{i,j,f}\widehat z_{pijf}\widehat z_{p'ijf}^*,\\
 (\rho_F)_{ff'}&=\sum_{p,i,j}\widehat z_{pijf}\widehat z_{pijf'}^*,\\
 (\rho_S)_{ij,i'j'}&=\sum_{p,f}\widehat z_{pijf}\widehat z_{pi'j'f}^*.
 \label{eq:app-cmi-reductions}
\end{align}
This is how the infinite particle hierarchy is handled: it is encoded in the past and future Gram contractions and never enumerated as a finite particle cutoff. Numerical errors instead come from time blocks, field-rank compression, propagation tolerances, and trajectory sampling.

For an average process, this calculation cannot be performed by separately normalizing each trajectory and then mixing the resulting density matrices if the desired object is $\mathcal N_{\rm op}(\E_\omega[\Upsilon_\omega])$. All trajectories must use one common orthonormal past and future field basis. With $w_\omega=1/N_{\rm tr}$,
\begin{align}
 \overline M_{fp}&=\sum_\omega w_\omega\sum_{i,e}B_{\omega,f;ie}A_{\omega,p;ie},\\
 \overline Z^{ij}_{fp}&=\sum_\omega w_\omega\sum_eB_{\omega,f;je}A_{\omega,p;ie}.
 \label{eq:app-cmi-average}
\end{align}
Equivalently retain a hidden index $h=(\omega,e)$, put $w_\omega$ on the past tensor once, and contract the hidden index only after forming the averaged open vector. The outer product of this one averaged vector contains cross-trajectory terms. Dropping them computes a different classical mixture.

\subsection{Second implementation: central Schmidt coordinates}
The independent coordinate check begins with the same $r,l,A,B$ and closed matrix $M=BA^T$. Perform the full SVD
\begin{equation}
 M=W\Sigma V^\dagger.
 \label{eq:app-cmi-svd}
\end{equation}
Because the closed contraction is bilinear, the transformed arrays are
\begin{equation}
 A_c=V^TA,\qquad B_c=W^\dagger B,\qquad B_cA_c^T=\Sigma,
 \label{eq:app-cmi-central-arrays}
\end{equation}
and the physical openings transform as
\begin{equation}
 Z_c^{ij}=B_cE_{ji}A_c^T=W^\dagger Z^{ij}V.
 \label{eq:app-cmi-central-open}
\end{equation}
Thus $W$ and $V$ change only orthonormal field coordinates; carrying the same physical $E_{ji}$ proves equality of the two implementations up to numerical precision. In the 500-trajectory calculation, the largest single-trajectory discrepancy over nine cuts was $1.03\times10^{-13}$ bit. The largest averaged-process discrepancy was $5.51\times10^{-14}$ bit.

The center construction of Ref.~\cite{wassnerCPT} distributes square roots of the Schmidt coefficients to the two sides. If $r,l$ are full rank, define the balancing map $L_{\rm bal}=\sqrt r\,V\Sigma^{-1/2}$. Then
\begin{equation}
 L_{\rm bal}^{-1}rL_{\rm bal}^{-\dagger}=\Sigma,\qquad L_{\rm bal}^\dagger lL_{\rm bal}=\Sigma,\qquad
 Z_c^{ij}=\sqrt\Sigma\,(L_{\rm bal}^{-1}E_{ji}L_{\rm bal})\sqrt\Sigma.
 \label{eq:app-cmi-gauge}
\end{equation}
The map $L_{\rm bal}$ is generally not unitary and is distinct from a random frame. Its induced metric is $H_{\rm bal}=L_{\rm bal}^\dagger L_{\rm bal}$ and the physical adjoint is $X^\sharp=H_{\rm bal}^{-1}X^\dagger H_{\rm bal}$. Therefore one cannot reshape a Schmidt index into a $4\times4$ pair and call the two positions system and environment. The physical interface must be transported as $L_{\rm bal}^{-1}E_{ji}L_{\rm bal}$, or the calculation must remain in the original physical coordinates. Rank-deficient sectors are handled with the full SVD and no unsupported inverse of $\Sigma$.

This gives a precise coordinate test for the two implementations. If $V$ and $W$ are full unitary field-basis changes, then $Z_c^{ij}=W^\dagger Z^{ij}V$ preserves the norm of the stacked opened vector. Its reduced states are related to the original ones by local unitaries on $P$ and $F$, so every entropy in Eq.~\eqref{eq:app-cmi-definition} is unchanged. Under a nonunitary central balancing, the same conclusion requires transporting both the matrix units $E_{ji}$ and the inner product induced by $H_{\rm bal}$; using an ordinary Euclidean adjoint after transporting only the arrays is not this calculation. Closure $\sum_iZ^{ii}=M$ checks the trace boundary, but alone does not identify the physical port because other openings can have the same closure.

There is a stronger algebraic check when the cut has full rank. The transported physical system algebra is $\mathcal A_S^c=L_{\rm bal}^{-1}[M_{d_S^2}(\mathbb C)\otimes\id_{E,\mathrm L}]L_{\rm bal}$, and its commutant is the similarly transported environment algebra. A unitary relabeling of central coordinates is available when this algebra is closed under the ordinary adjoint. In the uncompiled midpoint audit, the ordinary-adjoint closure residual is $1.84\times10^{-16}$ in physical coordinates and $0.436$ in the selected balanced coordinates. The full port-transport calculation keeps the induced metric and the transported matrix units, and Algorithms~\ref{alg:cmi-physical} and \ref{alg:cmi-central} then agree to numerical precision.

\subsection{Algorithm I: CMI from the physical system interface}
\label{alg:cmi-physical}
This procedure starts from the physical $S\otimes E$ auxiliary space, whose qubit Liouville dimensions are $4\times4$. It computes both averaging orders in Eq.~\eqref{eq:cmi-orders}; those orders are distinct outputs of one implementation, not two ways of constructing the interface.
\begin{enumerate}
\item \textbf{Input.} Fix $\rho_{SE}$, $H_E$, $H_{SE}$, $T$, the chosen cuts, a control protocol and strength, trajectory weights $w_\omega$ summing to one, the normalized field basis $p_\mu=\sigma_\mu/\sqrt2$, a time-block discretization, and a field-support tolerance. Use the free cPT: do not insert the finite $X/H$ test instruments or the final measured POVM. The boundary at $T$ is the full joint trace.
\item \textbf{Complete paths and emissions.} Draw each control as in Appendix~\ref{alg:control-paths}: fixed holding frames, Poisson frames and actual event times, Brownian frame paths, or joint OU frame--velocity paths. For each path construct its joint Liouville generator $Q_\omega(t)$ and the 15 nonvacuum emission maps $R_{\mu\mu'}$ in Eq.~\eqref{eq:app-cmi-R}. All trajectories at a cut use the same time blocks and orthonormal field basis.
\item \textbf{Past and future.} Propagate the past from $|\rho_{SE}\rangle\rangle$ and the future backward from $|\id_{SE}\rangle\rangle$ using Eq.~\eqref{eq:app-cmi-gram}; its emission transfer sums every particle number. For an individual path factor the positive Gram matrices as $A=(\sqrt r)^T$, $B=\sqrt l$. For the average, retain past and future coefficient banks in a \emph{common} orthonormal field basis before summing paths; separately chosen square-root gauges cannot be added as common coordinates. Compress only the field supports, retaining physical $i,j,e$ and the trajectory index.
\item \textbf{Physical opening.} Form $M_\omega=B_\omega A_\omega^T$ and all $Z_\omega^{ij}=B_\omega E_{ji}A_\omega^T$, and verify $\sum_iZ_\omega^{ii}=M_\omega$. The same environment index $e$ is contracted through the cut; system indices $i,j$ remain open. The transpose is ordinary, while conjugation enters in norms and density matrices.
\item \textbf{Trajectory-first result.} Stack $Z_\omega^{ij}$ into $z_{\omega,pijf}$, divide the whole vector by its Euclidean norm, form $\rho_P,\rho_F,\rho_S$ by Eq.~\eqref{eq:app-cmi-reductions}, and calculate $S(P)+S(F)-S(S)$ in bits. Average these scalars with $w_\omega$ to obtain $\mathcal N_{\rm traj}(t)$.
\item \textbf{Process-first result.} In the common field basis form $\overline Z^{ij}=\sum_\omega w_\omega Z_\omega^{ij}$ \emph{before} normalization. Equivalently place $w_\omega$ once on the past coefficient and retain a hidden $(\omega,e)$ index until the contraction in Eq.~\eqref{eq:app-cmi-average}. Normalize the single stacked $\overline Z$, form its reductions, and calculate $\mathcal N_{\rm avg}(t)$. This retains cross-trajectory terms; averaging normalized trajectory density matrices or Schmidt spectra does not.
\item \textbf{Checks.} Verify Gram positivity, open-to-closed reconstruction, normalization, and nonnegative CMI within numerical tolerance. Refine time blocks, field rank, propagation tolerance, and trajectory count separately, and compare identical cuts with Algorithm~\ref{alg:cmi-central}. Field compression sums the emission hierarchy; it is not a particle-number cutoff.
\end{enumerate}

\subsection{Algorithm II: central-Schmidt coordinate check}
\label{alg:cmi-central}
This calculation repeats both averaging orders after changing only the orthonormal coordinates of the past and future fields. It uses the same physical opening $E_{ji}$ as Algorithm~\ref{alg:cmi-physical}, so agreement tests the coordinate invariance of the operator-CMI construction.
\begin{enumerate}
\item \textbf{Input and closed matrix.} Reuse $A_\omega,B_\omega$ and the common-basis averaged arrays from Algorithm~\ref{alg:cmi-physical}, without intermediate normalization. Form $M=BA^T$ for one trajectory, or $\overline M=\sum_\omega w_\omega B_\omega A_\omega^T$ for the average. Future field coordinates index rows; past field coordinates index columns.
\item \textbf{Full central decomposition.} Compute $M=W\Sigma V^\dagger$ with full orthonormal bases, including the null space of the closed matrix. Put $A_c=V^TA$ and $B_c=W^\dagger B$. For the average use the hidden index $(\omega,e)$ and include $w_\omega$ on only one side before transforming. Verify $B_cA_c^T=\Sigma$ numerically.
\item \textbf{Transport the physical port.} Compute all $Z_c^{ij}=B_cE_{ji}A_c^T=W^\dagger Z^{ij}V$, then stack them. Do not reshape a Schmidt number into a system/environment pair. A closed-matrix null-space vector can contribute to an opened $Z^{ij}$; do not truncate it solely because its closed singular value is zero. If the full-rank balancing map $L_{\rm bal}$ is used, insert $L_{\rm bal}^{-1}E_{ji}L_{\rm bal}$ and its induced metric as in Eq.~\eqref{eq:app-cmi-gauge}, rather than an untransformed matrix unit.
\item \textbf{Entropy and comparison.} Normalize the entire central opened vector only after opening, take physical $P,S,F$ reductions, and evaluate the same base-two CMI. Repeat per trajectory and for the arithmetic process average. Compare each result to Algorithm~\ref{alg:cmi-physical}, the closed-matrix reconstruction, the opened norm, and weight outside the closed Schmidt support. A discrepancy signals a coordinate, transpose, port-transport, or truncation error.
\end{enumerate}

\subsection{Averaging orders and interface scope}
For each protocol and control strength we use 500 complete trajectories and nine cuts. The trajectory-first statistic is
\begin{equation}
 \mathcal N_{\rm traj}(t)=\frac1{500}\sum_{\omega=1}^{500}
 \mathcal N_{\rm op}(\Upsilon_\omega;t),
 \label{eq:app-cmi-traj}
\end{equation}
where each trajectory is normalized before its entropy is evaluated. The process-first statistic is
\begin{equation}
 \mathcal N_{\rm avg}(t)=\mathcal N_{\rm op}\!\left(\frac1{500}\sum_{\omega=1}^{500}\Upsilon_\omega;t\right),
 \label{eq:app-cmi-avg}
\end{equation}
where the common field basis and cross-trajectory terms in Eq.~\eqref{eq:app-cmi-average} are retained. The first is a mean of trajectory-wise operator diagnostics; the second is the same nonlinear diagnostic of the averaged physical process. The pure operator state used to compute either CMI is not the generally mixed physical Choi operator, and its cross-trajectory terms do not describe coherence between experimentally distinct control histories. Ref.~\cite{wassnerCPT} associates this operator-CMI choice with compressibility rather than a direct intervention-based interpretation. Both statistics can respond to classical and quantum conditional correlations.

The two implementations are numerically equivalent for the physical system--environment opening above. Together with the cMPS construction of Ref.~\cite{wassnerCPT}, this supplies the operator-CMI used in the numerical section and verifies that its value is independent of the chosen orthonormal field coordinates.

%% file: appendix_numerics.tex
\section{Numerical resolution and independent checks}
\label{app:numerics}

\subsection{Finite-duration probability experiments}
The two probability experiments use $T=1.2$, $\omega_E=1.7$, $g=0.8$ per coupled system qubit, and initial environment state $|+\rangle_E$. The single-system-qubit $H$--idle--$R_y$ task starts in $|0\rangle_S$; its fixed logical $Z$ readout has no-compilation probability $0.457606$ and uncoupled ideal $(2-\sqrt2)/4$. The Bell task starts in $|00\rangle_{S_1S_2}$, couples both system qubits to the same environment, and uses the two-outcome Bell POVM specified in Eq.~\eqref{eq:app-bell-generators} and the following text. Its no-compilation $p_0$ is $0.528444$, while the uncoupled ideal is exactly one. A zero-rate refresh still has a randomly drawn initial frame and therefore differs from the uncompiled fixed identity frame in both tasks.

Every scan point has 500 complete paths. Fixed and Poisson refresh trajectories are split at their actual events and instrument boundaries; within each piece the joint Hamiltonian is constant and propagated by a matrix exponential. The Bell calculation uses a pure eight-dimensional joint state and its equivalent joint-Hamiltonian form of Eq.~\eqref{eq:app-inst-overlap}; it does not assemble the full $64\times64$ Liouville transfer at every step. This propagates the same complete trajectory and preserves the shared environment across both gates. At zero coupling the Bell $p_0$ differed from one by at most $8.3\times10^{-15}$ across the checked trajectories and protocols. The fixed and Poisson Bell data use exact event times but do not have a separately computed exact ensemble average.

Brownian and OU paths require a control substep. In the single-system-qubit task, the strongest reported Brownian and time-accelerated OU points used 10242 and 15360 substeps, with paired refinements at 20484 and 30720 steps. The Bell scan used 2400 substeps for intermediate smooth-control points. Its strongest points were evaluated on 4800 steps and compared with 2400-step propagation using the same initial frames and underlying fine-grid noise increments. At Brownian strength $\kappa=128$, the paired $p_0$ means were $0.981566$ and $0.981784$ (coarse and fine); the mean fine-minus-coarse difference was $2.19\times10^{-4}$ with paired standard error $3.92\times10^{-4}$. At OU speed $r=64$, the corresponding means were $0.986856$ and $0.987146$, with difference $2.90\times10^{-4}$ and paired standard error $2.36\times10^{-4}$. These paired grid comparisons, together with the reported trajectory-sampling intervals, specify the numerical resolution of the displayed control points. The OU probability scans change $\tau_r=0.2/r$ and $\sigma_r^2=10r^2$ together. This is a different parameter scan from the particle experiment, where $\tau=0.1$ is fixed. Exact finite-state refresh-ensemble calculations provide an independent check for the fixed and Poisson data in that task.

\subsection{Two-particle amplitudes}
\paragraph{Finite qubit environment.}
The direct calculation uses $T=1.2$, $(t_1,t_2)=(0.4,0.8)$, $g=0.8$, $H_E=(1.7/2)\sigma_z^{(E)}$, and $\rho_{SE}(0)=|0\rangle\langle0|_S\otimes|+\rangle\langle+|_E$. In the logical frame $H_{SE,\omega}(t)=g[U_\omega^\dagger(t)\sigma_z U_\omega(t)]\otimes\sigma_x^{(E)}$. For each control segment, the complete four-dimensional propagator is applied to the $SE$ state. At $t_1$, each of the 15 operators $\mathcal P_{\nu_1}$ acts on the current state; the resulting 15 matrices are jointly propagated to $t_2$, where each of the 15 possible $\mathcal P_{\nu_2}$ is inserted. A final joint trace produces the $15\times15$ coefficient table. The environmental state is carried in the joint operator at every step and is never traced or replaced between insertions. The propagator after $t_2$ can be omitted only because this calculation has a trace boundary and no subsequent instrument or POVM. It must be retained for a measured terminal boundary.

Fixed-refresh propagators are split at every prescribed frame change and both insertion times. Poisson trajectories are split at their actual event times, not at a fixed approximation grid. The corresponding Pauli-conjugated coupling takes values $\pm g\sigma_z\otimes\sigma_x^{(E)}$; both signs are equally likely after an independent frame draw. Finite-dimensional Cayley edges are subdivided and propagated with the frame evaluated at each substep midpoint; this frozen-Hamiltonian exponential approximates the exact time-ordered edge propagator in Eq.~\eqref{eq:cayley-propagator}. For Brownian and OU driving, a midpoint control direction is likewise held on each numerical step and the resulting frozen Hamiltonian is exponentiated exactly. The qubit OU scan uses $\tau_r=0.2/r$ and $\sigma_r^2=10r^2$, as in the finite-duration qubit experiment. This is different from the fixed-$\tau$ Drude OU scan. All qubit amplitude results use 500 complete trajectories per control point. No normalization factor is inserted into the Pauli operation labels.
The uncompiled qubit baseline uses no random draw: $U(t)=\id$, and the two exact propagators passed to the same 225-label contraction are $e^{-iHt_1}$ and $e^{-iH(t_2-t_1)}$, with $H=H_E+0.8\sigma_z^S\otimes\sigma_x^E$. The $g=0$ table and all label conventions are identical to the scan. This gives a mean absolute deviation of $0.1044878974$ and an RMS deviation of $0.1959148075$; the largest individual modulus is $0.3673402641$. A single random initial Pauli sign is a different experiment and gives mean absolute deviation $0.0541846931$ at fixed $K=1$.

Independent checks cover every ordered label pair: direct matrix-chain evaluation and the vectorized table agree to machine precision, with a maximum checked difference of zero. The analytic frozen-step propagator and a separate matrix exponential differ by at most $2.0\times10^{-17}$. At $g=0$, all four controls reproduce the same ideal table pathwise to within $1.5\times10^{-15}$. Trace-preserving diagonal Pauli insertions give unit amplitude, and adjoint-paired labels obey complex conjugacy; these invariants were checked for all sampled trajectories. The strong-end paired comparison used the same initial frames and fine-grid Gaussian increments on 9600 and 19200 steps. Its root-mean-square difference between the two \emph{ensemble-mean complex tables} was $8.24\times10^{-4}$ for Brownian strength $\kappa=128$ and $2.26\times10^{-4}$ for OU speed $r=64$. A separate 19200/38400 comparison gave $6.53\times10^{-4}$ and $1.74\times10^{-4}$, respectively. These paired differences quantify time-discretization sensitivity separately from the bootstrap standard errors in Fig.~\ref{fig:qubit-particles}; the finer-grid comparisons provide the numerical resolution for the displayed Brownian and OU endpoints. The displayed endpoint results use the refined 9600-step calculation, with the paired refinements reported above.

\paragraph{Finite-temperature Drude bath.}
The low-temperature calculation fixes $g=1.5$, $\Lambda=\gamma_B=1$, $\beta=5$, $T=1$, and $(t_1,t_2)=(T/3,2T/3)$. Each particle leg has 15 nonvacuum labels, so the complete fixed two-particle table contains $15^2=225$ complex amplitudes. For a protocol point define the mean complex deviation first, and only then take its magnitude:
\begin{align}
 \Delta\overline C_{\nu_2\nu_1}
 &=\frac1{500}\sum_{\omega=1}^{500}C_{\nu_2\nu_1,\omega}
 -C_{\nu_2\nu_1,\mathrm{id}},\\
 M_{\rm abs}&=\frac1{225}\sum_{\nu_2,\nu_1}
 |\Delta\overline C_{\nu_2\nu_1}|,\\
 M_{\rm rms}&=\left[\frac1{225}\sum_{\nu_2,\nu_1}
 |\Delta\overline C_{\nu_2\nu_1}|^2\right]^{1/2}.
 \label{eq:app-numeric-metrics}
\end{align}
Bootstrap resampling recomputes the mean complex matrix before either nonlinear summary. Averaging the magnitude of each trajectory's deviation first would be a different statistic. Rows of the displayed heat map correspond to the later particle label, columns to the earlier one.
For the deterministic Drude comparator we fix every coupling sign to $+1$ in the exact event-time Matsubara integrals, with event boundaries only at $0,t_1,t_2$. We then use the same 16 system histories and 225 label contractions as in the refresh scan. With 2048 positive Matsubara terms and the same analytic fast-mode area correction, this yields $M_{\rm abs}=0.1889313618$, $M_{\rm rms}=0.3227692943$, and maximum label modulus $0.7065300547$. Every computed fixed-$K=1$ static-Pauli table agrees with this deterministic table to numerical precision: the centered Gaussian influence exponent is quadratic in the history sign, so a global sign reversal leaves it unchanged. This equality is model-specific and does not identify a randomly initialized frame with no compilation in other environments.

Fixed and Poisson refresh calculations resolve their actual holding intervals and event times. Brownian and OU calculations freeze the moving coupling on a bath grid after generating control paths on a finer mesh. All four retain the bath links between every pair of cells before QR and SVD compression of the four-valued double-eigenvalue history. No bath-memory cutoff or coupling-order truncation is imposed. The integrated ideal Drude kernel used 2048 positive Matsubara terms and an analytic fast-mode area correction. Representative strong-refresh points were recomputed with 8192 terms. The paired bath-grid stability target for the averaged complex table was
\begin{equation}
 \left[\frac1{225}\sum_{\nu_2,\nu_1}
 |\overline C_{N'}-\overline C_N|^2\right]^{1/2}<10^{-3},
 \qquad \max_{\nu_2,\nu_1}|\overline C_{N'}-\overline C_N|<3\times10^{-3}.
 \label{eq:app-grid-criterion}
\end{equation}
The representative fixed-refresh heat map uses $K=384$ physical holding intervals, whereas the strongest fixed-refresh scan point uses $K=768$; the strongest Brownian point used 768 bath cells. Control-substep, Matsubara, history-rank, and paired-grid changes were checked separately on the sampled trajectories. The history tensor is weighted for the corresponding control paths, so its resolution is assessed for each protocol and grid. Small-grid results match direct history enumeration to approximately $9\times10^{-15}$ in the checked cases, and a truncated single bosonic-mode calculation independently checks the influence-functional convention.

The independent single-mode check used a thermal mode of frequency $1.3$, inverse temperature $1$, coupling amplitude $0.43$, and total time $0.9$, with a fixed refresh at $0.4$ and particle insertions at $0.2$ and $0.63$. Direct joint evolution with 24 Fock levels was compared with the analytic Gaussian influence kernel on the same complete sign histories. The selected two-particle off-diagonal labels, their frame average, and the closed vacuum result agreed to within $4\times10^{-14}$ in the checked matrix elements. This finite-mode comparison checks branch ordering, insertion placement, and the single final bath trace. For deterministic Poisson-event truncation, Eq.~\eqref{eq:app-poisson-tail} bounds the omitted unnormalized event tail separately.

\subsection{Operator CMI}
The main CMI data use 80 time blocks, a maximum common field rank of 128 for the averaged process, 500 trajectories per point, and nine internal cuts. The Brownian midpoint scan uses $\kappa=0.2,0.6,2,6,16,64,128,256$, and the OU scan uses $r=0.05,0.2,0.6,2,8,32,64,128$. The trajectory-first Gram contraction itself is finite-dimensional and does not use that field-rank cutoff; the average-process contraction does. At the uncompiled midpoint the three pure-operator-state entropies are approximately $S(P)=2.00384895$, $S(F)=2.00165955$, and $S(S)=3.94751497$ bits, yielding $\mathcal N_{\rm op}(T/2)\simeq0.05799353$ bit. The uncompiled midpoint changes from $0.05799183$ to $0.05799390$ to $0.05799399$ bit on 40, 160, and 320 blocks, respectively. The uncoupled ideal is numerically zero at all nine cuts.

Increasing the common-field rank from 128 to 192 changes a representative intermediate Brownian averaged-process CMI from about $8.79\times10^{-5}$ to $9.02\times10^{-5}$ bit, and the corresponding OU point from $4.76\times10^{-5}$ to $4.94\times10^{-5}$ bit. The stronger endpoints were separately audited at the midpoint. Table~\ref{tab:cmi-resolution} reports absolute changes on doubling the time blocks from 80 to 160 at rank 128 and on increasing the field rank from 128 to 192 at 80 blocks. The $S_X$ column is the largest change among $S(P)$, $S(F)$, and $S(S)$, rather than the change in their sum. Thus a small CMI change does not hide cancellation among the larger entropy components. The corresponding entropy components enter the numerical resolution analysis.

\begin{fixedtable}
\fixedtablecaption{Midpoint operator-CMI resolution audit at the strongest displayed control setting. The columns headed $\Delta_t$ compare 80 and 160 blocks at common field rank 128; $\Delta_r$ compares ranks 128 and 192 at 80 blocks. All values are absolute changes in bits from the same 500 sampled control paths. $\Delta S_{\max}$ is the largest change of the three entropy components. The trajectory-first calculation has no common-field-rank truncation.}{tab:cmi-resolution}
\begin{ruledtabular}
\begin{tabular}{lrrrr}
Protocol & $\Delta_t\mathcal N_{\rm traj}$ & $\Delta_t\mathcal N_{\rm avg}$ & $\Delta_r\mathcal N_{\rm avg}$ & $\Delta_t S_{\max}$\\ \hline
Fixed, $K=768$ & $2.62\times10^{-7}$ & $8.22\times10^{-10}$ & $1.46\times10^{-10}$ & $1.84\times10^{-6}$\\
Poisson, $\lambda_r=640$ & $9.33\times10^{-7}$ & $2.04\times10^{-9}$ & $4.95\times10^{-10}$ & $3.57\times10^{-6}$\\
Brownian, $\kappa=256$ & $4.91\times10^{-6}$ & $8.33\times10^{-10}$ & $2.60\times10^{-8}$ & $6.95\times10^{-6}$\\
OU, $r=128$ & $3.03\times10^{-6}$ & $3.99\times10^{-9}$ & $1.96\times10^{-8}$ & $5.75\times10^{-6}$
\end{tabular}
\end{ruledtabular}
\end{fixedtable}

At fixed $K=768$, for example, the trajectory-first midpoint changes from $5.46253\times10^{-6}$ to $5.72419\times10^{-6}$ bit as the grid is doubled, while the process-first value changes from $1.75316\times10^{-8}$ to $1.83536\times10^{-8}$ bit. At 80 blocks the same process-first result changes to $1.76776\times10^{-8}$ bit at rank 192. The corresponding bootstrap standard errors are $3.26\times10^{-7}$ bit for the trajectory mean and $1.18\times10^{-8}$ bit for the process-first estimator. For the other strong endpoints, the trajectory standard errors are $1.15\times10^{-6}$, $3.72\times10^{-6}$, and $2.12\times10^{-6}$ bit for Poisson, Brownian, and OU, respectively. Their process-first bootstrap standard deviations are $7.16\times10^{-8}$, $1.44\times10^{-7}$, and $8.35\times10^{-8}$ bit. Together, the grid, rank, and bootstrap comparisons characterize the reduction in the specified diagnostic over the sampled controls and the numerical resolution of the reported values.

The error sources are conceptually separate. Time-block changes probe drift/emission splitting and control-path representation; common-field-rank changes probe compression; bootstrap resampling probes the finite control ensemble; diagonalization and eigenvalue thresholding affect the three von Neumann entropies. For normalized $d$-dimensional density operators at trace distance $\delta_{\rm tr}\le1-1/d$, the Audenaert bound $|S(\rho)-S(\sigma)|\le\delta_{\rm tr}\log_2(d-1)+h_2(\delta_{\rm tr})$ \cite{combref067} converts a verified trace-distance bound into an entropy bound. The reported CMI values are accompanied by grid, rank, propagation, and sampling comparisons, which jointly set their numerical resolution.